\documentclass[journal,twocolumn]{IEEEtran}

\usepackage[usenames,dvipsnames]{xcolor}

\newcommand{\rej}{\beta_{\rmr,j}}
\newcommand{\mLD}{\widetilde{\mathrm{LD}}}

\usepackage[mathscr]{eucal}
\usepackage[cmex10]{amsmath}
\usepackage{epsfig,epsf,psfrag}
\usepackage{amssymb,amsmath,amsthm,amsfonts,latexsym}
\usepackage{amsmath,graphicx,bm,xcolor,url,overpic}
\usepackage{fixltx2e}%ordering of single and double column floats
\usepackage{array}%array and tabular environments
\usepackage{verbatim}
\usepackage{bm}
\usepackage{algorithmic}
\usepackage{algorithm}
\usepackage{verbatim}
\usepackage{textcomp}
\usepackage{mathrsfs}
\usepackage{epstopdf}

\newcommand{\openone}{\leavevmode\hbox{\small1\normalsize\kern-.33em1}}

\catcode`~=11 \def\UrlSpecials{\do\~{\kern -.15em\lower .7ex\hbox{~}\kern .04em}} \catcode`~=13 

\allowdisplaybreaks[4]

\newcommand{\nn}{\nonumber}

\newcommand{\calJ}{\mathcal{J}}

\newcommand{\calL}{\mathcal{L}}

\newcommand{\calP}{\mathcal{P}}
\newcommand{\calQ}{\mathcal{Q}}

\newcommand{\calS}{\mathcal{S}}
\newcommand{\calT}{\mathcal{T}}

\newcommand{\calV}{\mathcal{V}}
\newcommand{\calW}{\mathcal{W}}
\newcommand{\calX}{\mathcal{X}}
\newcommand{\calY}{\mathcal{Y}}
\newcommand{\calZ}{\mathcal{Z}}

\newcommand{\ba}{\mathbf{a}}

\newcommand{\bb}{\mathbf{b}}

\newcommand{\be}{\mathbf{e}}

\newcommand{\bP}{\mathbf{P}}

\newcommand{\bQ}{\mathbf{Q}}

\newcommand{\bT}{\mathbf{T}}

\newcommand{\bv}{\mathbf{v}}

\newcommand{\by}{\mathbf{y}}
\newcommand{\bY}{\mathbf{Y}}
\newcommand{\bz}{\mathbf{z}}
\newcommand{\bZ}{\mathbf{Z}}

\newcommand{\rmc}{\mathrm{c}}

\newcommand{\rme}{\mathrm{e}}

\newcommand{\rmG}{\mathrm{G}}

\newcommand{\rmH}{\mathrm{H}}

\newcommand{\rmI}{\mathrm{I}}

\newcommand{\rmP}{\mathrm{P}}

\newcommand{\rmr}{\mathrm{r}}

\newcommand{\rmT}{\mathrm{T}}

\newcommand{\bbN}{\mathbb{N}}

\newcommand{\bbP}{\mathbb{P}}

\newcommand{\bbR}{\mathbb{R}}

\DeclareMathAlphabet{\mathbsf}{OT1}{cmss}{bx}{n}
\DeclareMathAlphabet{\mathssf}{OT1}{cmss}{m}{sl}% slanted sans serif

\DeclareSymbolFont{bsfletters}{OT1}{cmss}{bx}{n}  
\DeclareSymbolFont{ssfletters}{OT1}{cmss}{m}{n}
\DeclareMathSymbol{\bsfGamma}{0}{bsfletters}{'000}
\DeclareMathSymbol{\ssfGamma}{0}{ssfletters}{'000}
\DeclareMathSymbol{\bsfDelta}{0}{bsfletters}{'001}
\DeclareMathSymbol{\ssfDelta}{0}{ssfletters}{'001}
\DeclareMathSymbol{\bsfTheta}{0}{bsfletters}{'002}
\DeclareMathSymbol{\ssfTheta}{0}{ssfletters}{'002}
\DeclareMathSymbol{\bsfLambda}{0}{bsfletters}{'003}
\DeclareMathSymbol{\ssfLambda}{0}{ssfletters}{'003}
\DeclareMathSymbol{\bsfXi}{0}{bsfletters}{'004}
\DeclareMathSymbol{\ssfXi}{0}{ssfletters}{'004}
\DeclareMathSymbol{\bsfPi}{0}{bsfletters}{'005}
\DeclareMathSymbol{\ssfPi}{0}{ssfletters}{'005}
\DeclareMathSymbol{\bsfSigma}{0}{bsfletters}{'006}
\DeclareMathSymbol{\ssfSigma}{0}{ssfletters}{'006}
\DeclareMathSymbol{\bsfUpsilon}{0}{bsfletters}{'007}
\DeclareMathSymbol{\ssfUpsilon}{0}{ssfletters}{'007}
\DeclareMathSymbol{\bsfPhi}{0}{bsfletters}{'010}
\DeclareMathSymbol{\ssfPhi}{0}{ssfletters}{'010}
\DeclareMathSymbol{\bsfPsi}{0}{bsfletters}{'011}
\DeclareMathSymbol{\ssfPsi}{0}{ssfletters}{'011}
\DeclareMathSymbol{\bsfOmega}{0}{bsfletters}{'012}
\DeclareMathSymbol{\ssfOmega}{0}{ssfletters}{'012}

\newcommand{\tilP}{\tilde{P}}

\newcommand{\tilQ}{\tilde{Q}}

\newcommand{\tilT}{\tilde{T}}

\newcommand{\tilx}{\tilde{x}}

\newcommand{\tily}{\tilde{y}}
\newcommand{\tilY}{\tilde{Y}}

\newcommand{\barP}{\bar{P}}

\newcommand{\dotleq}{\stackrel{.}{\leq}}

\newcommand{\dotgeq}{\stackrel{.}{\geq}}

\DeclareMathOperator*{\argmax}{arg\,max}
\DeclareMathOperator*{\argmin}{arg\,min}

\DeclareMathOperator{\supp}{supp}

\usepackage{thmtools, thm-restate}
\newtheorem{theorem}{Theorem} 
\newtheorem{lemma}[theorem]{Lemma}

\newtheorem{corollary}[theorem]{Corollary}

\newtheorem{conjecture}[theorem]{Conjecture}

\graphicspath{{./figures/}}
\usepackage{cite}
\usepackage{amsmath,bbm}
\usepackage{amssymb,amsfonts}
\usepackage{graphicx}
\usepackage{subfigure}
\usepackage{amsthm}
\usepackage{textcomp}
\usepackage{booktabs}
\usepackage{float}
\usepackage{extarrows}
\usepackage{ulem}
\usepackage{xifthen}
\usepackage{pdfpages}
\usepackage{transparent}

\usepackage{enumerate}
\def\BibTeX{{\rm B\kern-.05em{\sc i\kern-.025em b}\kern-.08em
		T\kern-.1667em\lower.7ex\hbox{E}\kern-.125emX}}

\makeatletter 
\newcommand\figcaption{\def\@captype{figure}\caption} 
\newcommand\tabcaption{\def\@captype{table}\caption} 
\makeatother

\usepackage[font=footnotesize,labelfont=bf]{caption}

\title{Distributed Detection with \\ Empirically Observed Statistics 
}

\IEEEoverridecommandlockouts
\author{\IEEEauthorblockN{Haiyun He, {\em Student Member, IEEE} $\quad$ Lin Zhou, {\em Member, IEEE}$\quad$   Vincent Y.~F.~Tan, {\em Senior Member, IEEE} }   
\thanks{This work is   funded by a Singapore National Research Foundation Fellowship (R-263-000-D02-281) and the Research Scholar Budget (RSB) from NUS (C-261-000-207-532 and C-261-000-005-001).} 
\thanks{This paper was   presented in part at the IEEE Information Theory Workshop in Visby, Gotland, Sweden, 2019.}
\thanks{H.~He and V.~Y.~F.~Tan are with the Department of Electrical and Computer Engineering, National University of Singapore (NUS) (Emails:  haiyun.he@u.nus.edu and  vtan@nus.edu.sg).  V.~Y.~F.~Tan is also with the Department of Mathematics, NUS. L.~Zhou is with the  Department of Electrical Engineering and Computer Science, University of Michigan, Ann Arbor (Email:  linzhou@umich.edu).}
\thanks{Copyright (c) 2017 IEEE. Personal use of this material is permitted.  However, permission to use this material for any other purposes must be obtained from the IEEE by sending a request to pubs-permissions@ieee.org.}}

\begin{document}
\maketitle

\begin{abstract}
Consider a distributed detection problem in which the underlying distributions of the observations are unknown;  instead of these distributions, noisy versions of empirically observed statistics are available to the fusion center. These empirically observed statistics, together with source (test) sequences, are transmitted through different channels to the fusion center. The fusion center decides which distribution the source sequence is sampled from based on these data. For the binary case, we derive the optimal type-II error exponent given that the type-I error decays exponentially fast. The type-II error exponent is maximized over the proportions of channels for both source and training sequences. We conclude that as the ratio of the lengths of training to test sequences $\alpha$ tends to infinity, using only one channel is optimal.  By calculating the derived exponents numerically, we conjecture that the same is true when $\alpha$ is finite under certain conditions.  We relate our results to the classical distributed detection problem studied by Tsitsiklis, in which the underlying distributions are known. Finally, our results are extended to the case of $m$-ary distributed detection with a rejection option.
\end{abstract}

\section{Introduction}

 The problem of distributed detection~\cite{tenney1981detection,tsitsiklis1988decentralized} has a plethora of applications, such as in distributed radar and sensor networks; see~\cite{chamberland2007wireless}  and references therein for an overview. In these examples, the observed information at local sensors (processors) needs to be quantized before being sent to a fusion center. The fusion center then performs a specific inference task  such as hypothesis testing. 
 
 In the traditional distributed detection problem as studied in~\cite{tenney1981detection,tsitsiklis1988decentralized,chamberland2007wireless,tay2009bayesian}, the underlying   generating distributions  are available at the fusion center and one is tasked to design a test based on observations as well as the known distributions.
 %In the existing literature (e.g.,~\cite{tenney1981detection,tsitsiklis1988decentralized,chamberland2007wireless,tay2009bayesian}), the problem has been well-studied.
 However, in \emph{practical} applications, the fusion center has no knowledge of the underlying distributions and may only be given quantized   or noisy observations and   labelled training sequences (in place of the generating distributions). This leads to   new challenges in designing optimal   tests.

%In the traditional distributed detection problem, the underlying data generating distributions  are available at the fusion center and one is tasked to design a test based on observations as well as the known distributions. However, in practical applications, the fusion center has no knowledge of the underlying distributions and may only be given compressed, quantized   or noisy observations labelled training sequences. This leads to  new challenges in designing an optimal decision test.

Motivated by these practical issues and inspired by \cite{gutman1989asymptotically,tsitsiklis1988decentralized}, in this paper, we adopt a contemporary statistical learning approach and consider the distributed detection problem as shown, for the binary case, in Figure~\ref{fig:model} in which the distributions of sensor observations are \emph{unknown}. We term this problem as {\em distributed detection with empirically observed statistics}. %Our problem is inspired by \cite{gutman1989asymptotically,tsitsiklis1988decentralized} and motivated by practical applications where the underlying true distributions of data are usually unavailable.
We assume that the sensor observations are transmitted to the fusion center via different channels, which can also be regarded as compressors. Labelled training sequences generated from the different underlying distributions are pre-processed then provided to the fusion center. Our aim is to derive   fundamental performance limits of the classification problem as well as to potentially come up with the same conclusions as Tsitsiklis did in \cite{tsitsiklis1988decentralized}, i.e., to conclude that a small number of distinct channels or local decision rules suffices  to attain the optimal error exponent. 
\subsection{Main Contributions}
In this paper, our main contributions are as follows.

Firstly, for the binary distributed detection problem, we derive the asymptotically optimal type-II error exponent when the type-I error exponent is lower bounded by a positive constant. In the achievability proof, we introduce a generalized version of Gutman's test in \cite{gutman1989asymptotically} and prove that the so-designed test is asymptotically optimal.

Secondly, again restricting ourselves to the binary case, we discuss the optimal proportions of different channels that serve as  pre-processors of the training and source sequences. Let $\alpha$, a constant,  denote the ratio between the length of the training sequence and the length of the source sequence. When  $\alpha\to\infty$, we provide a closed-form expression for the type-II error exponent and prove that using only {\em one identical} channel for both training and source sequences is asymptotically optimal. This mirrors Tsitsiklis' result~\cite{tsitsiklis1988decentralized}. On the other hand, if $\alpha$ is sufficiently small, the type-II error exponent is identically equal to zero. When $\alpha$ does not take extreme values,   by calculating the derived exponent numerically, we conjecture that using one channel for the training sequence and another (possibly the same one) for the  source sequence  is optimal under certain conditions.

Thirdly, we relate our results to the classical distributed detection problem in Tsitsiklis' paper~\cite{tsitsiklis1988decentralized}. When $\alpha\to\infty$, the true distributions can be estimated to arbitrary accuracy and we naturally recover the results in \cite{tsitsiklis1988decentralized} for both the Neyman-Pearson and Bayesian settings.

Finally, we extend our analyses to consider an $m$-ary distributed detection problem with   rejection. We derive the asymptotically optimal type-$j$ rejection exponent for each $j\in[m]$ under the condition that all (undetected) error exponents are lower bounded by a positive constant $\lambda$. In the achievability proof, we introduce a generalized version of Unnikrishnan's test~\cite{unnikrishnan2015asymptotically} by identifying  an appropriate test statistic. %We  prove that the so-designed test is asymptotically optimal.

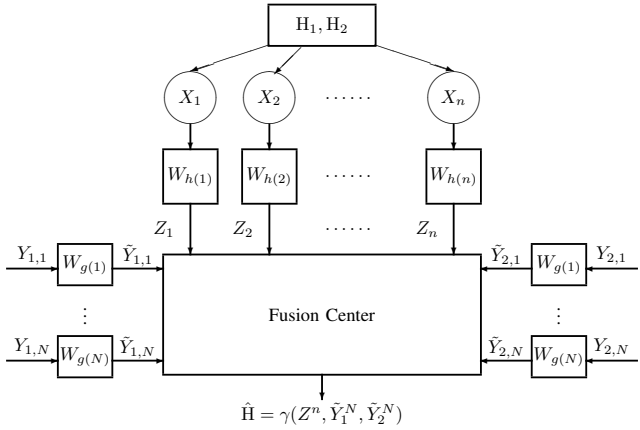
\begin{figure}[tb]
\centering
\setlength{\unitlength}{0.5cm}
\scalebox{0.7}{
\begin{picture}(10,16)(0,-12)
\linethickness{1pt}
\put(3,3){\framebox(4,1.5){$\rmH_1,\rmH_2$}}
\put(3,3){\vector(-3,-1){3}}
\put(4.2,2.9){\vector(-1,-1){1}}
\put(7,3){\vector(3,-1){3}}
\put(0,1){\circle{2}\makebox(0,0){$X_1$}}
\put(0,0){\vector(0,-1){1}}
\put(-1,-3){\framebox(2,2){$W_{h(1)}$}}
\put(0,-3){\vector(0,-1){2}}
\put(-1,-4){\makebox(0,0){$Z_1$}}

\put(3,1){\circle{2}\makebox(0,0){$X_2$}}
\put(3,0){\vector(0,-1){1}}
\put(2,-3){\framebox(2,2){$W_{h(2)}$}}
\put(3,-3){\vector(0,-1){2}}
\put(2,-4){\makebox(0,0){$Z_2$}}

\put(6,1){\makebox(0,0){$\ldots\ldots$}}

\put(10,1){\circle{2}\makebox(0,0){$X_n$}}
\put(10,0){\vector(0,-1){1}}
\put(6,-2){\makebox(0,0){$\ldots\ldots$}}
\put(9,-3){\framebox(2,2){$W_{h(n)}$}}
\put(10,-3){\vector(0,-1){2}}
\put(6,-4){\makebox(0,0){$\ldots\ldots$}}
\put(9,-4){\makebox(0,0){$Z_n$}}

\put(-1,-9.5){\framebox(12,4.5){Fusion Center}}
\put(-7,-5.5){\vector(1,0){2}}
\put(-6,-5){\makebox(0,0){$Y_{1,1}$}}
\put(-5,-6.1){\framebox(2,1.5){$W_{g(1)}$}}
\put(-3,-5.5){\vector(1,0){2}}
\put(-2,-5){\makebox(0,0){$\tilY_{1,1}$}}

\put(-4,-7.1){\makebox(0,0){$\vdots$}}

\put(-7,-9){\vector(1,0){2}}
\put(-6,-8.5){\makebox(0,0){$Y_{1,N}$}}
\put(-5,-9.6){\framebox(2,1.5){$W_{g(N)}$}}
\put(-3,-9){\vector(1,0){2}}
\put(-2,-8.5){\makebox(0,0){$\tilY_{1,N}$}}

\put(13,-5.5){\vector(-1,0){2}}
\put(12,-5){\makebox(0,0){$\tilY_{2,1}$}}
\put(13,-6.1){\framebox(2,1.5){$W_{g(1)}$}}
\put(17,-5.5){\vector(-1,0){2}}
\put(16,-5){\makebox(0,0){$Y_{2,1}$}}

\put(14,-7.1){\makebox(0,0){$\vdots$}}

\put(13,-9){\vector(-1,0){2}}
\put(12,-8.5){\makebox(0,0){$\tilY_{2,N}$}}
\put(13,-9.6){\framebox(2,1.5){$W_{g(N)}$}}
\put(17,-9){\vector(-1,0){2}}
\put(16,-8.5){\makebox(0,0){$Y_{2,N}$}}

\put(5,-9.5){\vector(0,-1){1}}
\put(5,-11){\makebox(0,0){$\hat{\rmH}=\gamma(Z^n,\tilY_1^N,\tilY_2^N)$}}
\end{picture}
}
\caption{System model for distributed detection with empirically observed statistics. Functions $h$ and $g$ represent {\em index mapping functions}. See Fig.~\ref{Fig:ba} for an illustration of $h(\cdot)$. The primary question in this paper is as follows: Given a set of channels  $\{W_j\}_{j\in[K]}$, what are the relative proportions of $W_j$'s that optimize the error exponent?  For binary classification, when is using {\em one} channel optimal (cf.~\cite{tsitsiklis1988decentralized})? See Sec.~\ref{sec:further} for   partial solutions.}
%\caption{System model for distributed detection with empirically observed statistics. Functions $h$ and $g$ represent {\em index mapping functions}. To illustrate this, say that $n=5$ and $K=2$ and the first $2$ source/test samples $X_1, X_2$  are passed through $W_1$ and the remaining   samples $X_3, X_4, X_5$ are passed through $W_2$. Then $h(1)=h(2)=1$ and $h(3)=h(4)=h(5)=2$. The primary question in this paper is as follows: Given a set of channels  $\{W_j\}_{j\in[K]}$, what are the relative proportions of $W_j$'s that optimize the error exponent?  For binary classification, when is using {\em one} channel optimal (cf.~\cite{tsitsiklis1988decentralized})? See Sec.~\ref{sec:further}.}
\label{fig:model}
\end{figure}

\subsection{Related Works}
The distributed detection literature is vast and so it would be futile to review all existing works. This paper, however, is mainly inspired by \cite{gutman1989asymptotically} and~\cite{tsitsiklis1988decentralized}. In \cite{gutman1989asymptotically}, Gutman proposed an asymptotically optimal type-based test for the binary classification problem. In \cite{tsitsiklis1988decentralized}, Tsitsiklis showed that using $\frac{1}{2}{m(m-1)} $ distinct local decision rules is optimal for $m$-ary hypotheses testing in standard Bayesian and Neyman-Pearson  distributed detection settings. Ziv~\cite{ziv1988classification} proposed a discriminant function related to universal data compression in the binary classification problem with empirically observed statistics. Chamberland and Veeravalli~\cite{chamberland2003decentralized} considered the classical distributed detection in a sensor network with a multiple access channel, capacity constraint and additive noise. Liu and Sayeed~\cite{liu2007type} extended the type-based distributed detection to wireless networks. Chen and Wang~\cite{chen2019anonymous} studied the anonymous heterogeneous distributed detection problem and quantified the price of anonymity.  Tay, Tsitsiklis and Win studied tree-based variations of the distributed detection problem in the Bayesian~\cite{tay2009bayesian} and Neyman-Pearson settings~\cite{tay2008data}. The authors also studied Bayesian distributed detection in a tandem sensor network\cite{tay2008subexponential}. The aforementioned works assume that the distributions are known. 

Nguyen, Wainwright and Jordan\cite{nguyen2005nonparametric} proposed a kernel-based algorithm for the nonparametric distributed detection problem with communication constraints. Similarly, Sun and Tay\cite{sun2016privacy} also studied   nonparametric distributed detection networks using kernel methods and in the presence of privacy constraints. While the problem settings in~\cite{nguyen2005nonparametric} and~\cite{sun2016privacy} involve training samples, the questions posed there are algorithmic in nature and hence, different. In particular, they  do not involve fundamental limits in the spirit of this paper. 

\subsection{Paper Outline}
The rest of this paper is organized as follows. In Section~\ref{Sec:main results}, we formulate the distributed detection problem with empirically observed statistics. We also present the optimal type-II error exponent and  analyze the optimal proportion of channels and recover analogues of the results in \cite{tsitsiklis1988decentralized} both for Neyman-Pearson and Bayesian settings. In Section~\ref{sec:mary}, we extend our results to the case in which there are $m\ge 2$ hypotheses and the rejection option is present.  We conclude our discussion and present avenues for future work in Section~\ref{sec:concl}. The proofs of our results are provided in the appendices.

\subsection{Notation}
Random variables and their realizations are in upper (e.g.,  $X$) and lower case (e.g.,  $x$) respectively. All sets are denoted in calligraphic font (e.g.,  $\mathcal{X}$). We use $\calX^{\mathrm{c}}$ to denote the complement of $\calX$. Let $X^n:=(X_1,\ldots,X_n)$ be a random vector of length $n$. All logarithms are base $e$.  Given any two integers $(a,b)\in\bbN^2$, we use $[a:b]$ to denote the set of integers $\{a,a+1,\ldots,b\}$ and use $[a]$ to denote $[1:a]$. The set of all probability distributions on a finite set $\calX$ is denoted as $\calP(\calX)$ and the set of all conditional probability distributions from $\calX$ to $\calY$ is denoted as $\calP(\calY|\calX)$. Given $P\in\calP(\calX)$ and $V\in\calP(\calY|\calX)$, we use $PV$ to denote the marginal distribution on $\calY$ induced by $P$ and $V$. We denote the support of $P$ as $\supp(P)$. Given a vector $x^n = (x_1,x_2,\ldots,x_n) \in\calX^n$, the {\em type} or {\em empirical distribution}~\cite{Csi97} is denoted as $T_{x^n}(a)=\frac{1}{n}\sum_{i=1}^n\mathbbm{1}\{x_i=a\}$ where $a\in\calX$. We interchangeably use $\calT_{T_{x^n}}^n$ and $\calT_{x^n}:=\{\tilx^n\in\calX^n:T_{\tilx^n}(a)=T_{x^n}(a),~\forall\, a\in\calX \}$ to denote the type class of $T_{x^n}$. Let $\calP_n(\calX)$ denote the set of types with denominator $n$.  For two positive sequences $\{a_n\}$ and $\{b_n\}$, we write $a_n\dotleq b_n$ if $\limsup_{n\to\infty}\frac{1}{n}\log \frac{a_n}{b_n}\le 0$. The notations $\dotgeq$ and $\doteq$ are defined similarly. For a given vector $\ba \in\bbR^d$, we let $\mathrm{supp}(\ba):=\{i \in [d] : a_i \ne 0\}$ denote the support   of $\ba$.

\section{Binary Distributed Detection with Training Samples}\label{Sec:main results}

In this section, we formulate the problem in which there are two hypotheses and instead of distributions, only training samples are available.
\subsection{Problem Formulation}
We assume that there are $K$ fixed  {\em compressors} or {\em channels} (these are called {\em local decision rules} in \cite{tsitsiklis1988decentralized}), where for each $j\in[K]$, the $j$-th channel is  $W_j\in\calP(\calZ|\calX)$. This channel has  input alphabet $\calX=[M]$  and output alphabet $\calZ=[L]$.  For  notational simplicity, we assume that $|\calX|=M<\infty$ but our results go through for uncountably infinite $\calX$ as well.  We let $\calW:=\{W_j\}_{j\in[K]}$ be a fixed set of channels. Furthermore, let $h:[n]\mapsto[K]$ and $g:[N]\mapsto[K]$ to be functions that map the index of the test/training sample to the channel index.

The system model is as follows (see Figure~\ref{fig:model}). There are $n$ sensors and a source/test sequence $X^n$ generated i.i.d.\ according to some unknown distribution defined on $\calX$. For each $i\in[n]$, the $i$-th sensor observes $X_i\in\calX$ and maps it to $Z_i$ using the channel $W_{h(i)}$. The $Z_i$'s from all local sensors are transmitted to a fusion center. In addition to $Z_i$'s, the fusion center observes two noisy versions of training sequences $(Y_1^N,Y_2^N)\in\calX^{2N}$ which are  generated i.i.d.\ according to some \emph{unknown} but fixed distributions $(P_1,P_2)\in\calP(\calX)^2$. The fusion center observes noisy sequences $(\tilY_1^N,\tilY_2^N)$, where $\tilY_{1,i}\sim W_{g(i)}(\cdot|Y_{1,i})$ and $\tilY_{2,i}\sim W_{g(i)}(\cdot|Y_{2,i})$  for all $ i\in[N]$.
%\, and $V$ is another channel, not necessarily in $\calW$
 With $(\tilY_1^N,\tilY_2^N)$ and $Z^n$, the fusion center uses a decision rule $\gamma:[L]^{2N+n}\mapsto\rm\{H_1,H_2\}$ to discriminate between the following two hypotheses:  
\begin{itemize}
\item $\rmH_1$: the source sequence $X^n$ and the training sequence $Y_1^N$ are generated according to the same distribution;
\item $\rmH_2$: the source sequence $X^n$ and the training sequence $Y_2^N$ are generated according to the same distribution.
\end{itemize}
%In our setting, we assume that the fusion center only uses one channel $V$ to pre-process the second training sequence $Y_2^N$. The reason for this is that the optimal test we consider in \eqref{Eq:test} to follow depends only on $Z^n$ and $\tilY_1^N$. Nonetheless, $V$ needs to satisfy a technical assumption (see Assumption~\ref{assump: W}). 

%\blue{In our setting, we assume that the fusion center only uses one channel $V$ to pre-process the second training sequence $Y_2^N$. The reason for doing so is simply to ease the analyses and the notations. In fact, this assumption can be easily removed, as demonstrated in the final remark after Theorem \ref{Thm:rejection exponent}.
%}
%
\begin{figure}[t]
\centering \hspace{-1in}
\begin{overpic}[scale=0.71,unit=1mm]{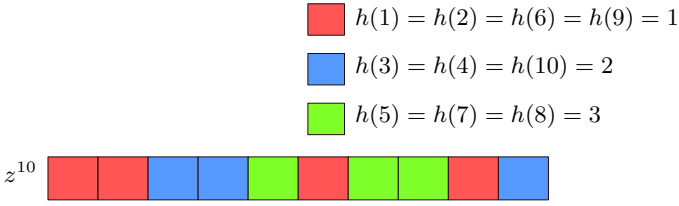}
\put(10,03){\small $z^{10}$}
\put(68,29){\small $h(1) = h(2) = h(6)=h(9)=1$}
\put(68,21){\small $h(3) = h(4) = h(10)=2$}
\put(68,13){\small $h(5) = h(7) = h(8)=3$}
\end{overpic}
\caption{An illustration for $\ba$ and $h$ when $K=3$ and $n=10$. Here we have $a_1=\frac{4}{10}$, $a_2=\frac{3}{10}$, $a_3=\frac{3}{10}$. }
	\label{Fig:ba}
\end{figure}

We assume $N=\lceil \alpha n\rceil$ for some $\alpha\in\bbR_+$.\footnote{We ignore the integer constraints of $(n,N)$ and write $N=\alpha n$.} For each $j\in[K]$, we use $a_j^{(n)}$ and $b_j^{(n)}$ to denote the proportions of $[n]$ and $[N]$ in which the channel $W_j$ is used to process the source and training sequences respectively, i.e.,
\begin{align}
\!\!a_j^{(n)}\! :=\! \frac{\sum_{i\in[n]}\!\mathbbm{1}\{h(i)\! =\! j\}}{n},\quad b_j^{(n)}\! :=\! \frac{\sum_{i\in[\alpha n]}\!\mathbbm{1}\{g(i)\! =\! j\} }{\alpha n}.
\end{align}
An example is given in Figure \ref{Fig:ba}. Furthermore, we let $\ba^{(n)}=(a_1^{(n)},\ldots,a_K^{(n)})$ and $\bb^{(n)}=(b_1^{(n)},\ldots,b_K^{(n)})$. We assume that the following limits  exist:
\begin{align}\label{Eq:a b limits}
	a_j:=\lim_{n\to\infty}a_j^{(n)},\quad b_j:=\lim_{n\to\infty}b_j^{(n)}, \quad\forall j\in[K].
\end{align}
To avoid clutter in subsequent mathematical expressions, we abuse notation subsequently and drop the superscript $(n)$ in $\ba^{(n)}$ and $\bb^{(n)}$ in all non-asymptotic expressions, with the understanding that $\ba$ (resp.\ $a_j$) appearing in a non-asymptotic expression should be interpreted as $\ba^{(n)}$ (resp.\ $a_j^{(n)}$).

Given any decision rule $\gamma$ at the fusion center and any pair of distributions $(P_1,P_2)$ according to which the training sequences $(Y_1^N,Y_2^N)$ are generated, the performance metrics we consider are the type-I and type-II error probabilities 
\begin{align}
 \beta_\nu(\gamma,P_1,P_2|\ba,\bb,\calW)= \bbP_\nu\{\gamma(Z^n,\tilY_1^N,\tilY_2^N)\! \neq\! H_\nu\},
\end{align}
where for $\nu\in [2]$, we use $\mathbb{P}_\nu:=\Pr\{\cdot|\rmH_\nu\}$ to denote the joint distribution of $Z^n$ and $(\tilY_1^N,\tilY_2^n)$ under hypothesis $\rmH_\nu$. In the remainder of this paper, we use $\beta_\nu(\gamma,P_1,P_2)$ to denote $\beta_\nu(\gamma,P_1,P_2|\ba,\bb,\calW)$  if there is no risk of confusion.

Inspired by \cite{gutman1989asymptotically}, in this paper, we are interested in the maximal type-II error exponent with respect to a pair of target distributions for any decision rule at the fusion center whose type-I error probability decays exponentially fast with a certain fixed exponential  rate for all pairs of distributions, i.e., given any $\lambda\in\bbR_+$, the {\em optimal non-asymptotic type-II error exponent} is
\begin{align}
\nn& E^*(n,\alpha,P_1,P_2,\lambda|\ba,\bb,\calW)\\*
&:=\sup\{E\in\bbR_+:\exists~\gamma~\mathrm{s.t.}~\beta_2(\gamma,P_1,P_2)\leq \exp(-nE)  \mbox{~and~} \nn\\
&\qquad\beta_1(\gamma,\tilP_1,\tilP_2)\leq \exp(-n\lambda),~\forall\, (\tilP_1,\tilP_2)\in\calP(\calX)^2\}. \label{eqn:E_star}
\end{align}

\subsection{Definitions}
\label{sec:moti}

 To state our results succinctly, we begin by stating some somewhat non-standard definitions. Given any pair of distributions $(Q,\tilQ)\in\calP([L])^2$ and any $\alpha\in\mathbb{R}_+$, the \emph{generalized Jensen-Shannon divergence}~\cite[Eqn.~(3)]{zhou2018second} is defined as 
\begin{equation}
	\mathrm{GJS}(\tilQ,Q,\alpha):=D\Big(Q\Big\|\frac{Q+\alpha\tilQ}{1+\alpha}\Big)+\alpha D\Big(\tilQ\Big\|\frac{Q+\alpha\tilQ}{1+\alpha}\Big).
\end{equation}

% Let $\mathbf{Q}=(Q_1,\ldots,Q_K)\in\calP([L])^{K}$ and $\tilde{\mathbf{Q}}=(\tilQ_1,\ldots,\tilQ_K)\in\calP([L])^{K}$ be two collections of distributions. Given any $(\bQ,\tilde{\bQ})$, any $(P,\tilP)\in\calP(\calX)^2$, any $\alpha\in\mathbb{R}_+$, any pair $(\ba,\bb)\in\calP([K])^2$, define the following {\em linear combination of divergences}
% \begin{align}
%	\mathrm{LD}(\bQ,\tilde{\bQ},P,\tilP|\alpha,\ba,\bb,\calW):=\sum_{k\in[K]}\big(a_kD(Q_k\|P W_k)+\alpha b_k D(\tilQ_k\|\tilP W_k)\big)\label{def:LD},
% \end{align}
  Let $\mathbf{Q}=(Q_1,\ldots,Q_K)\in\calP([L])^{K}$ and $\tilde{\mathbf{Q}}_i=(\tilQ_{i,1},\ldots,\tilQ_{i,K})\in\calP([L])^{K}$ where $i\in[2]$ be three collections of distributions. Given any $(\bQ,\tilde{\bQ}_1,\tilde{\bQ}_2)\in\calP([L])^{3K}$, any $(P,\tilP_1,\tilP_2)\in\calP(\calX)^3$, any $\alpha\in\mathbb{R}_+$, any pair $(\ba,\bb)\in\calP([K])^2$, define the following {\em linear combination of divergences}
 \begin{align}
 &\mathrm{LD}(\bQ,\tilde{\bQ}_1,\tilde{\bQ}_2,P,\tilP_1,\tilP_2|\alpha,\ba,\bb,\calW) \nn\\*
 &:=\sum_{k\in[K]}\big(a_kD(Q_k\|P W_k)+\sum_{i\in[2]}\alpha b_k D(\tilQ_{i,k}\|\tilP_i W_k)\big)\label{def:LD},
 \end{align}
and furthermore, given any $\lambda\in\bbR_+$, define the following set of collections of distributions
 \begin{align}
&\calQ_{\lambda}(\alpha,\ba,\bb,\calW):=\bigg\{(\bQ,\tilde{\bQ}_1,\tilde{\bQ}_2)\in \calP([L])^{3K}: \nn\\*
& \min_{(\tilP,P)\in\calP(\calX)^2}\mathrm{LD}(\bQ,\tilde{\bQ}_1,\tilde{\bQ}_2,\tilP,\tilP,P|\alpha,\ba,\bb,\calW)\!\leq\!\lambda\bigg\}\label{def:calQlambda}.
 \end{align}
Finally, define the following minimum linear combination of divergences over the collections of distributions in $\calQ_{\lambda}(\alpha,\ba,\bb,\calW)$ as 
 \begin{align}
&f_\alpha(P_1,P_2|\ba,\bb,\calW)\nn\\*
&:=\min_{\substack{(\bQ,\tilde{\bQ}_1,\tilde{\bQ}_2)
		 \\ \in \calQ_{\lambda}(\alpha,\ba,\bb,\calW)}}\mathrm{LD}(\bQ,\tilde{\bQ}_1,\tilde{\bQ}_2,P_2,P_1,P_2|\alpha,\ba,\bb,\calW). \label{eqn:f_alpha}
\end{align}

\subsection{Main Results}\label{Subsec:binary result}
The following theorem is our main result and presents a single-letter expression for the optimal type-II exponent.
\begin{theorem}\label{Thm:exponent}
Given any $(\lambda,\alpha)\in\mathbb{R}_+^2$, any pair of target distributions $(P_1,P_2)\in\calP(\calX)^2$,
\begin{align}
 \lim_{n\to\infty}E^*(n,\alpha,P_1,P_2,\lambda|\ba,\bb,\calW) =f_\alpha(P_1,P_2|\ba,\bb,\calW).\label{Eq:thm1 exponent}
\end{align}
\end{theorem}
The proof of Theorem \ref{Thm:exponent} is given in Appendix \ref{Sec:proof of thm1}. Several remarks are in order.

 Firstly, in the achievability proof of Theorem \ref{Thm:exponent}, we make use of the following test at the fusion center
\begin{align}
\label{Eq:test}
&\gamma(Z^n,\tilY_1^N,\tilY_2^N ) \nn\\*
&= \left\{
\begin{array}{cc}
\rmH_1 & \min_{\tilP,P }\mathrm{LD}\big(\bT_{\bZ^{n\ba}},\bT_{\tilde{\bY}_1^{N\bb}},\bT_{\tilde{\bY}_2^{N\bb}},\tilP,\tilP,P \big)\leq \lambda,\\
\rmH_2 & \text{otherwise},
\end{array}
\right. 
\end{align}
where we suppressed the dependence of $\mathrm{LD}$ on $(\alpha,\ba,\bb,\calW)$ and for each $k\in[K]$, we use $Z^{na_k}$ to denote the collection of $Z_i$ where $i\in[n]$ satisfies $h(i)=k$ and similarly for  $\tilY_j^{Nb_k}$ for $j\in[2]$.  Furthermore, we use $\bT_{\bz^{n\ba}}$ to denote the vector of types $(T_{z^{na_1}},\ldots,T_{z^{na_K}})$ and use $\bT_{\tilde{\by}_j^{N\bb}}$   for $j\in[2]$ similarly. Theorem \ref{Thm:exponent} indicates that the test in \eqref{Eq:test} is asymptotically optimal. The test in \eqref{Eq:test} basically compares a certain distance between $\bT_{\bz^{n\ba}}$ and $\bT_{\tilde{\by}_1^{N\bb}}$ plus a bias term related to $\bT_{\tilde{\by}_2^{N\bb}}$ to a threshold $\lambda$. When the distance is small enough, we declare that $\bT_{\bz^{n\ba}}$ and $\bT_{\tilde{\by}_1^{N\bb}}$ are generated according to the same distribution; otherwise, we declare that they are not.
%\green{I think it would be instructive to provide some intuition for the test in \eqref{Eq:test}.} \red{Haiyun: I have added two lines there.}
% to denote the collection of $\tilY_{1,i}$ where $i\in[N]$ satisfies  $g(i)=k$.}

%Secondly, Theorem \ref{Thm:exponent} shows that the optimal type-II error exponent is independent of the channel $V$ that is used to pre-process the second training sequence $Y_2^N$. This is not unnatural in view of the form of the test in \eqref{Eq:test}. Indeed, this test depends {\em only} on $Z^n$ and $\tilY^N_1$ (and not on $\tilY_2^N$).

 Secondly, the test in \eqref{Eq:test} is a generalization of Gutman's test in \cite{gutman1989asymptotically}. To see this, we note that if we let $K=1$, $M=L$ and consider the deterministic channel denoted as $W=I_{L}$, the test in \eqref{Eq:test} reduces to Gutman's test using $(Z^n,\tilY_1^N,\tilY_2^N)$ since
 \begin{align}
	&\min_{\tilP,P }\mathrm{LD}\big(\bT_{\bZ^{n\ba}},\bT_{\tilde{\bY}_1^{N\bb}},\bT_{\tilde{\bY}_2^{N\bb}},\tilP,\tilP,P \, \big|\,\alpha,\ba,\bb,\{I_L\} \big)\nn\\
	&=\min_{\tilP,P }D(T_{Z^{n}}\|\tilP )+\alpha D(T_{Y_1^{N}}\|\tilP )+\alpha D(T_{Y_2^{N}}\|P )\\
	&=\mathrm{GJS}(T_{\tilY_1^N},T_{Z^n},\alpha),
\end{align}
 and the exponent in Theorem \ref{Thm:exponent} reduces to the type-II exponent for binary classification~\cite[Thm.~3]{gutman1989asymptotically}, i.e., % \green{somehow you should say that the min over LD reduces to GJS. some elaboration is needed here.} \red{Haiyun: I add an elaboration here.}
\begin{equation}
\gamma(Z^n,\tilY_1^N,\tilY_2^N)=\begin{cases}
\rmH_1 & \mathrm{GJS}(T_{\tilY_1^N},T_{Z^n},\alpha)\leq \lambda,\\
\rmH_2 & \text{otherwise},
\end{cases}\label{Eq:Gutman test}
\end{equation}
and
\begin{align}
& \lim_{n\to\infty}E^*(n,\alpha,P_1,P_2,\lambda|\ba,\bb,\calW) \nn\\*
& =\min_{\begin{subarray}{c}
(Q,\tilQ)\in \mathcal{P}\mathcal{(Z)}^{2}:\\
\mathrm{GJS}(\tilQ,Q,\alpha)\leq\lambda	
\end{subarray}} D(Q\|P_2)+\alpha D(\tilQ\|P_1).
\label{Eq:Gutman expo}
\end{align}

 Finally, to better understand the effect of not knowing the true distributions, we numerically plot the optimal type-II exponent $f_\alpha(P_1,P_2|\ba,\bb,\calW)$ (defined in \eqref{eqn:f_alpha}) in Figure \ref{fig:f_alpha_infty}. As shown in Figure \ref{fig:f_alpha_infty}, the optimal type-II exponent $f_\alpha(P_1,P_2|\ba,\bb,\calW)$ increases as $\alpha=\frac{N}{n}$ increases and converges to a threshold as $\alpha\to\infty$. This threshold is the optimal error exponent when the true distributions are known~\cite[Theorem 2]{tsitsiklis1988decentralized}. The gap $f_\infty-f_\alpha$ thus quantifies the loss due to the fact that the generating distributions are unknown and only training samples are available to the learner.
\begin{figure}[t]
	\centering
	\includegraphics[scale=0.45]{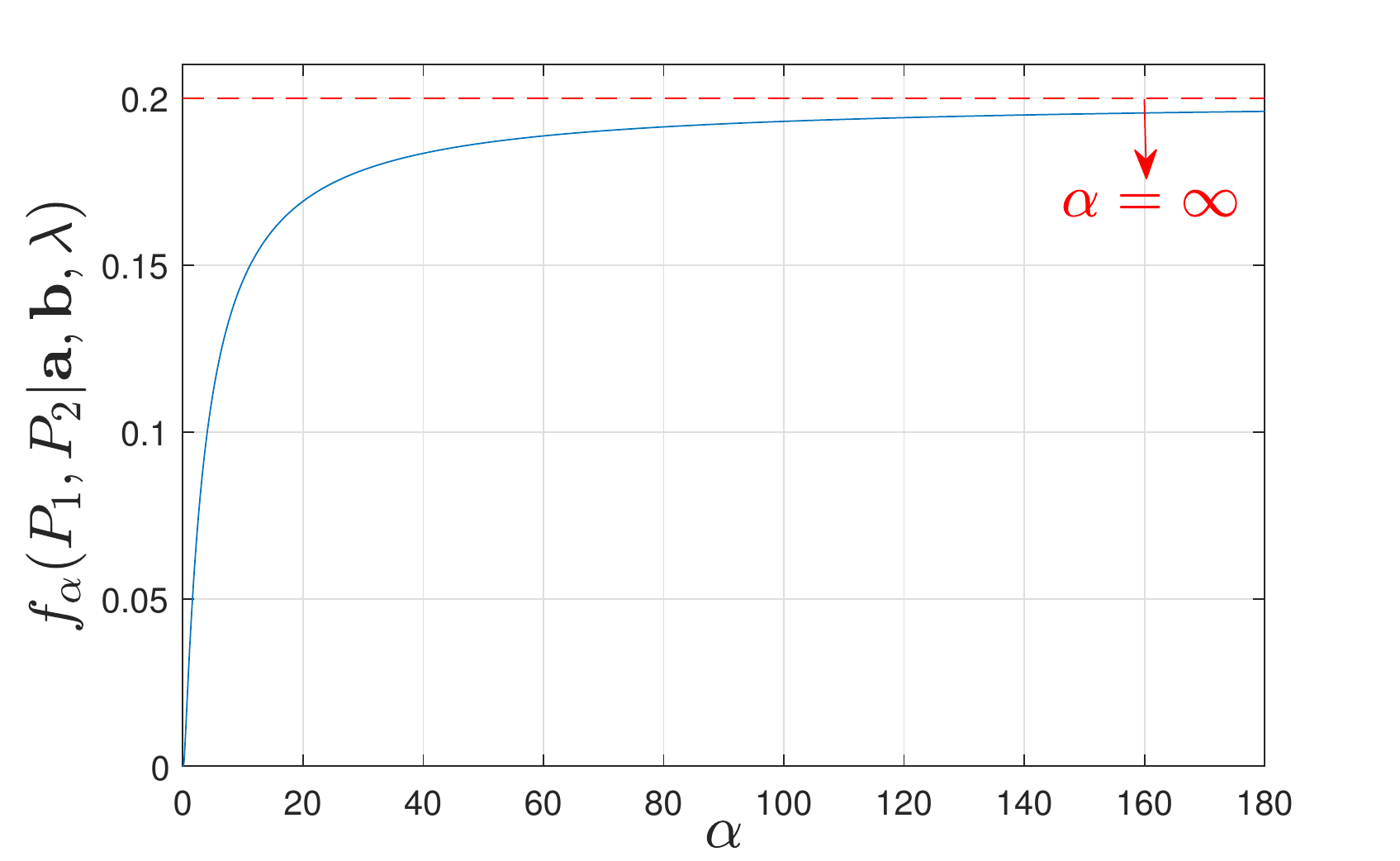}
	\caption{Plot of $f_{\alpha}(P_1,P_2|\ba,\bb,\lambda)$ for various values of $\alpha$, fixed $(\ba,\bb,\lambda)$ and $K=2$}
	\label{fig:f_alpha_infty}
\end{figure}

\subsection{Further Discussions on the Impact of the Proportions of Local Decision Rules $(\ba,\bb)$ on the Exponent}\label{sec:further}
%Given any pair of $(P_1,P_2)\in\calP(\calX)^2$ and any set of $\calW$ satisfying Assumption \ref{assump: W}, 
 In this subsection, we discuss the choices of the proportion of local decision rules, denoted by $(\ba,\bb)$, to achieve the optimal type-II exponent $f_\alpha(P_1,P_2|\ba,\bb,\calW)$. Throughout the section, we fix a pair of target distributions $(P_1,P_2)$.  For brevity, we define
\begin{align}
f_{\alpha}(\ba,\bb,\lambda):=f_\alpha(P_1,P_2|\ba,\bb,\calW)\label{def:f}.
\end{align}
Since the type-II error exponent depends on $(\ba,\bb)$, inspired by the result in \cite{tsitsiklis1988decentralized} which states that one local decision rule is optimal for binary hypotheses testing (in the Neyman-Pearson and Bayesian settings), we can further optimize the type-II error exponent with respect to the design of  the proportion of channels (encoded in $\ba$ and $\bb$) and thus study %the following quantity,
\begin{align}
f^*_{\alpha}(\lambda):=\max_{(\ba,\bb) \in \calP([K])^2}f_{\alpha}(\ba,\bb,\lambda)
\end{align}
and the corresponding optimizers $\ba^*$ and $\bb^*$
%We discuss results concerning $f^*_{\alpha}(\lambda)$ 
for different values of $\alpha$. For this purpose, given any vector $\bv\in\calP([K])$ and any distribution $\tilP\in\calP(\calX)$, define
\begin{align}
\calP(\tilP|\bv,\calW)&:=\big\{
P\in\calP(\calX):  \nn\\* 
&\qquad\forall~k\in[K],v_k\|PW_k-\tilP W_k\|_{\infty}=0\big\}\label{def:calPP1}.
\end{align}
Note that $\tilP\in\calP(\tilP|\bv,\calW)$ and if $\mathrm{supp}(\bv)=[K]$, then $PW_k=\tilP W_k$ for all $k\in[K]$.

Furthermore, given any $\bQ\in\calP([L])^K$, any $P_1\in\calP(\calX)$ and any pair $(\ba,\bb)\in\calP([K])^2$, let
\begin{align}
\!\! \kappa(\bQ,P_1|\ba,\bb,\calW)\!:=\!\min_{\tilP\in\calP(P_1|\bb,\calW)}\!\sum_{k\in[K]}\!a_kD(Q_k\|\tilP W_k)\label{def:kappa}.
\end{align}

\begin{lemma}
	\label{extreme:f}
	The function $f_{\alpha}(\ba,\bb,\lambda)$ satisfies
	\begin{align} \label{eqn:limit_alpha_large}
	&\lim_{\alpha\to\infty}f_{\alpha}(\ba,\bb,\lambda) =f_{\infty}(\ba,\bb,\lambda) \nn\\* 
	&:=\min_{\bQ\in\calP([L])^K:\kappa(\bQ,P_1|\ba,\bb,\calW)\leq \lambda}\sum_{k\in[K]}a_kD(Q_k\|P_2W_k).
	\end{align}
\end{lemma}
The proof of Lemma \ref{extreme:f} is provided in Appendix \ref{proof:extremef}.

We say that $\ba \in \calP([K])$ is {\em deterministic}  if there exists a $j\in [K]$ such that $a_j=1$. Let $\be_j$ be the $j$-th standard basis vector in $\bbR^K$, i.e., the vector $\be_j$ equals $1$ in the $j$-th location and $0$ in other locations.

\begin{corollary}\label{Coro:alpha extreme}
Given any $\lambda\in\bbR_+$, we have % as $\alpha\to\infty$, we have
\begin{align}
\sup_{(\ba,\bb)\in\calP([K])^2}f_{\infty}(\ba,\bb,\lambda)
=\max_{k\in[K]}f_{\infty}(\be_k,\be_k,\lambda),
\end{align}
and thus the maximizers $(\ba^*,\bb^*)$ for $f_{\infty}(\ba,\bb,\lambda)$ satisfy that $(\ba^*,\bb^*)$ are both deterministic and $\ba^*=\bb^*$.
\end{corollary}

The proof of Corollary \ref{Coro:alpha extreme} is provided in Appendix \ref{proof:coro:alpha extreme}. Corollary \ref{Coro:alpha extreme} says that when the length of the training sequence is much longer than the test sequence, it is optimal to use a {\em single} local decision rule or channel to pre-process the training data and source sequence; this is analogous to \cite[Theorem 1]{tsitsiklis1988decentralized}.

\begin{figure}[t]
%	\begin{minipage}[t]{0.48\linewidth}
		\centering
		\includegraphics[scale=0.45]{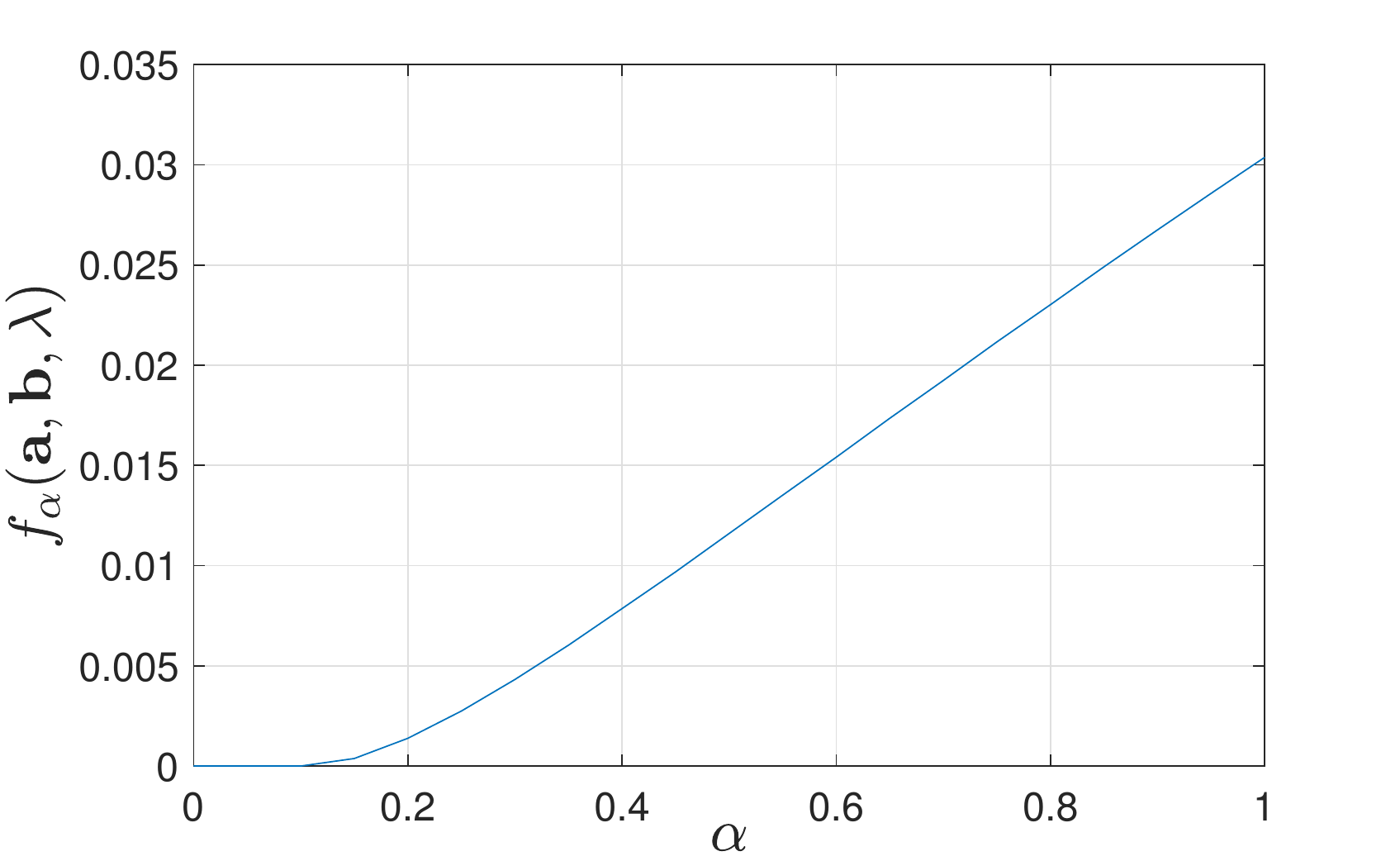}
		\caption{$f_{\alpha}(\ba,\bb,\lambda)$ for small $\alpha$,  fixed $(\ba,\bb,\lambda)$ and $K=2$}
		\label{fig:f_alpha}
%	\end{minipage}
\end{figure}

Given any $\alpha\in\bbR_+$ and any $(\ba,\bb)\in\calP([K])^2$, let
 \begin{align}
&\!\!\rmG_{\alpha}(\ba,\bb)\nn\\* 
&\!\!:=\min_{(\tilP,P)\in\calP(\calX)^2}\sum_{k\in[K]}\Big(a_kD(P_2W_k\|\tilP W_k) \nn\\* 
&\!\!\qquad+\alpha b_kD(P_1W_k\|\tilP W_k)+\alpha b_kD(P_2W_k\|P W_k)\Big)\\
&\!\!=\min_{\tilP\in\calP(\calX)}\sum_{k\in[K]}\Big(a_kD(P_2W_k\|\tilP W_k)+\alpha b_kD(P_1W_k\|\tilP W_k)\Big).\label{def:rmG}
\end{align}
Given any $\lambda\in\bbR_+$, let $\alpha_0(\ba,\bb,\lambda)$ be the solution (in $\alpha$) to the following equation
\begin{align}
\lambda=\rmG_{\alpha}(\ba,\bb).
\end{align}
Since $\rmG_{\alpha}(\ba,\bb)$ is an increasing function of $\alpha$ and $\rmG_0(\ba,\bb)=0$, for any $\lambda\in\bbR_+$, we have $\alpha_0(\ba,\bb,\lambda)>0$ unless $\lambda=0$. 

\begin{lemma}
\label{extreme:alpha=0}
Given any $(\ba,\bb)\in\calP([K])^2$ and any $\lambda\in\bbR_+$, if $\alpha \in [0,\alpha_0(\ba,\bb,\lambda)]$, then 
\begin{align}
f_{\alpha}(\ba,\bb,\lambda)=0.
\end{align}
\end{lemma}
We verified Lemma~\ref{extreme:alpha=0}  numerically by plotting $f_{\alpha}(\ba,\bb,\lambda)$ as a function of $\alpha$ when $\alpha$ is small  for certain values of $(\ba,\bb,\lambda)$  in Figure \ref{fig:f_alpha}. The proof of Lemma \ref{extreme:alpha=0} is straightforward since when $\alpha\leq \alpha_0(\ba,\bb,\lambda)$, $(\bQ^*,\tilde{\bQ}_1^*,\tilde{\bQ}_2^*)\in\calQ_{\lambda}(\alpha,\ba,\bb,\calW)$ where $Q_k^*=P_2W_k$, $\tilQ_{1,k}^*=P_1W_k$ and $\tilQ_{2,k}^*=P_2W_k$ and thus $\mathrm{LD}(\bQ^*,\tilde{\bQ}_1^*,\tilde{\bQ}_2^*,P_2,P_1,P_2|\alpha,\ba,\bb,\calW)=0$. The intuition is that when $\alpha$ is small enough, for any $\lambda>0$, the decision rule $\gamma$ in~\eqref{Eq:test} always declares $\rmH_1$, which means that $\beta_2(\gamma,P_1,P_2)=1$, so the corresponding exponent is identically $0$.

\begin{figure*}[t]
	\centering
	\subfigure[$\alpha=10$ and $W_1,W_2$ are both \emph{stochastic} matrices]{
		\begin{minipage}[t]{0.48\linewidth}
			\centering
			\includegraphics[width=2.5in]{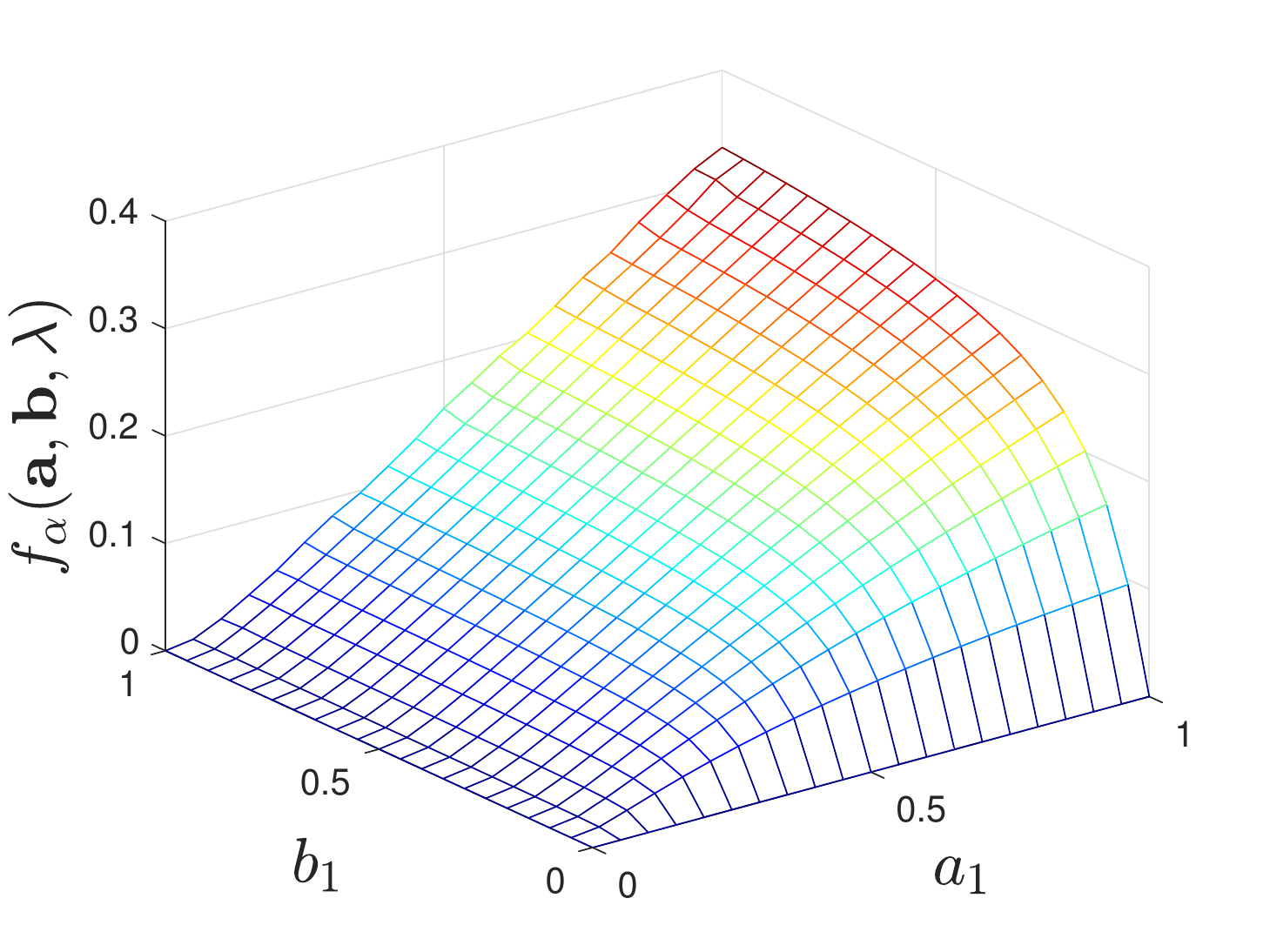}
			\label{fig:result1_withoutV}
	\end{minipage}}
	\subfigure[$\alpha=1$ and $W_1,W_2$ are both \emph{deterministic} matrices]{
		\begin{minipage}[t]{0.48\linewidth}
			\centering
			\includegraphics[width=2.5in]{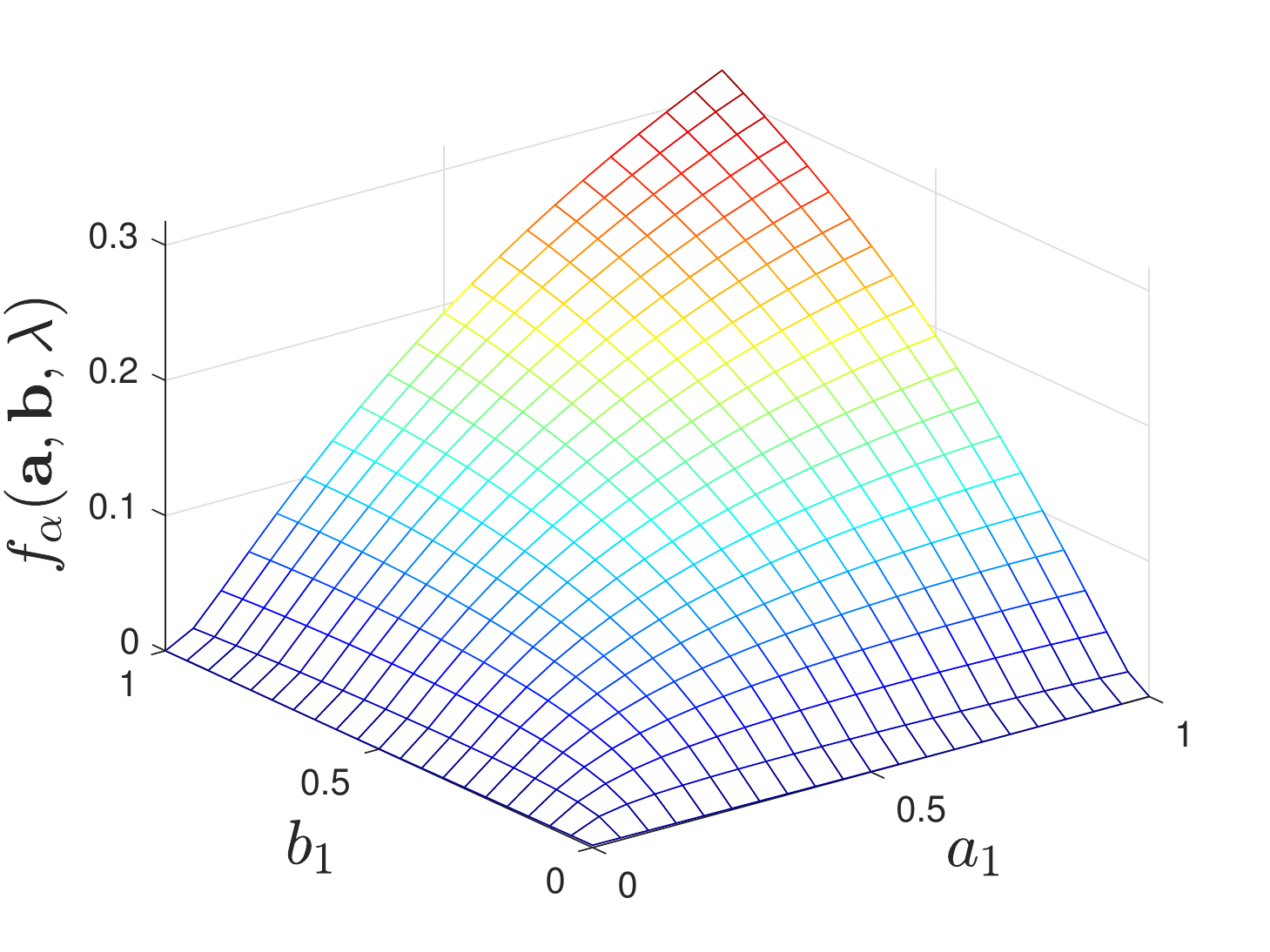}
			\label{fig:result2_withoutV}
	\end{minipage}}
	\caption{Numerical evaluations  of $f_{\alpha}(\ba,\bb,\lambda)$ when $K=2$ and $\lambda=0.01$. Note that the maxima occur at the corner points of $[0,1]^2$.}
	\label{fig:result_withoutV}
\end{figure*}

%\begin{figure}[t]
%\centering
%\subfigure[$\alpha=10$ and $W_1,W_2$ are both stochastic matrices]{
%\begin{minipage}[t]{0.48\linewidth}
%		\centering
%		\includegraphics[width=2.5in]{result3.pdf}
%		\label{fig:result1}
%	\end{minipage}}
%\subfigure[$\alpha=1$ and $W_1,W_2$ are both deterministic matrices]{
%\begin{minipage}[t]{0.48\linewidth}
%		\centering
%		\includegraphics[width=2.5in]{result2_resize.pdf}
%		\label{fig:result2}
%	\end{minipage}}
%\caption{Numerical evaluations  of $f_{\alpha}(\ba,\bb,\lambda)$ \blue{with Assumption ??} when $K=2$ and $\lambda=0.01$. Note that the maxima occur at the corner points of $[0,1]^2$.}
%\label{fig:result}
%\end{figure}

\subsection{ Numerical Study on Optimal Proportions of Local Decision Rules}
\label{sec:numerical}
%\subsection{Numerical Results for $\alpha_0(\ba,\bb,\lambda)<\alpha<\infty$}
In the following, we present numerical results to illustrate the properties of the optimal proportions of   local decision rules $(\ba_\alpha^*,\bb_\alpha^*):=\argmax_{\ba,\bb}f_{\alpha}(\ba,\bb,\lambda)$ when $\alpha\in (\alpha_0(\ba,\bb,\lambda),\infty)$; that is, $\alpha$ is moderate.

{\color{black}
When $K=2$, regardless of the stochasticity of the channels in $\calW$, we find that the maximal value of $f_{\alpha}(\ba,\bb,\lambda)$ always lies at a corner point of the feasible set of $(\ba,\bb)$. See Figure~\ref{fig:result_withoutV} for numerical examples.
%by calculating $f_{\alpha}(\ba,\bb,\lambda)$ for various  $(P_1,P_2)$ and $\calW$, we find that when $\alpha$ is moderate (i.e., neither $0$ nor $\infty$), the maximal value of $f_{\alpha}(\ba,\bb,\lambda)$ always lies at a corner point of the feasible set of $(\ba,\bb)$. This is shown numerically in Figures~\ref{fig:result_withoutV} for some choices of $W_1$ and $W_2$.
When $K\geq 3$,  we find that the results are more involved and it is not necessarily optimal to use only one local decision rule. To clarify our observations, we first describe cases where our numerical calculations of the exponent suggest that it is optimal to use only one local decision rule to achieve the optimal type-II error exponent; this is analogous to~\cite[Theorem 2]{tsitsiklis1988decentralized}. We then consider other cases and briefly discuss why it is not always optimal to use one local decision rule.

\subsubsection{When One Local Decision Rule is Optimal}
In most practical distributed detection systems, the local decision rule at each sensor is a {\em deterministic} compressor or quantizer. However, under certain conditions, randomized local decision rules can be used to provide privacy~\cite{gilani2019distributed,liao2017hypothesis,du2012privacy} or to satisfy power constraints~\cite[Sec.~IV]{tuncel2005error}. We  now describe a class of local decision rules  for which the exponent can be simplified and numerical calculations of the exponents suggest that full diversity of local decision rules is unnecessary.

Let $\calV_{\rmI}$ be the set of stochastic matrices (channels) with $M=|\calX|$ rows and $L=|\calZ|$ columns whose rows contain a permutation of the rows of $I_{L}$, the $L\times L$ identity matrix.  The set $\calV_{\rmI}$ includes all deterministic mappings (e.g., Figure \ref{fig:deterministic}) and a subset of stochastic mappings as long as for each $z\in \cal Z$, there exists an $x_z\in \cal X$ that maps directly to it, as illustrated in Figure~\ref{fig:random1}.  Note that Tsitsiklis~\cite{tsitsiklis1988decentralized}  considers only deterministic local decision rules, which certainly falls into the class $\calV_\rmI$.  The definition is extended in the obvious way if $M=\infty$ (i.e., for all $z\in\calZ$, there exists $x_z\in\calX$ such that $V(z|x_z)=1$).

\begin{figure}[t]
	\centering
	\setlength{\unitlength}{0.6cm}
	\subfigure[Stochastic $V$]{
		\begin{minipage}[t]{0.48\textwidth}
			\centering
			\scalebox{0.55}{
				\begin{picture}(5,9)(1,-5)
				%\linethickness{1pt}
				\Large
				\put(-2,4.5){\makebox(0,0){$\calX$}}
				\put(-2,3){\circle{1}\makebox(-2,0){1}}
				\put(-2,1.5){\circle{1}\makebox(-2,0){2}}
				\put(-2,0){\circle{1}\makebox(-2,0){3}}
				\put(-2,-1.5){\circle{1}\makebox(-2,0){4}}
				\put(-2,-3){\circle{1}\makebox(-2,0){5}}
				\put(-2,-4.5){\circle{1}\makebox(-2,0){6}}	
				
				\put(-1.5,3.2){\vector(4,-1){9}}
				\put(-1.5,1.5){\vector(2,-1){9}}
				\put(-1.5,1.5){\vector(4,-1){9}}
				\put(-1.5,0){\vector(3,-1){9}}
				\put(-1.5,-1.3){\vector(1,0){9}}
				\put(-1.5,-3){\vector(1,0){9}}
				\put(-1.7,-2.6){\vector(3,1){9.5}}
				\put(-1.5,-4.5){\vector(3,1){9}}
				\put(-1.7,-4.2){\vector(2,1){9.5}}
				
				\put(8,4.5){\makebox(0,0){$\calZ$}}	
				\put(8,1){\circle{1}\makebox(1.75,0){1}}
				\put(8,-1){\circle{1}\makebox(1.75,0){2}}
				\put(8,-3){\circle{1}\makebox(1.75,0){3}}			
				\end{picture}}
			\label{fig:random1}
	\end{minipage}}
	\subfigure[Deterministic $V$]{
		\begin{minipage}[t]{0.48\textwidth}
			\centering
			\scalebox{0.55}{
				\begin{picture}(5,9)(1,-5)
				%\linethickness{1pt}
				\Large
				\put(-2,4.5){\makebox(0,0){$\calX$}}
				\put(-2,3){\circle{1}\makebox(-2,0){1}}
				\put(-2,1.5){\circle{1}\makebox(-2,0){2}}
				\put(-2,0){\circle{1}\makebox(-2,0){3}}
				\put(-2,-1.5){\circle{1}\makebox(-2,0){4}}
				\put(-2,-3){\circle{1}\makebox(-2,0){5}}
				\put(-2,-4.5){\circle{1}\makebox(-2,0){6}}	
				
				\put(-1.5,3.2){\vector(4,-1){9}}
				\put(-1.5,1.5){\vector(2,-1){9}}
				\put(-1.5,0){\vector(3,-1){9}}
				\put(-1.5,-1.3){\vector(1,0){9}}
				\put(-1.5,-3){\vector(1,0){9}}
				%\put(-1.5,-4.5){\vector(3,1){9}}
				\put(-1.7,-4.2){\vector(2,1){9.5}}
				
				\put(8,4.5){\makebox(0,0){$\calZ$}}	
				\put(8,1){\circle{1}\makebox(1.75,0){1}}
				\put(8,-1){\circle{1}\makebox(1.75,0){2}}
				\put(8,-3){\circle{1}\makebox(1.75,0){3}}		
				\end{picture}}
			\label{fig:deterministic}
	\end{minipage}}
	\caption{Examples of random and deterministic $V\in\calV_I$.}
	\label{fig:W example}
\end{figure}
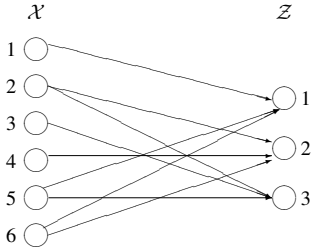
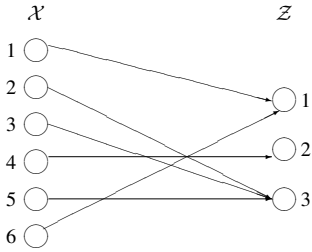

%However, under certain conditions, randomized local decision rules can be used to provide privacy\cite{gilani2019distributed,liao2017hypothesis,du2012privacy} or to satisfy power constraints\cite[Sec.~IV]{tuncel2005error}. 
%On the other hand, we place no restriction on the channels $W_j$'s for all $j\in[K]$ that are used to pre-process the source sequence $X^n$ and first training sequence $Y_1^N$. However, 

We assume that the second training sequence $Y_2^N$ is pre-processed by one local decision rule $V\in\calV_I$, i.e., $\tilY_{2,i}\sim V(\cdot|Y_{2,i})$ for all $i\in[N]$. The channels $\{W_j\}_{j\in[K]}$ that are used to  pre-process the test sequence $X^n$ and the first training sequence $Y_1^N$ are arbitrary. 
%Let $\calV_{\rmI}$ be the set of stochastic matrices (channels) with $M=|\calX|$ rows and $L=|\calZ|$ columns whose rows contain a permutation of the rows of $I_{L}$, the $L\times L$ identity matrix.
% \begin{assumption}\label{assump: W}
%  $V\in\calV_{\rmI}$.
% \end{assumption}
%The set $\calV_{\rmI}$ includes all deterministic mappings and a subset of stochastic mappings as long as for each $z\in \cal Z$, there exists an $x_z\in \cal X$ that maps directly to it, as illustrated in Figure~\ref{fig:random1}.  Note that Tsitsiklis~\cite{tsitsiklis1988decentralized}  considers only deterministic local decision rules, i.e., deterministic channels as in Figure~\ref{fig:deterministic}. The definition is extended in the obvious way if $M=\infty$ (i.e., for all $z\in\calZ$, there exists $x_z\in\calX$ such that $V(z|x_z)=1$). 
%Furthermore, there is no restriction on the channels $W_j$'s that are used to pre-process the source sequence $X^n$ and first training sequence $Y_1^N$.   
%This   assumption on $V$ is used in the converse proof of Theorem \ref{Thm:exponent}.
Under such a setting, we can simplify the asymptotically optimal error exponent and test (cf.~\eqref{Eq:thm1 exponent} and \eqref{Eq:test}) as follows:
\begin{align}
&\lim_{n\to\infty}E_{\calV_\rmI}^*(n,\alpha,P_1,P_2,\lambda|\ba,\bb,\calW) \nn\\*
&=\min_{\substack{(\bQ,\tilde{\bQ}) \\\in \calQ_{\lambda,\calV_I}(\alpha,\ba,\bb,\calW)}}\sum_{k\in[K]}\big(a_kD(Q_k\|P_2 W_k) \nn\\*
 &\qquad\qquad\qquad\qquad+\alpha b_k D(\tilQ_k\|P_1 W_k)\big)
\end{align}
where $\calQ_{\lambda,{\calV_\rmI}}(\alpha,\ba,\bb,\calW)
:=\big\{(\bQ,\tilde{\bQ})\in \calP([L])^{K}:\min_{\tilP\in\calP(\calX)}\sum_{k\in[K]}\big(a_kD(Q_k\|\tilP W_k)+\alpha b_k D(\tilQ_k\|\tilP W_k)\big)\leq \lambda\big\}$, and $\gamma_{\calV_\rmI}$ is given  in \eqref{eqn:gamma_calV_rmI} at the top of the next page.
\begin{figure*}
\begin{align}
\gamma_{\calV_\rmI}(Z^n,\tilY_1^N,\tilY_2^N) = \left\{
\begin{array}{cc}
\rmH_1 & \min_{\tilP}\sum_{k\in[K]}\big(a_kD(T_{Z^{na_k}}\|\tilP W_k)  +\alpha b_k D(T_{Y_1^{Nb_k}}\|\tilP W_k)\big)\leq \lambda,\\
\rmH_2 & \text{otherwise}.
\end{array}
\right.  \label{eqn:gamma_calV_rmI}
\end{align}\hrulefill
\end{figure*}

\begin{figure*}[t]
	\centering
	\subfigure[$\bb=(1,0,0)$]{
		\begin{minipage}[t]{0.48\linewidth}
			\centering
			\includegraphics[width=2.5in]{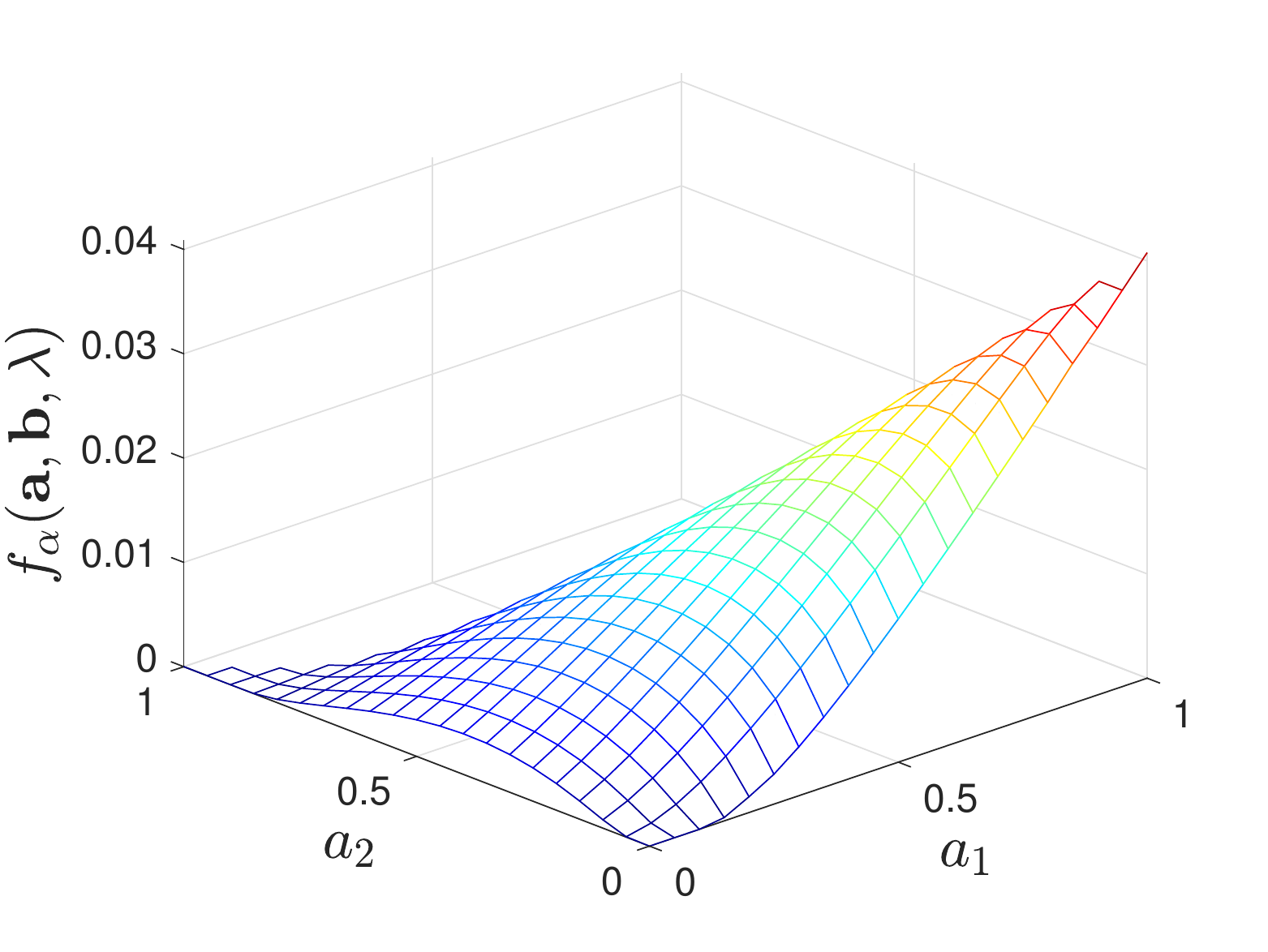}
			\label{fig:f K=3 b=100}
	\end{minipage}}
	\subfigure[$\bb=(0,1,0)$]{
		\begin{minipage}[t]{0.48\linewidth}
			\centering
			\includegraphics[width=2.5in]{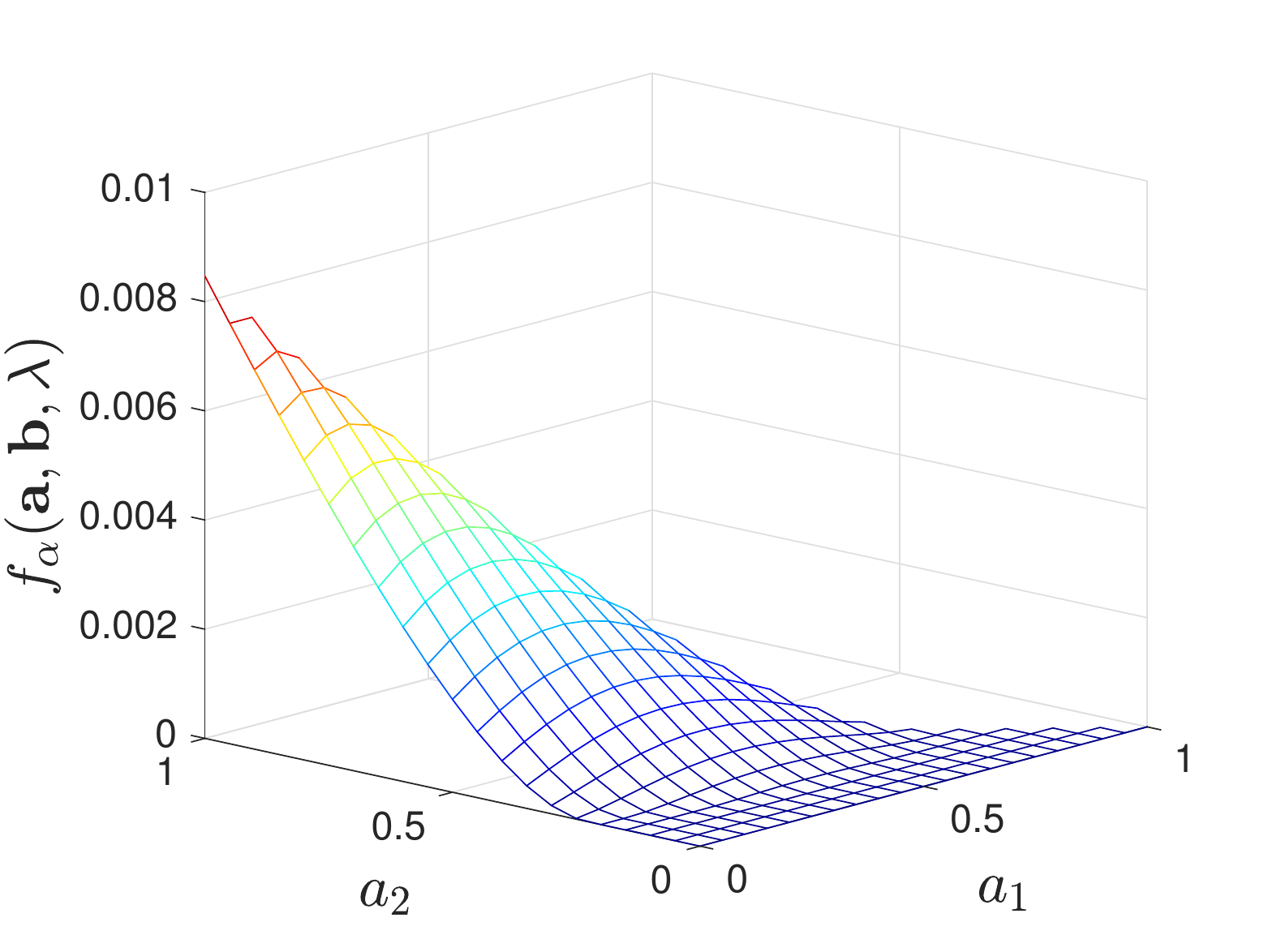}
			\label{fig:f K=3 b=010}
	\end{minipage}}
	\subfigure[$\bb=(0,0,1)$]{
		\begin{minipage}[t]{0.48\linewidth}
			\centering
			\includegraphics[width=2.5in]{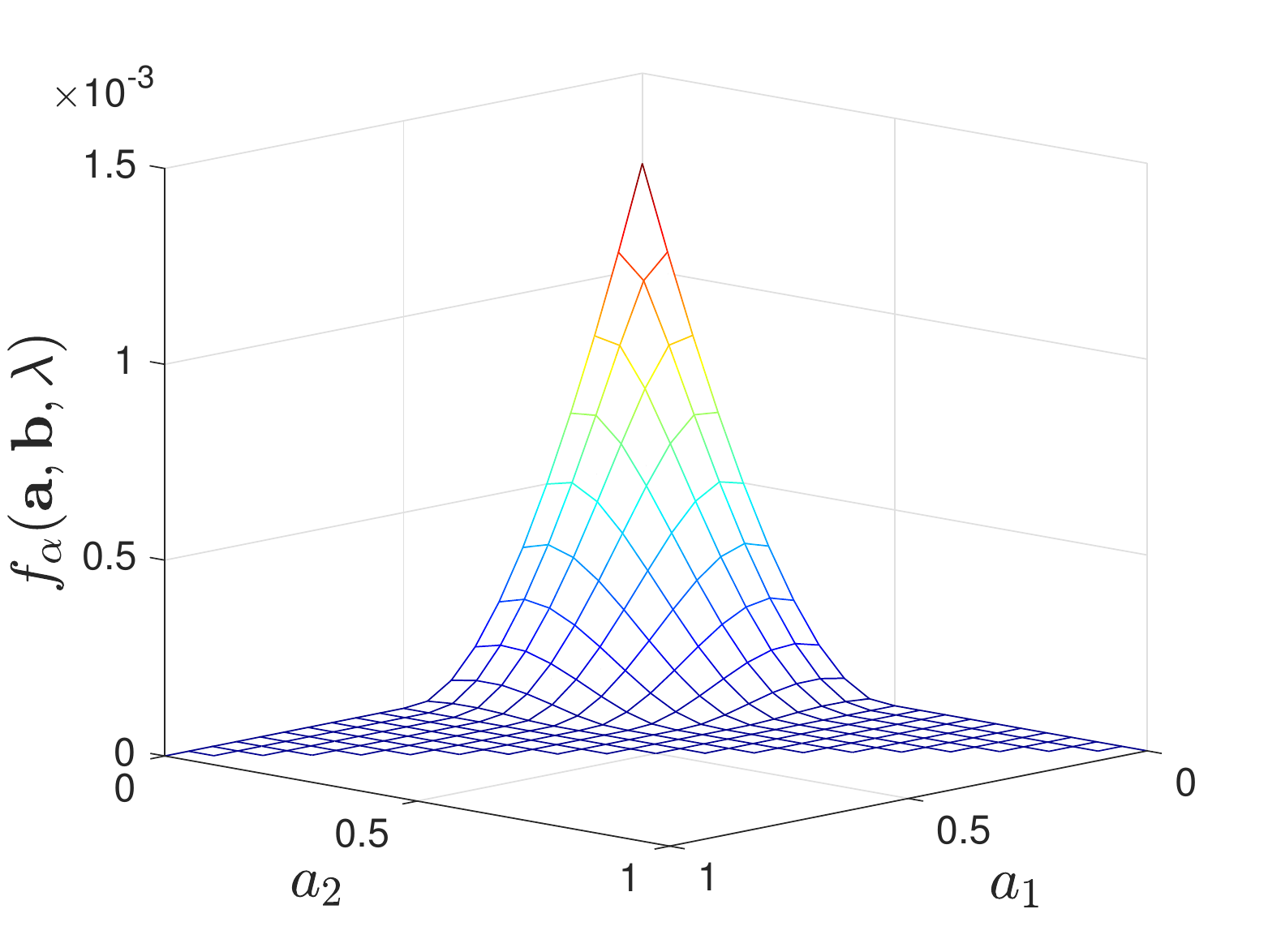}
			\label{fig:f K=3 b=001}
	\end{minipage}}
	\subfigure[$\bb=(\frac{1}{3},\frac{1}{3},\frac{1}{3})$]{
		\begin{minipage}[t]{0.48\linewidth}
			\centering
			\includegraphics[width=2.5in]{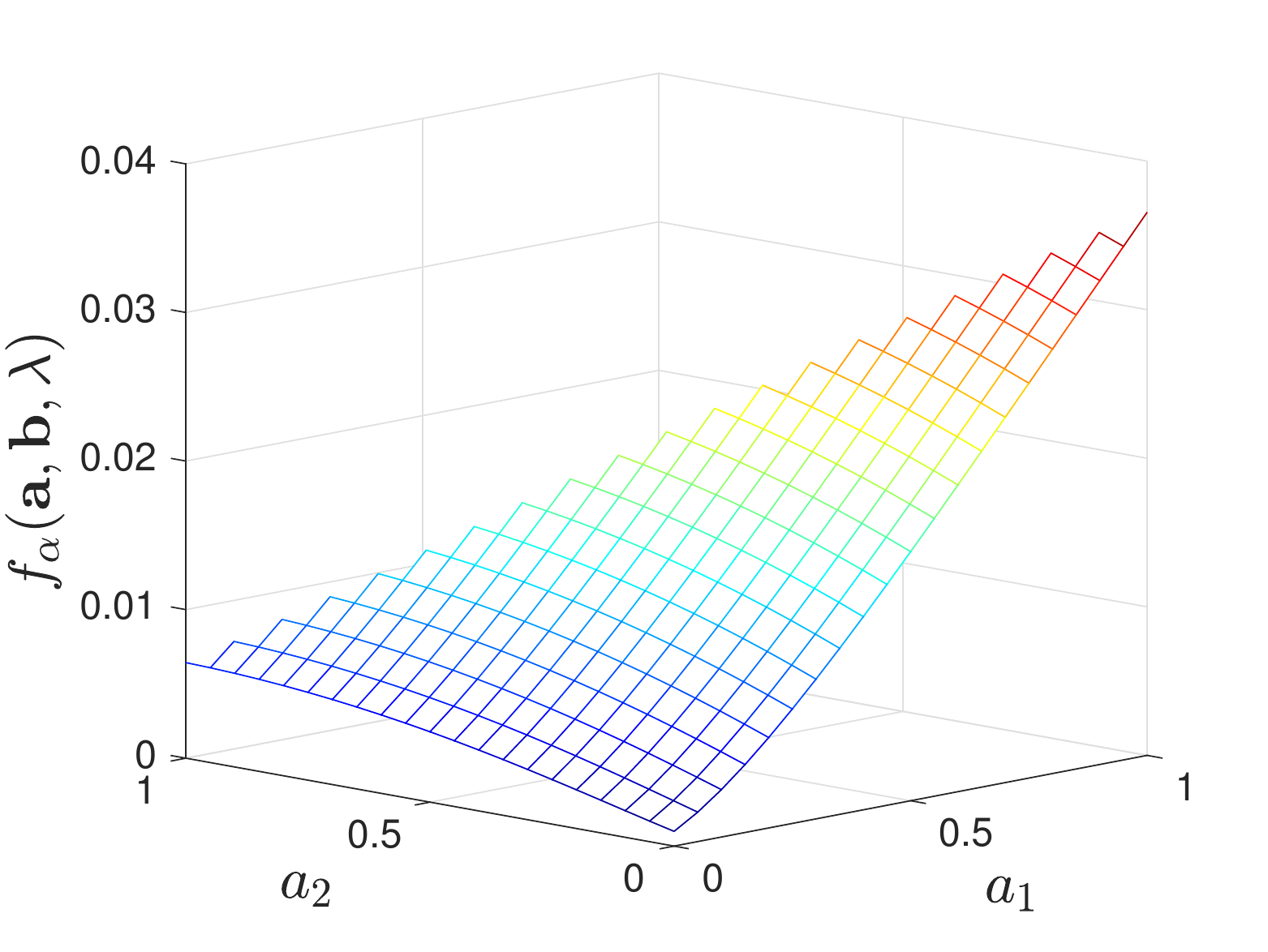}
			\label{fig:f K=3 b=0.3333}
	\end{minipage}}
	\caption{Numerical evaluations  of $f_{\alpha}(\ba,\bb,\lambda)$ with the assumption of $V\in\calV_I$ when $K=3$, $\lambda=0.01$, $\alpha=10$ and $W_1,W_2,W_3$ are deterministic. Note that for each vector $\bb$, the maxima with respect to $\ba$ occur at the corner points of $[0,1]^2$.}
	\label{fig:f K=3}
\end{figure*}

By calculating $f_{\alpha}(\ba,\bb,\lambda)$ for various  $(P_1,P_2)$ and $\calW$, we find that when $\alpha$ is moderate (i.e., neither $\le\alpha_0$ nor $\infty$), the maximal value of $f_{\alpha}(\ba,\bb,\lambda)$ always lies at a corner point of the feasible set of $(\ba,\bb)$, as shown in Figure~\ref{fig:f K=3}. %Apart from Figure~\ref{fig:result_withoutV}, we also illustrate this for the case of deterministic $W_1,W_2,W_3$ in Figure~\ref{fig:f K=3} for $K=3$.
Additional numerical results are shown in Appendix \ref{Sec:figures}. 
Inspired by these numerical results, we present the following conjecture:
%\begin{figure}[t]
%\centering
%\subfigure[$\alpha=10$ and $W_1,W_2$ are both stochastic matrices]{
%\begin{minipage}[t]{0.48\linewidth}
%		\centering
%		\includegraphics[width=2.5in]{result3.pdf}
%		\label{fig:result1}
%	\end{minipage}}
%\subfigure[$\alpha=1$ and $W_1,W_2$ are both deterministic matrices]{
%\begin{minipage}[t]{0.48\linewidth}
%		\centering
%		\includegraphics[width=2.5in]{result2_resize.pdf}
%		\label{fig:result2}
%	\end{minipage}}
%\caption{Numerical evaluations  of $f_{\alpha}(\ba,\bb,\lambda)$ \blue{with Assumption ??} when $K=2$ and $\lambda=0.01$. Note that the maxima occur at the corner points of $[0,1]^2$.}
%\label{fig:result}
%\end{figure}
}

\begin{conjecture} \label{conj:det}
 For all $\alpha,\lambda\in\bbR_+$, the vectors  $\ba_\alpha^*$ and $\bb_\alpha^*$ that maximize $f_{\alpha}(\ba,\bb,\lambda)$ are deterministic if 
%	one of the following conditions is satisfied
%	\begin{enumerate}
%	\item All local decision rules $W_j$'s for all $j\in[K]$ are in $\calV_{I}$;
%	\item 
	the second training sequence $Y_2^N$ is pre-processed by a single channel $V\in\calV_I$.
%	\end{enumerate}
\end{conjecture}

{\color{black}
\subsubsection{Results without Assumptions on Channels}

In this subsection, we discuss the general case where there is no assumption (e.g., deterministic or membership in $\calV_\rmI$) on local decision rules used to pre-process $X^n$, $Y_1^N$ and $Y_2^N$.

By calculating $f_{\alpha}(\ba,\bb,\lambda)$ for various  $(P_1,P_2)$ and $\calW$, we find that when $K=3$ and $\alpha$ is moderate, the maximal value of $f_{\alpha}(\ba,\bb,\lambda)$ does not always lie at a corner point; instead, it may occur at non-corner points within the feasible set of $(\ba,\bb)$. We illustrate this in Figure~\ref{fig:K3_Wrandom_withoutV} for the case of stochastic $W_1,W_2,W_3$ and for some $\bb=(b_1,b_2,b_3)$.

Our numerical  results suggest that without the knowledge of true distributions, it is not always optimal to use only one identical channel to process the test and training samples, which differs from the result in Tsitsiklis' paper \cite{tsitsiklis1988decentralized}. The difference can be explained intuitively as follows. First,  in \cite{tsitsiklis1988decentralized},  deterministic decision rules are considered while in our setting, we allow the channels to be stochastic. %Thus, the fusion center may need more information to make a final decision due to the randomness induced by the channels.
 Second, when $\alpha$ is moderate, we are not able to estimate the true distributions with arbitrary accuracy using the training samples. Thus, to combat the randomness induced by the channels and to compensate for the loss of (full) knowledge of the true distributions, the fusion center may require more information; hence the need for more diversity in the local decision rules.}

\begin{figure*}[t]
	\centering
	\subfigure[$\bb=(0.3,0.3,0.4)$]{
		\begin{minipage}[t]{0.48\linewidth}
			\centering
			\includegraphics[width=2.5in]{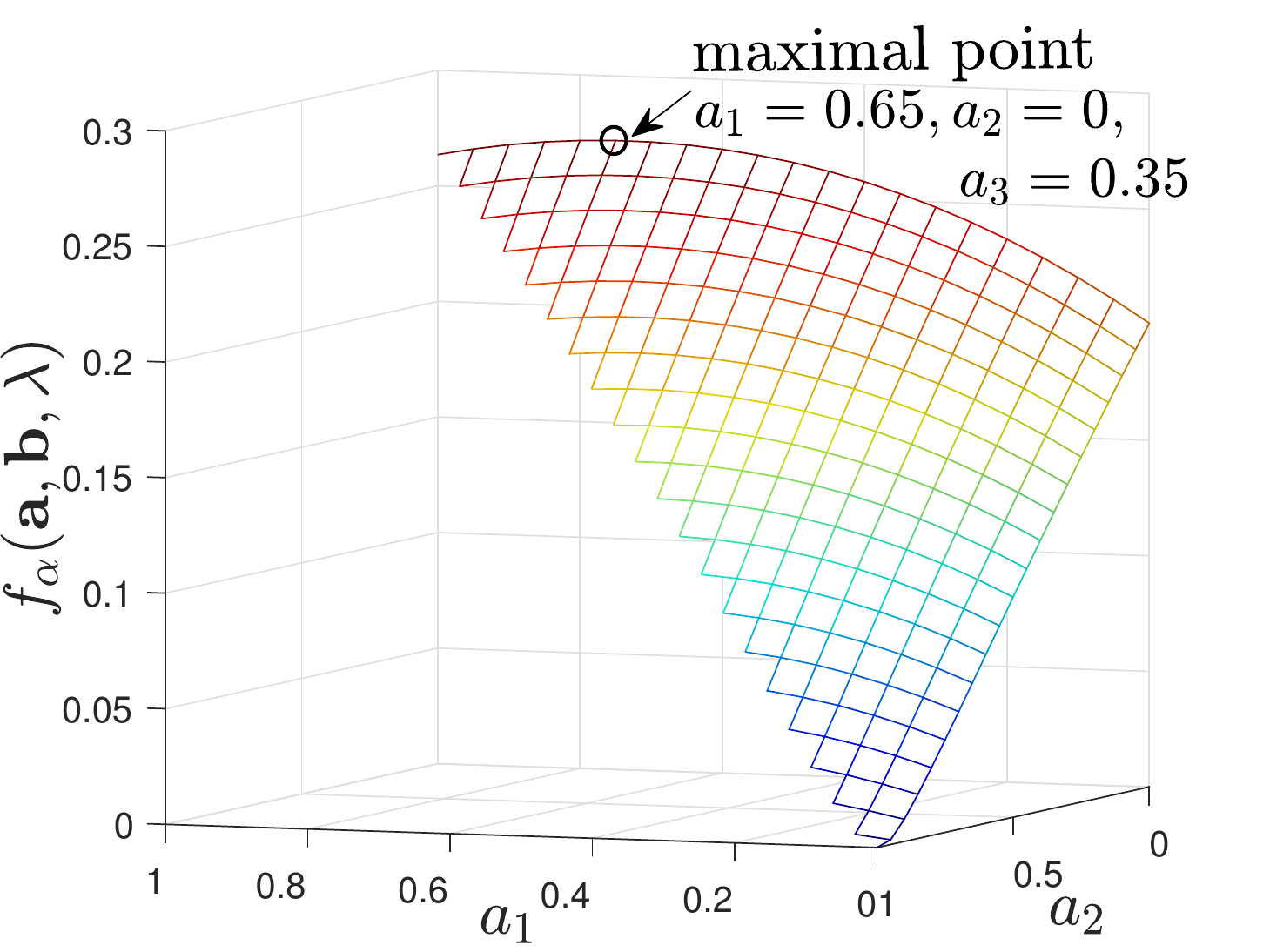}
			\label{fig:K3_Wrandom_1_withoutV}
	\end{minipage}}
	\subfigure[$\bb=(1,0,0)$]{
		\begin{minipage}[t]{0.48\linewidth}
			\centering
			\includegraphics[width=2.5in]{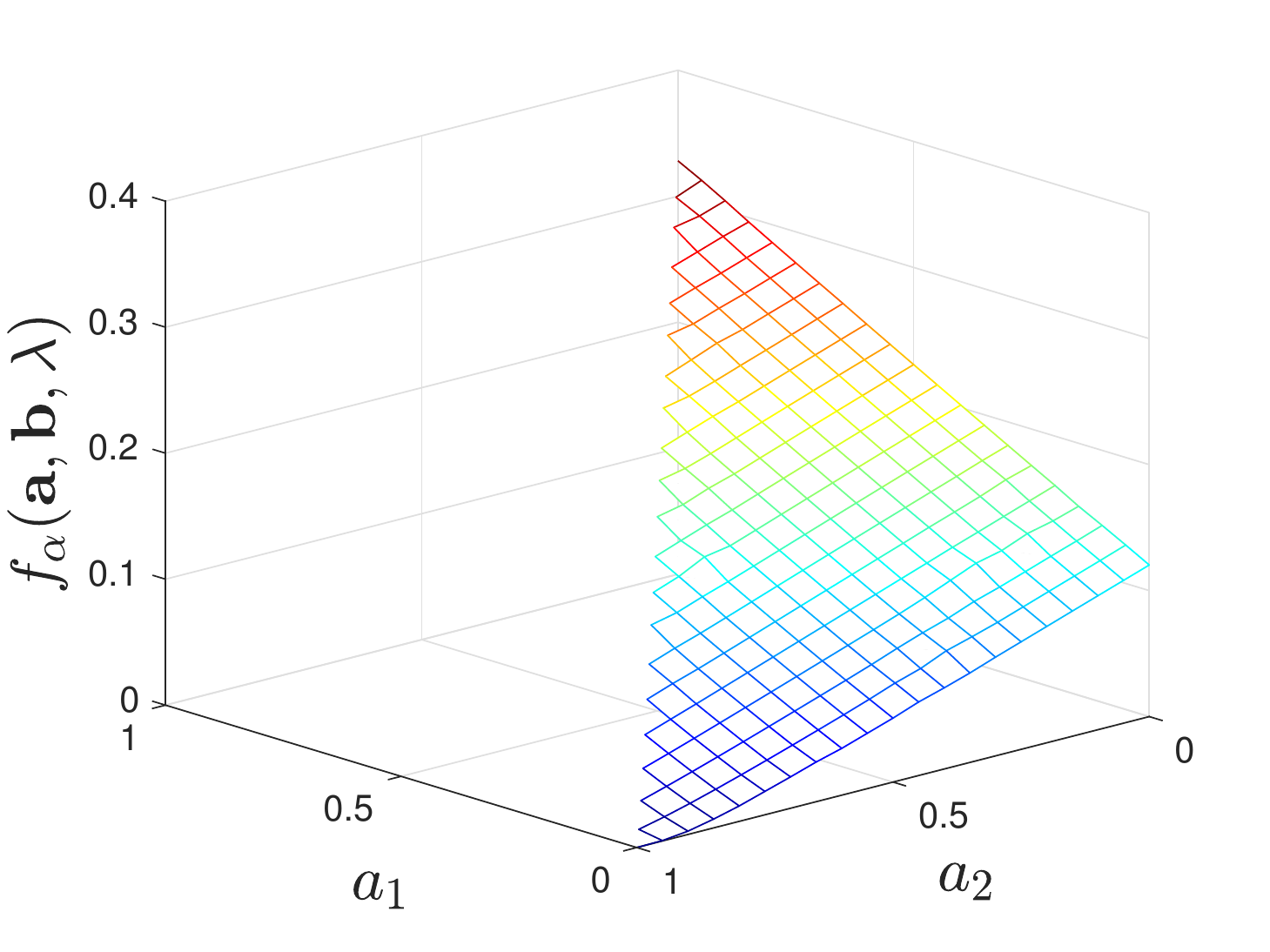}
			\label{fig:K3_Wrandom_2_withoutV}
	\end{minipage}}
	\caption{Numerical evaluations  of $f_{\alpha}(\ba,\bb,\lambda)$ when $K=3$, $\alpha=10$, $\lambda=0.01$ and $W_1,W_2,W_3$ are all stochastic matrices. Note that in Fig.~\ref{fig:K3_Wrandom_1_withoutV} the maximal value occurs at $(a_1,a_2,a_3)=(0.65,0,0.35)$. However, in Fig.~\ref{fig:K3_Wrandom_2_withoutV}, the maximal value occurs at $(a_1,a_2,a_3)=(1,0,0)$.}
	\label{fig:K3_Wrandom_withoutV}
\end{figure*}

\subsection{Connections to Results in Distributed Detection}
\label{sec:connections}

We discuss the connections between Theorem \ref{Thm:exponent} and   \cite{tsitsiklis1988decentralized}, which concerns distributed detection when the underlying distributions are known. %\blue{\sout{The direct parts of the following results are corollaries of Theorem~\ref{Thm:exponent}, Lemma~\ref{extreme:f} and Corollary~\ref{Coro:alpha extreme} by letting $\lambda\downarrow 0$ and by solving $\max_{k\in[K]}f_{\infty}(\be_k,\be_k,\lambda)=\lambda$ respectively. The (strong) converse parts follow from} \cite{tsitsiklis1988decentralized}. \sout{Since the justifications are straightforward, we omit them for the sake of brevity.}}
 Throughout this subsection, to emphasize the dependence of error probabilities on $(\ba,\bb)$, we use $\beta_\nu(\gamma,P_1,P_2|\ba,\bb)$ to denote the type-$\nu$ error probability with respect to distributions $(P_1,P_2)$ when test $\gamma$ is used at the fusion center.

We first consider the Neyman-Pearson setting~\cite[Sec.~11.8]{cover2012elements}. Given any $\varepsilon\in[0,1]$, let $\Gamma_{\varepsilon}(\ba,\bb)$ be the set of tests satisfying that for all $(\tilP_1,\tilP_2)\in\calP(\calX)^2$,
\begin{equation}\label{eqn:upper_bd_exp}
\beta_1(\gamma,\tilP_1,\tilP_2|\ba,\bb)\leq \varepsilon.
\end{equation}
Let  the {\em optimal type-II error probability} subject to~\eqref{eqn:upper_bd_exp} be
\begin{align}
\beta_2^*(P_1,P_2)&:=\min_{\substack{(\ba,\bb)\in \\ \calP_n([K])\times\calP_{\alpha n}([K])}}\min_{\gamma\in\Gamma_{\varepsilon}(\ba,\bb)}\beta_2(\gamma,P_1,P_2|\ba,\bb).
\end{align}
Note that $\beta_2^*(P_1,P_2)$ depends on $n$, $\alpha$ and $\varepsilon$ but this dependence is suppressed for the sake of brevity. 

\begin{corollary}\label{Coro:tsitsiklis neyman}
Given any $V\in\calV_\rmI$ and any   $(P_1,P_2)\in\calP(\calX)^2$, 
\begin{equation}
\begin{aligned}
\!\!\lim_{\alpha\to\infty}\lim_{n\rightarrow\infty} \frac{1}{n}\log \frac{1}{\beta_2^*(P_1,P_2)}\!=\!\max_{k\in[K]}D(P_1W_k\|P_2W_k).
\end{aligned}\label{eqn:neyman}
\end{equation} 
\end{corollary}

\begin{proof}[Proof sketch of Corollary \ref{Coro:tsitsiklis neyman}]
The direct parts of the result are corollaries of Theorem~\ref{Thm:exponent}, Lemma~\ref{extreme:f} and Corollary~\ref{Coro:alpha extreme} by letting $\lambda=\frac{1}{n}\log\frac{1}{\varepsilon}\downarrow 0$. The converse parts follow from \cite[Theorem~2]{tsitsiklis1988decentralized} where the distributions are known since one can never obtain better (larger) exponents with unknown distributions than with known distributions.
Since the justifications are straightforward, we omit the details for the sake of brevity. 
\end{proof}

We also consider the Bayesian setting. Assume the prior probabilities for $\rmH_1$ and $\rmH_2$ are $\pi_1$ and $\pi_2$ respectively. Clearly, $\pi_1+\pi_2=1$. Given any $(\ba,\bb)\in\calP_n([K])\times\calP_{\alpha n}([K])$ and any $\calW$, let the {\em Bayesian error probability} be 
\begin{align}
\rmP_\rme(\gamma,P_1,P_2|\ba,\bb)&:=\sum_{i=1}^2\pi_i\beta_i(\gamma,P_1,P_2|\ba,\bb). %+\pi_2\beta_2(\gamma,P_1,P_2|\ba,\bb).
\end{align}
Furthermore, let the maximum {\em Chernoff information} between $P_2W_k$ and $P_1W_k$ be
\begin{equation}
   \lambda^*=\max_{k\in[K]}\max_{\rho\in[0,1]} \log \frac{1}{\sum_{z }(P_2W_k)^{\rho}(z)(P_1W_k)^{1-\rho}(z)}.
\end{equation}
and let $\Gamma_{\rm{Bayes}}(\ba,\bb)$ be the set of tests at the fusion center satisfying that for all $(\tilP_1,\tilP_2)\in\calP(\calX)^2$,
\begin{equation}\label{Eq:beta1 leq lambdak}
\beta_1(\gamma,\tilP_1,\tilP_2|\ba,\bb)\leq\exp(-n\lambda^*).
\end{equation}
Finally, let the {\em optimal Bayesian error probability} be 
\begin{align}
&\rmP_\rme^*(P_1,P_2)\nn\\*
&:=\min_{\substack{(\ba,\bb) \\ \in\calP_n([K])\times \calP_{\alpha n}([K])}}\min_{\gamma\in\Gamma_{\rm{Bayes}}(\ba,\bb)} \rmP_\rme(\gamma,P_1,P_2|\ba,\bb).
\end{align}
Again $\rmP_\rme^*(P_1,P_2)$ depends on both $n$ and $\alpha$.
\begin{corollary}\label{Coro:bayesian}
Given any $V\in\calV_\rmI$ and any  $(P_1,P_2)\in\calP(\calX)^2$,
%	\begin{equation}\label{Eq:bayesian exponent}
%    \begin{aligned}
%		&\lim_{\alpha\to\infty}\lim\limits_{n\rightarrow\infty} \frac{1}{n}\log\frac{1}{ \rmP_\rme^*(P_1,P_2)}\\
%		&=\max_{k\in[K]}\max_{\rho\in[0,1]} \log\left(\frac{1}{\sum\limits_{z\in\mathcal{Z}}(P_2W_k)^{\rho}(z)(P_1W_k)^{1-\rho}(z)}\right).
%	\end{aligned}
%	\end{equation}
\begin{equation}\label{Eq:bayesian exponent}
\lim_{\alpha\to\infty}\lim\limits_{n\rightarrow\infty} \frac{1}{n}\log\frac{1}{ \rmP_\rme^*(P_1,P_2)}=\lambda^*.
\end{equation}
\end{corollary}

\begin{proof}[Proof sketch of Corollary \ref{Coro:bayesian}]
	The direct parts of the following results are corollaries of Theorem~\ref{Thm:exponent}, Lemma~\ref{extreme:f} and Corollary~\ref{Coro:alpha extreme} by solving $\max_{k\in[K]}f_{\infty}(\be_k,\be_k,\lambda)=\lambda$. 
	The (strong) converse parts follow from \cite{tsitsiklis1988decentralized}.
\end{proof}
Under the Bayesian setting, the exponents of the type-I and type-II error probabilities are equal~\cite[Thm.~11.9.1]{cover2012elements}. 

Note that Corollaries \ref{Coro:tsitsiklis neyman} and \ref{Coro:bayesian} are analogous to {\em  distributed detection} \cite{tsitsiklis1988decentralized} for the binary case under the Neyman-Pearson and Bayesian settings respectively where the true distributions $(P_1,P_2)$ are known. The intuition is that when the lengths of the training sequences are much longer than that of the source sequence (i.e., $\alpha\to\infty$), we can estimate the true distributions to arbitrary precision, i.e., as accurately as desired.

\section{$m$-ary Distributed Detection with the Rejection Option and Training Samples} \label{sec:mary}
%\green{Can we set the title analogous to the binary case? e.g. $m$-ary Distributed Detection with Rejection and Training Samples?}

In this section, we generalize the binary distributed detection problem to the scenario in which we desire to discriminate between $m\ge 2$ hypotheses with the rejection option. Our main contribution here is the identification of an appropriate  test statistic and test that achieves the optimal rejection exponent for a fixed lower bound on all error exponents.

\subsection{Problem Formulation} 

%Recall the problem formulation of binary distributed detection in Section \ref{Sec:problem}. 
In the $m$-ary distributed detection problem, there are $m$ training sequences $\{Y_i^N\}_{i\in[m]}$, each generated i.i.d.\ according to an \emph{unknown} distribution $P_i\in\calP(\calX)$. There are  $n$ sensors. Each sensor observes a source symbol $X_i$ and compress/processes it into a noisy version  $Z_i$ just as in the binary distributed detection problem. Given noisy training sequences $\{\tilY_i^N\}_{i\in[m]}$ and the compressed source sequence $Z^n$, in which $\tilY_{i,j}\sim W_{g(j)}(\cdot|Y_{i,j})$ for all $i\in[m]$ and $j\in[N]$, the fusion center uses a decision rule $\gamma:[L]^{mN+n}\mapsto \{\rmH_1,\ldots,\rmH_m, \rmH_\rmr\}$ to discriminate among the following $m+1$ hypotheses:
\begin{itemize}
	\item $\rmH_j, \, \forall \, j\in[m]$: the source sequence $X^n$ and the $j$-th training sequence $Y_j^N$ are generated according to the same distribution;
	\item $\rmH_\rmr$: the source sequence $X^n$ is generated according to a distribution different from those which the training sequences are generated from and hence we reject all $\rmH_j,\,\forall\,  j \in [m]$.
\end{itemize}
Thus, the decision rule $\gamma$ partitions the sample space $[L]^{mN+n}$ into $m+1$ disjoint regions: $m$ acceptance regions $\{\Lambda_j(\gamma) \}_{j\in[m]}$, where $\Lambda_j(\gamma)$ favors hypothesis $\rmH_j$, i.e.,
\begin{align}
\Lambda_j(\gamma)&:=\Big\{(z^n,\tily_1^N,\ldots,\tily_m^N)\in[L]^{mN+n}: \nn\\* 
&\qquad\gamma(Z^n,\tily_1^N,\ldots,\tily_m^N)=\rmH_j \Big\},
\end{align}
and one rejection region $\Lambda_\rmr(\gamma):=(\cup_{j\in[m]}\Lambda_j(\gamma))^{\rmc}$ which favors hypothesis $\rmH_\rmr$. Note that here we assume that all $m$ training sequences are processed with channels in $\calW$ using the same index mapping function $g$. That is, all the first components $Y_{1,1},\ldots , Y_{m,1}$  are passed through the {\em same} channel, which is one element from $\calW$. The same is true for all the other $N-1$ components.

For conciseness, we  set $\bY^N=(Y_1^N,\ldots,Y_m^N)$ and use $\tilde{\bY}^N$ similarly. Furthermore, we set $\bP=(P_1,\ldots,P_m)$ and use $\tilde{\bP}$ and $\bQ$ similarly. Recall the definition of $\ba, \bb$ and the assumption that $N=\lceil\alpha n\rceil$. Given any decision rule $\gamma$ at the fusion center and any tuple of distributions $\bP$, the performance metrics we consider are the error probabilities and the rejection probabilities for each $j\in[m]$, i.e.,
\begin{align}
\beta_j(\gamma,\bP|\ba,\bb,\calW)&:=\bbP_j\{\gamma(Z^n,\tilde{\bY}^N)\notin\{ \rmH_j, \rmH_\rmr \} \},\\
\rej(\gamma,\bP|\ba,\bb,\calW)&:=\bbP_j\{\gamma(Z^n,\tilde{\bY}^N)=\rmH_\rmr \}.
\end{align}

We use $\beta_j(\gamma,\bP)$ and $\rej(\gamma,\bP)$ in place of $\beta_j(\gamma,\bP|\ba,\bb,\calW)$ and $\rej(\gamma,\bP|\ba,\bb,\calW)$ respectively if there is no risk of confusion.
For this setting, we are interested in tests that can simultaneously ensure exponential decay of the error probabilities under any hypothesis for {\em any} tuple of distributions and  exponential decay of  the rejection probabilities under each hypothesis for a {\em particular} tuple of distributions. To be concrete, given any tuple of distributions $\bP$ and any $\lambda\in\bbR_+$, we are interested in the following optimal exponent  of the rejection probability under hypothesis $\rmH_j$:
\begin{align}
&\nn E_j^*(n,\alpha,\bP,\lambda|\ba,\bb,\calW)\nn\\*
&:=\sup\Big\{E_j\in\bbR: \exists~ \gamma ~ \mathrm{s.t.}~ \forall~j\in[m], \nn\\* 
&\qquad\qquad\beta_j(\gamma,\tilde{\bP})\leq\exp(-n\lambda),~\forall~\tilde{\bP}\in\calP(\calX)^m, \nn\\*
&\qquad\qquad\rej(\gamma,\bP)\leq \exp(-nE_j) \Big\}\label{def:Ej*}.
\end{align} 
We emphasize that in this formulation, under each hypothesis, the error exponent is at least $\lambda$ for all tuples of distributions.

\subsection{Main Results}\label{Sec:mary results}
Before presenting the main result, we present some preliminary definitions.  Recall that for each $k\in[K]$, we use $z^{na_k}$ to denote the collection of $z_i$ satisfying $h(i)=k$, use $\bz^{n\ba}$ to denote $(z^{na_1},\ldots,z^{na_K})$ and use $\bT_{\bz^{n\ba}}$ to denote the vector of types $(T_{z^{na_1}},\ldots,T_{z^{na_K}})$. Similarly, for each $k\in[K]$ and $j\in[m]$, we use the notations $\tily_j^{Nb_k}$, $\tilde{\by}_j^{N\bb}$ and $\bT_{\tilde{\by}_j^{N\bb}}$. Given any tuple of distributions $\bP=(P_1,\ldots,P_m)\in\calP(\calX)^m$ and any $j\in[m]$, define the following linear combination of divergences
%\begin{align}
%\mathrm{LD}_j^*(\bT_{\bZ^{n\ba}},\bT_{\tilde{\bY}_j^{N\bb}})&:=\min_{\tilP\in\calP(\calX)}\mathrm{LD}(\bT_{\bZ^{n\ba}},\bT_{\tilde{\bY}_j^{N\bb}}, \tilP,\tilP|\alpha,\ba,\bb,\calW ),\\
%\mathrm{LD}_j^{[m]}(\bT_{\bZ^{n\ba}},\{\bT_{\tilde{\bY}_i^{N\bb}}\}_{i\in[m]},\bP)&:=\sum_{k\in[K]}\Big(a_kD(T_{Z^{na_k}}\| P_jW_k)+\sum_{i\in[m]}\alpha b_k D(T_{\tilY_i^{Nb_k}}\|P_i W_k) \Big).\label{def:LDjm}\\
%\label{Eq:mary LD} \mLD_j(\bT_{\bZ^{n\ba}},\{\bT_{\tilde{\bY}_i^{N\bb}}\}_{i\in[m]})&:=\mathrm{LD}_j^*(\bT_{\bZ^{n\ba}},\bT_{\tilde{\bY}_j^{N\bb}})+\sum_{i\in[m]: i\neq j}\sum_{k\in[K]}\alpha b_k D(T_{\tilY_i^{Nb_k}}\| \tilP_i^* W_k),
%\end{align}
\begin{align}
&	\mathrm{LD}_j^{[m]}(\bT_{\bZ^{n\ba}},\{\bT_{\tilde{\bY}_i^{N\bb}}\}_{i\in[m]},\bP)\nn\\* 
	&:=\sum_{k\in[K]}\Big(a_kD(T_{Z^{na_k}}\| P_jW_k)+\sum_{i\in[m]}\alpha b_k D(T_{\tilY_i^{Nb_k}}\|P_i W_k) \Big),\label{def:LDjm}
\end{align}
and furthermore, let
\begin{align}
&	\mLD_j(\bT_{\bZ^{n\ba}},\{\bT_{\tilde{\bY}_i^{N\bb}}\}_{i\in[m]})\nn\\* 
	&:=\min_{\bP\in\calP(\calX)^m}\mathrm{LD}_j^{[m]}(\bT_{\bZ^{n\ba}},\{\bT_{\tilde{\bY}_i^{N\bb}}\}_{i\in[m]},\bP)\label{Eq:mary LD}.
\end{align}
Note that $\mathrm{LD}_j^{[m]}(\cdot)$ 
%is analogous to $\mathrm{LD}$ in \eqref{def:LD}, but the former takes into account $m$ hypotheses (set $m=2$ and respectively let $P_j$ and $\{P_i\}_{i\in[m]}$ in \eqref{def:LDjm} to be  $P$ and $\{\tilP_i\}_{i\in[2]}$ in \eqref{def:LD}, respectively).
is a slight generalization of $\mathrm{LD}(\cdot)$ in \eqref{def:LD}.

In the following, we will see that $ \mLD_j(\cdot)$ is an appropriate  test statistic that will be used in the achievability proof and an optimized version of $\mathrm{LD}_j^{[m]}(\cdot)$ is the corresponding exponent. Finally, given any $(i,l)\in[m]^2$ satisfying $i\neq l$, we define the following set of the collections of distributions:
\begin{align}\label{def:mary calQ}
&\calQ_{\lambda,i,l}(\alpha,\ba,\bb,\calW):=\Big\{(\bQ,\{\tilde{\bQ}_i\}_{i\in[m]})\in\calP(\calX)^{(m+1)K}:\nn\\* 
&\quad \mLD_i(\bQ,\{\tilde{\bQ}_i\}_{i\in[m]})\leq \lambda,\mLD_l(\bQ,\{\tilde{\bQ}_i\}_{i\in[m]})\leq \lambda \Big\}.
\end{align}
Note that if we choose $(\bQ,\{\tilde{\bQ}_i\}_{i\in[m]})\in\calP(\calX)^{(m+1)K}$ such that $\bQ=\tilde{\bQ}_1=\ldots=\tilde{\bQ}_m$ and $Q_1=P_0W_1,~Q_2=P_0W_2,\ldots,Q_K=P_0W_K$ for any $P_0\in\calP(\calX)$, then for all $j \in [m]$, we have
\begin{equation}
\mathrm{LD}_j^*(\bQ,\tilde{\bQ}_j)=\mathrm{LD}_j^*(\{P_0W_k\}_{k\in[K]},\{P_0W_k\}_{k\in[K]})=0
\end{equation}
and 
\begin{align}
&\min_{\tilP_i}\sum_{k\in[K]}\alpha b_k D(\tilQ_{i,k}\| \tilP_i W_k) \nn\\*
&=\min_{\tilP_i}\sum_{k\in[K]}\alpha b_k D(P_0W_k\| \tilP_i W_k)=0, \;\;\forall\, i\in[m],
\end{align}
which implies that $\mLD_i(\bQ,\{\tilde{\bQ}_i\}_{i\in[m]})=0$ for all $i\in[m]$. Thus, $\calQ_{\lambda,i,l}(\alpha,\ba,\bb,\calW)$ is non-empty for any $(\alpha,\ba,\bb,\calW)$.

Our main result in this section is the following asymptotic characterization of $E_j^*(n,\alpha,\bP,\lambda|\ba,\bb,\calW)$.
\begin{theorem}\label{Thm:rejection exponent}
	Given any $(\lambda,\alpha)\in\mathbb{R}_+^2$ and any tuple of target distributions $\bP\in\calP(\calX)^m$, for all $j\in[m]$, we have
	\begin{align}\label{Eq:rejection exponent}
	&\lim\limits_{n\to\infty}E_j^*(n,\alpha,\bP,\lambda|\ba,\bb,\calW)\nn\\*
	&=\min_{(i,l)\in[m]^2:l\neq i}\min_{\substack{(\bQ,\{\tilde{\bQ}_i\}_{i\in[m]})\\
			\in\calQ_{\lambda,i,l}(\alpha,\ba,\bb,\calW)}}\mathrm{LD}_j^{[m]}(\bQ,\{\tilde{\bQ}_i\}_{i\in[m]},\bP).
	\end{align}
\end{theorem}
The proof of Theorem \ref{Thm:rejection exponent} is given in Appendix \ref{proof:rej exponent}. %Note that in this theorem, Assumption~\ref{assump: W} is not required. Indeed, we only consider the channels in the set $\calW=\{W_i\}_{i \in [K]}$ (and not $V$).

First, as shown in Fig.~\ref{fig:mary_rej_expo}, there exists $\lambda\in\bbR_+$ such that $\lim_{n\to\infty}E_j^*(n,\alpha,\bP,\lambda|\ba,\bb,\calW) < \lambda$, which implies that the type-$j$ rejection exponents can be designed to be smaller than all the error exponents with an appropriate choice of $\lambda$. This scenario is reminiscent of practical communication scenarios~\cite{Forney68, HayashiTan15} (automatic repeat request or ARQ) where the rejection probability is designed to be much larger than the (undetected) error probability as declaring a rejection is typically much less costly than a genuine mistake being made.
\begin{figure*}[t]
	\centering
	\subfigure[$W_1,W_2$ are both deterministic matrices]{
		\begin{minipage}[t]{0.48\linewidth}
			\centering
			\includegraphics[width=3in]{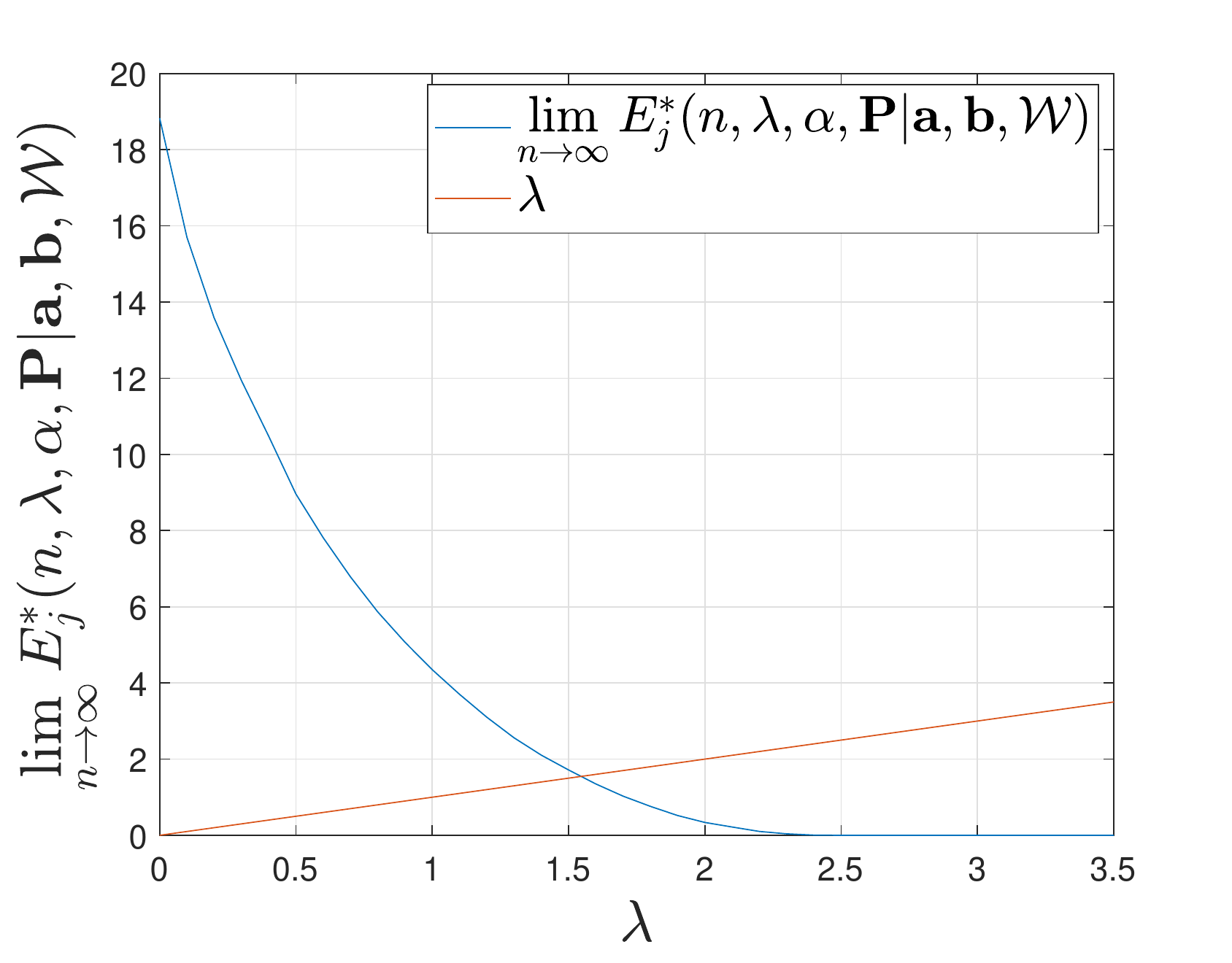}
			\label{fig:mary_Wdeter_exponent}
	\end{minipage}}
	\subfigure[$W_1,W_2$ are both stochastic matrices]{
		\begin{minipage}[t]{0.48\linewidth}
			\centering
			\includegraphics[width=3in]{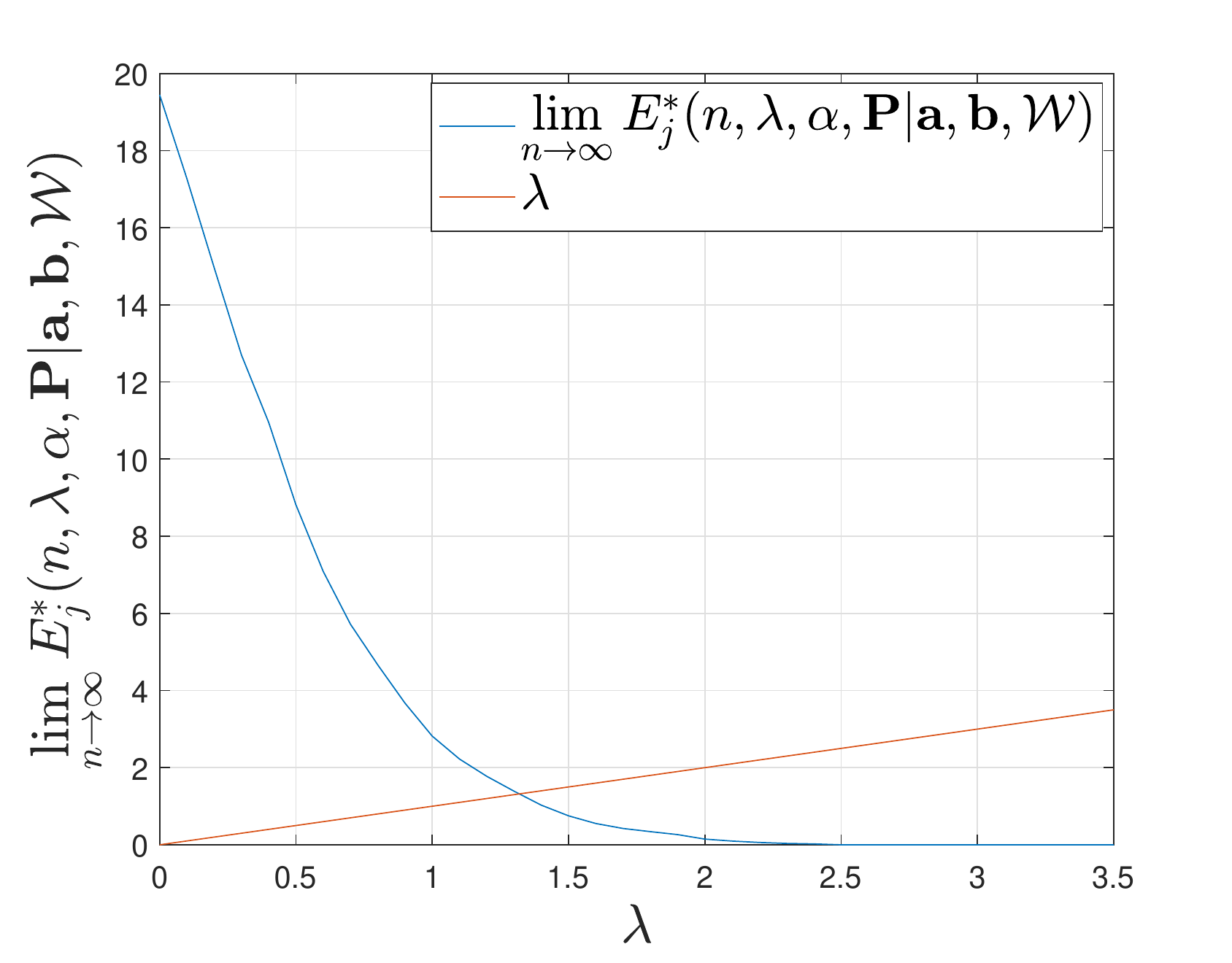}
			\label{fig:mary_Wrandom_exponent}
	\end{minipage}}
	\caption{Numerical evaluations  of $\lim_{n\to\infty}E_j^*(n,\alpha,\bP,\lambda|\ba,\bb,\calW)$ in \eqref{Eq:rejection exponent} versus $\lambda$ when $m=4$, $K=2$, $\ba=(0.2,0.8)$, $\bb=(0.6,0.4)$ and $\alpha=10$.}
	\label{fig:mary_rej_expo}
\end{figure*}

Second,  let us describe the test that is used in the achievability proof of  Theorem \ref{Thm:rejection exponent}. This test is a generalized version of Unnikrishnan's test~\cite{unnikrishnan2015asymptotically}.  Given any tuple of types $(\bT_{\bZ^{n\ba}},\{\bT_{\tilde{\bY}_i^{N\bb}}\}_{i\in[m]})$, we define the indices of the minimum and the second minimum of $\mLD_i(\bT_{\bZ^{n\ba}},\{\bT_{\tilde{\bY}_i^{N\bb}}\}_{i\in[m]})$ over all $i\in[m]$ as
\begin{align}
\label{def:i1}\!\!\! i_1\!=\! i_1(\bZ^n,\tilde{\bY}^N)\!:=\!\argmin_{i\in[m]}\mLD_i(\bT_{\bZ^{n\ba}},\{\bT_{\tilde{\bY}_i^{N\bb}}\}_{i\in[m]}) 
\end{align}
and
\begin{align}
\label{def:i2}\!\!\! i_2\!=\! i_2(\bZ^n,\tilde{\bY}^N)\!:=\!\argmin_{i\in[m]:i\neq i_1} \mLD_i(\bT_{\bZ^{n\ba}},\{\bT_{\tilde{\bY}_i^{N\bb}}\}_{i\in[m]})
\end{align}
respectively. 
If the index of the right hand side of \eqref{def:i2} is not unique, we define $i_2(\bZ^n,\tilde{\bY}^N)$ as the smallest index of all $i\in[m]$ such that the value of $\mLD_i(\bT_{\bZ^{n\ba}},\{\bT_{\tilde{\bY}_i^{N\bb}}\}_{i\in[m]})$ is second smallest. % among all the $m$ test statistics $\{\mLD_i(\bT_{\bZ^{n\ba}},\{\bT_{\tilde{\bY}_i^{N\bb}}\}_{i\in[m]})\}_{i\in[m]}$. 
%For simplicity, we use $i_1$ and $i_2$ to denote $i_1(\bZ^n,\tilde{\bY}^N)$ and $i_2(\bZ^n,\tilde{\bY}^N)$ respectively. 
In the achievability proof of Theorem \ref{Thm:rejection exponent}, we make use of the following test 
\begin{align}\label{Eq:Unn test}
&\gamma(Z^n,\tilde{\bY}^N) \nn\\*
&=\left\{
\begin{array}{cl}
\rmH_j & \text{if}~ i_1=j~\text{and}~ \mLD_{i_2}(\bT_{\bZ^{n\ba}},\{\bT_{\tilde{\bY}_i^{N\bb}}\}_{i\in[m]})> \lambda,\\
\rmH_\rmr & \text{if}~\mLD_{i_2}(\bT_{\bZ^{n\ba}},\{\bT_{\tilde{\bY}_i^{N\bb}}\}_{i\in[m]})\leq \lambda.
\end{array}
\right.
\end{align}
In words, we declare that $\rmH_j$ is true if the minimum of $\mLD_i$ occurs when $i=j$ {\em and} the second largest $\mLD_i$ exceeds a certain threshold $\lambda$. The latter condition intuitively indicates that our decision that the true  hypothesis is $\rmH_j$ is made with high enough confidence. If there are at least  two test statistics $\mLD_i$ that are no larger than $\lambda$  (i.e., $|\{  i \in [m]: \mLD_i \le\lambda\}|\ge 2$), our confidence in the quality of the training and test data is low or that our confidence that $X^n$ is generated from one of the distributions that generated $Y_i^N, i\in [m]$ is low, and as such, a rejection event should be declared.
%which is a generalized version of Unnikrishnan's test~\cite{unnikrishnan2015asymptotically}. Thus, the test in \eqref{Eq:Unn test} is asymptotically optimal.

Third, when $m=2$, the test in \eqref{Eq:Unn test} specializes to the test given in \eqref{Eq:Unn test m=2} presented at the top of the next page
\begin{figure*}
\begin{equation}\label{Eq:Unn test m=2}
\gamma(Z^n,\tilde{\bY}^N)=\left\{
\begin{array}{cl}
\rmH_1 & \text{if}~ \mLD_{1}(\bT_{\bZ^{n\ba}},\{\bT_{\tilde{\bY}_i^{N\bb}}\}_{i\in[m]})<\mLD_{2}(\bT_{\bZ^{n\ba}},\{\bT_{\tilde{\bY}_i^{N\bb}}\}_{i\in[m]})~\text{and}~ \mLD_{2}(\bT_{\bZ^{n\ba}},\{\bT_{\tilde{\bY}_i^{N\bb}}\}_{i\in[m]})> \lambda,\\
\rmH_2 & \text{if}~ \mLD_{2}(\bT_{\bZ^{n\ba}},\{\bT_{\tilde{\bY}_i^{N\bb}}\}_{i\in[m]})<\mLD_{1}(\bT_{\bZ^{n\ba}},\{\bT_{\tilde{\bY}_i^{N\bb}}\}_{i\in[m]})~\text{and}~ \mLD_{1}(\bT_{\bZ^{n\ba}},\{\bT_{\tilde{\bY}_i^{N\bb}}\}_{i\in[m]})> \lambda,\\
\rmH_\rmr & \text{if}~\mLD_{1}(\bT_{\bZ^{n\ba}},\{\bT_{\tilde{\bY}_i^{N\bb}}\}_{i\in[m]})\leq \lambda ~\text{and}~ \mLD_{2}(\bT_{\bZ^{n\ba}},\{\bT_{\tilde{\bY}_i^{N\bb}}\}_{i\in[m]})\leq \lambda;
\end{array}
\right.
\end{equation} \hrulefill
\end{figure*}
and the type-$j$ rejection exponent  in~\eqref{Eq:rejection exponent} simplifies to
\begin{align}
&\lim\limits_{n\to\infty}E_j^*(n,\alpha,\bP,\lambda|\ba,\bb,\calW)  \nn\\*
&=\min_{\substack{(\bQ,\tilde{\bQ}_1,\tilde{\bQ}_2)\\
		\in\calQ_{\lambda,1,2}(\alpha,\ba,\bb,\calW)}}\sum_{k\in[K]}\big(a_kD(Q_k\| P_jW_k)\nn\\*
		&\qquad+\alpha b_k D(\tilQ_{1,k}\|P_1 W_k)+\alpha b_k D(\tilQ_{2,k}\|P_2 W_k) \big).
\end{align}
Note that here we consider binary classification with rejection, which is in contrast to the case of binary classification without a rejection option (cf.~\eqref{Eq:thm1 exponent} and \eqref{Eq:test}).

%Compared to the test in \eqref{Eq:test} and the type-II error exponent in \eqref{Eq:thm1 exponent}, here we consider binary distributed detection with rejection option.}
%make use of \emph{both} the training sequences  $(\tilY_1^N, \tilY_2^N)$ in the test. In the binary distributed detection problem without rejection, we accept  $\rmH_2$ if and only if we reject  $\rmH_1$. %Recall that the test in the binary distributed detection problem without rejection depends {\em only} on the (processed version of the) first training sequence $\tilY_1^N$. However, in this section, we enable the rejection option and thus the optimal test needs to depend on {\em both} training sequences $(\tilY_1^N, \tilY_2^N)$.

Fourthly, let us compare the exponents obtained for $m=2$ in Theorems \ref{Thm:exponent} and \ref{Thm:rejection exponent}. For the binary distributed detection problem without rejection (Theorem~\ref{Thm:exponent}), the acceptance region for hypothesis $\rmH_1$ is $\calQ_{\lambda}(\alpha,\ba,\bb,\calW)$ (cf.~\eqref{def:calQlambda}). In this section, when $m=2$, the rejection region is $\calQ_{\lambda,1,2}(\alpha,\ba,\bb,\calW)$ (cf.~\eqref{def:mary calQ}). For any  $(\bQ,\tilde{\bQ}_1,\tilde{\bQ}_2)\in\calQ_{\lambda,1,2}(\alpha,\ba,\bb,\calW)$ , we have
\begin{align}
 \mLD_1(\bQ,\tilde{\bQ}_1,\tilde{\bQ}_2)\leq\lambda,
\end{align}
which implies that $(\bQ,\tilde{\bQ}_1,\tilde{\bQ}_2)\in\calQ_{\lambda}(\alpha,\ba,\bb,\calW)$.  With this observation, we see that %Since $\sum_{k\in[K]}\alpha b_kD(\tilQ_{2,k}\|P_2W_k)\geq 0$ for any $\tilde{\bQ}_2\in\calP([L])^K$, we have
 \begin{align}
&\lim_{n\to\infty}E^*(n,\alpha,P_1,P_2,\lambda|\ba,\bb,V,\calW)\nn\\*
&=\min_{\substack{(\bQ,\tilde{\bQ}_1,\tilde{\bQ}_2) 
		\\\in \calQ_{\lambda}(\alpha,\ba,\bb,\calW)}}\sum_{k\in[K]}\big(a_kD(Q_k\| P_2W_k) \nn\\*
		&\qquad+\alpha b_k D(\tilQ_{1,k}\|P_1 W_k)+\alpha b_k D(\tilQ_{2,k}\|P_2 W_k)\big)\\
&\leq \min_{\substack{(\bQ,\tilde{\bQ}_1,\tilde{\bQ}_2)\\
		\in\calQ_{\lambda,1,2}(\alpha,\ba,\bb,\calW)}}\sum_{k\in[K]}\big(a_kD(Q_k\| P_2W_k)\nn\\*
		&\qquad+\alpha b_k D(\tilQ_{1,k}\|P_1 W_k)+\alpha b_k D(\tilQ_{2,k}\|P_2 W_k) \big)\\
&=\lim\limits_{n\to\infty}E_2^*(n,\alpha,\bP,\lambda|\ba,\bb,\calW),
\end{align}
which indicates that the type-II rejection exponent in \eqref{Eq:rejection exponent} is not smaller than the type-II error exponent in \eqref{Eq:thm1 exponent} when restricted to the binary setting.
%The intuition is that for any fixed $(\bQ,\tilde{\bQ}_1,\tilde{\bQ}_2)\in\Lambda_\rmr$, we can claim that $(\bQ,\tilde{\bQ}_1)\in \Omega_1$, and thus the type-II rejection probability, i.e., $\bbP_2\{\Lambda_r \}$, is smaller than the type-II error probability in the binary case, i.e., $\bbP_2\{\Omega_1 \}$.
 The rough intuition here  is that for the same $\lambda$, if it happens that  the optimal test for binary distributed detection {\em with} rejection in~\eqref{Eq:Unn test m=2} decides on rejecting the two hypotheses, this implies that the optimal test for binary distributed detection  {\em without} rejection in~\eqref{Eq:test} necessarily declares that $\rmH_1$ is true, thus resulting in a  type-II error. The reverse implication, however, is not true.
%The intuition is that the optimal test for binary distributed detection with rejection decides rejection implies that the optimal test for binary distributed detection makes an type-II error.

Finally, if we let $K=1$ and consider all channels to be deterministic, the test in \eqref{Eq:Unn test m=2} reduces to one presented in \eqref{Eq:Unn test m=2 GJS} at the top of the next page
\begin{figure*}
\begin{equation}\label{Eq:Unn test m=2 GJS}
\gamma(Z^n,\tilde{\bY}^N)=\left\{
\begin{array}{cl}
\rmH_j & \text{if}~ \mathrm{GJS}(T_{\tilY_j^N},T_{Z^n},\alpha)<\mathrm{GJS}(T_{\tilY_i^N},T_{Z^n},\alpha)~\text{and}~ \mathrm{GJS}(T_{\tilY_i^N},T_{Z^n},\alpha)> \lambda, ~\forall i\neq j,\\
\rmH_\rmr & \text{if}~\mathrm{GJS}(T_{\tilY_l^N},T_{Z^n},\alpha)\leq \lambda ~\text{and}~ \mathrm{GJS}(T_{\tilY_i^N},T_{Z^n},\alpha)\leq\lambda, ~\exists (i,l)\in[m]^2:i\neq l.
\end{array}
\right.
\end{equation} \hrulefill
\end{figure*}
For the $m$-ary hypotheses testing problem with rejection in \cite{gutman1989asymptotically}, Gutman used $\gamma_m^{\mathrm{Gut}}(Z^n,\tilde{\bY}^N)$ presented at \eqref{eqn:gut_test} on the next page
\begin{figure*}
\begin{equation}
\gamma_m^{\mathrm{Gut}}(Z^n,\tilde{\bY}^N)=\left\{
\begin{array}{cl}
\rmH_1 & \text{if}~ \mathrm{GJS}(T_{\tilY_i^N},T_{Z^n},\alpha)> \lambda,~\forall i\geq 2,\\
\rmH_j & \text{if}~ \mathrm{GJS}(T_{\tilY_j^N},T_{Z^n},\alpha)\leq \lambda ~\text{and}~ \mathrm{GJS}(T_{\tilY_i^N},T_{Z^n},\alpha)>\lambda, ~\forall i\neq j,~ j\geq 2\\
\rmH_\rmr & \text{if}~\mathrm{GJS}(T_{\tilY_l^N},T_{Z^n},\alpha)\leq \lambda ~\text{and}~ \mathrm{GJS}(T_{\tilY_i^N},T_{Z^n},\alpha)\leq\lambda, ~\exists (i,l)\in[m]^2:i\neq l.
\end{array}
\right. \label{eqn:gut_test}
\end{equation} \hrulefill
\end{figure*}
It can be seen that the rejection regions for both the tests are the same. However, the acceptance regions for $\gamma_m^{\mathrm{Gut}}$ are not symmetric for different hypotheses. In contrast, Unnikrishnan's test in \eqref{Eq:Unn test m=2 GJS} is symmetric in the $m$ hypotheses. Thus, it is more convenient to use  the generalized Unnikrishnan's test $\gamma$ in \eqref{Eq:Unn test} to analyze the error and rejection exponents.

%{\color{blue}
%Finally, we remark that using the definition of $\mLD_{j}$ in \eqref{Eq:mary LD} for $m=2$ and $j\in[2]$, we can actually derive the optimal type-II error exponent without Assumption 1. Without Assumption 1, assuming that we use the same set of channels to pre-process both training sequences $(Y_1^N,Y_2^N)$, we find that the following test achieves the optimal type-II error exponent:
%\begin{align}
%\gamma(Z^n,\tilY_1^N,\tilY_2^N)=\begin{cases}
%\rmH_1 & \mLD_1(\bT_{\bZ^{n\ba}},\{\bT_{\tilde{\bY}_i^{N\bb}}\}_{i\in[2]})\leq \lambda,\\
%\rmH_2 & \text{otherwise}.
%\end{cases}
%\end{align}
%To present the result, we need the following set
%\begin{align}
%\tilde{\calQ}_{\lambda}(\alpha,\ba,\bb,\calW)
%:=\bigg\{(\bQ,\tilde{\bQ}_1,\tilde{\bQ}_2)\in \calP([L])^{3K}:\mLD_1(\bQ,\tilde{\bQ}_1,\tilde{\bQ}_2)\leq \lambda\bigg\}.
%\end{align}
%Then the optimal type-II error exponent without Assumption 1 states as follows. Given any $(\lambda,\alpha)\in\mathbb{R}_+^2$ and any tuple of target distributions $(P_1,P_2)\in\calP(\calX)^2$, the optimal type-II error exponent satisfies
%\begin{align}
%\nn&\lim\limits_{n\to\infty}E^*(n,\alpha,P_1,P_2,\lambda|\ba,\bb,\calW)\\
%&=\min_{\substack{(\bQ,\tilde{\bQ}_1,\tilde{\bQ}_2)\\\in\tilde{\calQ}_{\lambda}(\alpha,\ba,\bb,\calW)}}\sum_{k\in[K]}\Big(a_kD(Q_k\| P_2W_k)+\sum_{i\in[2]}\alpha b_k D(\tilQ_{i,k}\|P_i W_k) \Big).
%\end{align}
%The proof of the above result is similar to that of Theorem \ref{Thm:exponent} and is thus omitted.
%}

\subsection{Further discussions on $(\ba,\bb)$}\label{Subsec:mary discussion}
 We have the following corollary of Theorem \ref{Thm:rejection exponent}. 
\begin{corollary}\label{Lem:mary alpha infty}
For each $j\in[m]$, the type-$j$ rejection exponent satisfies
	\begin{align}
	\label{eqn:mart_alpha_infty}
	&\lim\limits_{\alpha\to\infty}\lim\limits_{n\to\infty}E_j^*(n,\alpha,\bP,\lambda|\ba,\bb,\calW) \nn\\*
	&=\min_{(i,l)\in[m]^2:i\neq l}\min_{\substack{\bQ\in\calP([L])^K:\\
			\kappa(\bQ,P_i|\ba,\bb,\calW)\leq \lambda,\\
			\kappa(\bQ,P_l|\ba,\bb,\calW) \leq \lambda}}\sum_{k\in[K]}a_kD(Q_k\|P_jW_k),
	\end{align}
	where $\kappa(\bQ,P|\ba,\bb,\calW)$ was defined in \eqref{def:kappa}. 
\end{corollary}
The proof of Corollary \ref{Lem:mary alpha infty} is similar to that of Lemma \ref{extreme:f} and hence is omitted. %and is provided in Appendix \ref{proof: lem mary alpha infty} for completeness. 

We remark that when $\lambda$ is smaller  than a certain threshold   $\lambda_0=\lambda_0( \bP,\ba,\bb)$,  the   type-$j$ rejection exponent in~\eqref{eqn:mart_alpha_infty} is infinite. This is because if $\lambda\le\lambda_0$, the two constraint sets defined by $\kappa(\bQ,P_i|\ba,\bb,\calW)\le\lambda$ and $\kappa(\bQ,P_l|\ba,\bb,\calW)\le\lambda$ are disjoint for all distinct pairs of $i$ and $l$, and thus, the minimization in \eqref{eqn:mart_alpha_infty} is infeasible.

Let $f_{\infty,j}(\ba,\bb,\lambda)$ denote the right-hand-side of \eqref{eqn:mart_alpha_infty}.
When $\lambda$ is chosen such that $f_{\infty,j}(\ba,\bb,\lambda)<\infty$ for all $(\ba,\bb)\in\calP([K])^2$, we have the following corollary concerning the optimizers of $f_{\infty,j}(\ba,\bb,\lambda)$ when there are only two hypotheses, i.e., $m=2$.

%\begin{corollary}\label{coro:optimizer m>2}
%	For the $m$-ary distributed detection problem with rejection, given each $j\in[m]$, we have
%	\begin{equation}
%		\sup_{(\ba,\bb)\in\calP([K])^2}f_{\infty,j}(\ba,\bb,\lambda)=\sup_{(\ba,\bb)\in\calP([K])^2:\calJ(\ba,\bb)=\emptyset} \min_{(i,l)\in[m]^2:i\neq l} \min_{\substack{\bQ\in\calP([L])^K:\\
%				\sum_{k\in[K]}a_kD(Q_k\|P_iW_k)\leq \lambda,\\
%				\sum_{k\in[K]}a_kD(Q_k\|P_lW_k) \leq \lambda}}\sum_{k\in[K]}a_kD(Q_k\|P_jW_k).
%	\end{equation}
%\end{corollary}
%The proof of Corollary \ref{coro:optimizer m>2} is similar to that of Corollary \ref{Coro:alpha extreme} and hence is omitted. Specifically, when there are only two hypotheses, i.e., $m=2$, we have the following corollary concerning the optimizers of $f_{\infty,j}(\ba,\bb,\lambda)$.

\begin{corollary}
\label{optimizer:mary}
For the binary distributed detection problem with rejection, given each $j\in[2]$, we have
\begin{align}
\sup_{(\ba,\bb)\in\calP([K])^2} f_{\infty,j}(\ba,\bb,\lambda )=\max_{k\in[K]}f_{\infty,j}(\be_k,\be_k,\lambda),
\end{align}
and thus the optimizers $(\ba^*,\bb^*)$ are deterministic and satisfy   $\ba^*=\bb^*$.
\end{corollary}
The proof of Corollary \ref{optimizer:mary} is similar to that of Corollary \ref{Coro:alpha extreme} and hence  is omitted. % is given in Appendix \ref{proof:optimizer:mary}.

 Corollary \ref{optimizer:mary} implies that when the length of the training sequences are much longer than the length of test sequence, it is optimal to use identical local decision rules at each sensor to pre-process the training sequences. % in the case of  and compress the source sequence for binary distributed detection with the rejection option. 
 It is natural to wonder whether there is a generalization of Corollary~\ref{optimizer:mary} for larger $m$. Numerical optimization of the rejection exponent in~\eqref{eqn:mart_alpha_infty} over $(\ba,\bb) \in\calP([K])^2$ shows that when   $m = 3$ and $K=3$, in general, it is optimal to use all $3$ local decision rules or channels to randomize the test and training sequences. When $K=4$, however, it is optimal to use all $4$ channels in general. This differs  from the main finding  in Tsitsiklis' paper~\cite{tsitsiklis1988decentralized}, in which $\frac{1}{2}m(m-1)$ local decision rules suffice to achieve optimality. If the result were analogous, one would expect that for any $K$, only $\frac{1}{2}\cdot 3\cdot(3-1)=3$ local decision rules suffice. This difference can be intuitively explained as follows. With the rejection option, the fusion center needs to partition the sample space into more regions compared to the case without rejection. Roughly speaking, this means that the fusion center needs more information or diversity from the training and test samples to attain optimality. Hence, more channels or local decision rules (compared to~\cite{tsitsiklis1988decentralized}) are needed.

\section{Conclusion and Future Work} \label{sec:concl}

This work has taken a first step at considering the distributed detection problem \`a la Tsitsiklis~\cite{tsitsiklis1988decentralized} when the underlying distributions are unknown but in place of them, we have noisy training samples. We adopted the Gutman formulation~\cite{gutman1989asymptotically} in~\eqref{eqn:E_star}  and derived asymptotically optimal exponents for the binary and $m$-ary cases with and without rejection. While results as conclusive as those in Tsitsiklis' paper~\cite{tsitsiklis1988decentralized} were not obtained, we have several important contributions, including the identification of optimal tests and the conclusion that in the binary case (with and without rejection) and when the number of training samples far exceeds test samples, {\em one} decision rule suffices for achieving the optimal error exponent. 

In the future, one can consider  the following avenues for future work. First, a resolution of Conjecture~\ref{conj:det} would be desirable as it would allow us to parallel the main results in~\cite{tsitsiklis1988decentralized} for arbitrary and finite $\alpha\in\bbR_+$. Second, we can consider deriving second-order asymptotic results in the spirit of Zhou, Tan, and Motani~\cite{zhou2018second}. This would shed further insights into the finite-length behavior of the proposed tests. Finally, it would be fruitful to study the statistical learning versions of other distributed detection formulations, e.g., the anonymous heterogeneous version proposed by Chen and Wang~\cite{chen2019anonymous}.

\appendix
\subsection{Proof of Theorem \ref{Thm:exponent}}\label{Sec:proof of thm1}
 Recall the definitions of $\bT_{\bz^{n\ba}}$, $\bT_{\tilde{\by}_1^{N\bb}}$ and $\bT_{\tilde{\by}_2^{N\bb}}$ in Section \ref{Subsec:binary result}.
 Define the following set of types
\begin{align}
&\calT_{{\lambda}}(\alpha,\ba,\bb,\calW) \nn\\*
&:=\bigg\{(\bT,\tilde{\bT}_1,\tilde{\bT}_2)\in \prod_{k\in[K]}\big(\calP_{na_k}([L])\times\calP_{Nb_k}([L])^2\big):\nn\\*
&\min_{(\tilP,P)\in\calP(\calX)^2}\mathrm{LD}(\bT,\tilde{\bT}_1,\tilde{\bT}_2,\tilP,\tilP,P|\alpha,\ba,\bb,\calW)\!\leq\! \lambda \bigg\}\label{def:calTlambda},
\end{align}
and the following set of sequences
 \begin{align}
&\calL_{ {\lambda}}(\alpha,\ba,\bb,\calW):=\Big\{(\bz^{n\ba},\by_1^{N\bb},\by_2^{N\bb})\in[L]^{n+N}:\nn\\*
&\qquad(\bT_{\bz^{n\ba}},\bT_{\by_1^{N\bb}},\bT_{\by_2^{N\bb}})\in \calT_{ {\lambda}}(\alpha,\ba,\bb,\calW)   \Big\}.
\end{align}
Note that  $\calT_{{\lambda}}(\alpha,\ba,\bb,\calW)$ is the set $\calQ_{\lambda}(\alpha,\ba,\bb,\calW)$ (defined in~\eqref{def:calTlambda}) but restricted to types.
%Finally, throughout this proof, we assume $\calW$ satisfies Assumption \ref{assump: W}.

\subsubsection{Achievability}

In the achievability part, given any pair $(\ba,\bb)\in\calP_n([K])\times\calP_{\alpha n}([K])$, we assume that the test $\gamma$ at the fusion center is given by \eqref{Eq:test}, but $\lambda$ is replaced by $\tilde{\lambda}=\lambda+\frac{c_n}{n}$, where $c_n:=\sum_{k=1}^{K}(L\log(na_k+1)+2L\log(\alpha nb_k+1)) = O(\log n)$. Then for all pairs   $(\tilP_1,\tilP_2)\in \calP(\calX)^2$, the type-I error probability can be upper bounded as follows:
\begin{align}
\label{Eq:beta_1 upper bound begin}
\nn&\beta_1(\gamma,\tilP_1,\tilP_2)\\*
&=\bbP_1\{\gamma(Z^n,\tilY_1^N,\tilY_2^N)\neq \rmH_1\}\\
&=\sum_{\begin{subarray}{c}
	(\bz^{n\ba},\tilde{\by}_1^{N\bb},\tilde{\by}_2^{N\bb})\\
	\notin\calL_{\tilde{\lambda}}(\alpha,\ba,\bb,\calW)
	\end{subarray}}\prod_{k=1}^{K}\Big((\tilP_1W_k)^{na_k}(z^{na_k})  \nn\\*
	&\qquad\times (\tilP_1W_k)^{Nb_k}(\tily_1^{Nb_k})(\tilP_2W_k)^{Nb_k}(\tily_2^{Nb_k})\Big)\\
&=\sum_{\begin{subarray}{c}
	(\bz^{n\ba},\tilde{\by}_1^{N\bb},\tilde{\by}_2^{N\bb})\\
	\notin\calL_{\tilde{\lambda}}(\alpha,\ba,\bb,\calW)
	\end{subarray}}\exp\bigg\{\sum_{k=1}^{K}\bigg(\sum_{i: h(i)=k}\log (\tilP_1W_k)(z_i)\nn\\*
	&\qquad+\sum_{i: g(i)=k}\log (\tilP_1W_k)(\tily_{1,i})\!+\!\sum_{i: g(i)=k}\log (\tilP_2W_k)(\tily_{2,i})\bigg) \bigg\}\\
&\label{Eq:nD+ND}\leq\sum_{\begin{subarray}{c}
	(\bT_{\bz^{n\ba}},\bT_{\tilde{\by}_1^{N\bb}},\bT_{\tilde{\by}_2^{N\bb}})\\
	\notin\calT_{\tilde{\lambda}}(\alpha,\ba,\bb,\calW)
	\end{subarray}}\!\!\!\!\exp\bigg\{\!-n\!\min_{ \substack{(\tilP_1,\tilP_2)\in \\ \calP(\calX)^2}} \sum_{k=1}^{K}\Big(a_k D(T_{z^{na_k}}\|\tilP_1W_k)\nn\\*
	&\qquad\!+\!\alpha b_k D(T_{\tily_1^{Nb_k}}\|\tilP_1W_k)\!+\!\alpha b_k D(T_{\tily_2^{Nb_k}}\|\tilP_2W_k)\Big)\bigg\}\\
&< \exp(-n\tilde{\lambda})\prod_{k=1}^{K}|\mathcal{P}_{na_k}([L])|\times\big(|\mathcal{P}_{Nb_k}([L])|\big)^2 \label{eqn:useT}\\
&\leq \exp\left(-n\Big(\tilde{\lambda}-\frac{c_n}{n}\Big)\right)\label{Eq:beta_1 upper bound end}=\exp(-n\lambda),
\end{align}
where \eqref{eqn:useT} follows from the definition of $\calT_{{\lambda}}(\alpha,\ba,\bb,\calW)$ in~\eqref{def:calTlambda}.

For any $(P_1,P_2)\in\mathcal{P}(\mathcal{X})^2$, the type-II error probability can be upper bounded as follows: \begin{align}
\label{Eq:beta_2 upper bound begin}
\nn&\beta_2(\gamma,P_1,P_2)\\*
&=\sum_{\begin{subarray}{c}
	(\bz^{n\ba},\tilde{\by}_1^{N\bb},\tilde{\by}_2^{N\bb})\\
	\in\calL_{\tilde{\lambda}}(\alpha,\ba,\bb,\calW)
	\end{subarray}}\prod_{k=1}^{K}\Big((P_2W_k)^{na_k}(z^{na_k}) \nn\\*
	&\qquad\times  (P_1W_k)^{Nb_k}(\tily_1^{Nb_k})(P_2W_k)^{Nb_k}(\tily_2^{Nb_k})\Big)\\
&\leq \sum_{\begin{subarray}{c}
	(\bT_{\bz^{n\ba}},\bT_{\tilde{\by}_1^{N\bb}},\bT_{\tilde{\by}_2^{N\bb}})\\
	\in\calT_{\tilde{\lambda}}(\alpha,\ba,\bb,\calW)
	\end{subarray}}\exp\bigg\{-n\sum_{k=1}^{K}\Big(a_kD(T_{z^{na_k}}\|P_2W_k) \nn\\*
	&\qquad+\alpha b_k D(T_{\tily_1^{Nb_k}}\|P_1W_k)+\alpha b_k D(T_{\tily_2^{Nb_k}}\|P_2W_k) \Big)\bigg\} \\
&\dotleq \exp\bigg\{-n \min_{\begin{subarray}{c}
	(\bT_{\bz^{n\ba}},\bT_{\tilde{\by}_1^{N\bb}},\bT_{\tilde{\by}_2^{N\bb}})\\
	\in\calT_{\tilde{\lambda}}(\alpha,\ba,\bb,\calW)
	\end{subarray}}\sum_{k=1}^{K}\Big(a_k D(T_{z^{na_k}}\|P_2W_k) \nn\\*
	&\qquad+\alpha b_k D(T_{\tily_1^{Nb_k}}\|P_1W_k) +\alpha b_k D(T_{\tily_2^{Nb_k}}\|P_2W_k) \Big)   \bigg\}\\
&\label{Eq:beta_2 upper bound end}\leq\exp\bigg\{-n\min_{\substack{(\bQ,\tilde{\bQ}_1,\tilde{\bQ}_2)\\\in\calQ_{\tilde{\lambda}}(\alpha,\ba,\bb,\calW)}}\sum_{k=1}^{K}\Big(a_k D(Q_k\|P_2W_k)\nn\\*
&\qquad+\alpha b_k D(\tilQ_{1,k}\|P_1W_k)+\alpha b_k D(\tilQ_{2,k}\|P_2W_k)\Big) \bigg\}.
\end{align}

Thus, using the definition of $\mathrm{LD}$ in~\eqref{def:LD} and the result in~\eqref{Eq:beta_2 upper bound end}, we have the following lower bound on the type-II error exponent:
\begin{align}
	&\liminf\limits_{n\rightarrow\infty}\frac{1}{n}\log\frac{1}{\beta_2(\gamma,P_1,P_2)}\nn\\*
	&\geq \min_{\substack{
		(\bQ,\tilde{\bQ}_1,\tilde{\bQ}_2)
		\\ \in\calQ_{\lambda}(\alpha,\ba,\bb,\calW)
		} }\mathrm{LD}(\bQ,\tilde{\bQ}_1,\tilde{\bQ}_2,P_2,P_1,P_2|\alpha,\ba,\bb,\calW).\label{Eq:achievability}
\end{align}

\subsubsection{Converse}
Our converse proof proceeds by showing (i) type-based tests (i.e., tests $\gamma$ that depend only on the types or partial types of the data $(Z^n,\tilY_1^N,\tilY_2^N)$) are almost optimal, (ii) the test in \eqref{Eq:test} is an asymptotically optimal type-based test.

The following lemma extends that of~\cite[Lemma 7]{zhou2018second}.
\begin{lemma}\label{Lem:extended from Lin}
	For any deterministic test $\gamma$,   $\eta\in[0,1]$, $(P_1, P_2)\in \calP(\calX)^2$ and $(\ba,\bb)\in\calP_n([K])\times\calP_{\alpha n}([K])$,
	 we can construct a type-based test $\gamma^\rmT$ such that
	 \begin{align}
	 	\beta_1(\gamma,P_1,P_2 ) &\geq \eta \beta_1(\gamma^\rmT,P_1,P_2),\\*
	 	\beta_2(\gamma,P_1,P_2)&\geq (1-\eta) \beta_2(\gamma^\rmT,P_1,P_2).
	 \end{align}
\end{lemma}
\begin{proof}[Proof of Lemma \ref{Lem:extended from Lin}]
	For   $h$ and $g$ with proportions $\ba$ and $\bb$ respectively and any  $(x^n,y_1^N,y_2^N)$, let $(Z^n,\tilY_1^N,\tilY_2^N)\sim\big(\{W_{h(i)}(\cdot|x_i)\}_{i\in[n]},  \{W_{g(i)}(\cdot|y_{1,i})\}_{i\in[N]},\{W_{g(i)}(\cdot|y_{2,i})\}_{i\in[N]}\big)$. 
%Let $\calP_{n\ba+N\bb+N}([L])$ denote the set $\big(\prod_{k\in[K]}\calP_{na_k}([L])\times\calP_{Nb_k}([L])\big)\times\calP_{N}([L])$ and let $\mathbf{Q}=(Q_{1,1},\ldots,Q_{1,K},Q_{2,1},\ldots,Q_{2,K},Q_3)\in\calP_{n\ba+N\bb+N}([L])$. For any $\mathbf{Q}$, we use $\mathcal{T}_{\mathbf{Q}}^{n+2N}$ to denote the set of sequence triples $(Z^n,\tilY_1^N,\tilY_2^N)$ such that for all $ k$, $Z^{na_k}\in \mathcal{T}_{Q_{1,k}}^{na_k}$,  $\tilY_1^{Nb_k}\in \mathcal{T}_{Q_{2,k}}^{Nb_k}$ and $\tilY_2^N \in \mathcal{T}_{Q_3}^N$.
 Let $\calP_{n\ba+2N\bb}([L])$ denote the set $\big(\prod_{k\in[K]}\calP_{na_k}([L])\times\calP_{Nb_k}([L])^2\big)$ and let $\mathbf{Q}=(Q_{1,1},\ldots,Q_{1,K},Q_{2,1},\ldots,Q_{2,K},Q_{3,1},\ldots,Q_{3,K})\in\calP_{n\ba+2N\bb}([L])$. For any $\mathbf{Q}$, we use $\mathcal{\tilT}_{\mathbf{Q}}^{n+2N}$ to denote the set of sequence triples $(Z^n,\tilY_1^N,\tilY_2^N)$ such that for all $k\in[K]$, $Z^{na_k}\in \mathcal{\tilT}_{Q_{1,k}}^{na_k}$,  $\tilY_1^{Nb_k}\in \mathcal{\tilT}_{Q_{2,k}}^{Nb_k}$ and $\tilY_2^{Nb_k}\in \mathcal{\tilT}_{Q_{3,k}}^{Nb_k}$.
	
	Given any test $\gamma$, define the following acceptance region:
	 \begin{align}
 		\mathcal{A}(\gamma,\ba,\bb) \!:=\!\{(z^n,\tily_1^N,\tily_2^N):\gamma(z^n,\tily_1^N,\tily_2^N)\!=\!\rmH_1\}.
	\end{align}  
	Fix any $\eta\in [0,1]$. Given any $\mathbf{Q}\in\calP_{n\ba+2N\bb}([L])$, we can then construct the following type-based test $\gamma^\rmT$:
	\begin{itemize}
		\item If an $\eta$ fraction of sequence triples in  $\mathcal{\tilT}_{\mathbf{Q}}^{n+2N}$ favors hypothesis $\rmH_2$, i.e. $|\mathcal{A}^{\mathsf{c}}(\gamma,\ba,\bb)\cap \mathcal{\tilT}_{\mathbf{Q}}^{n+2N}|>\eta |\mathcal{\tilT}_{\mathbf{Q}}^{n+2N}|$, then $\gamma^\rmT(\bQ)=\rmH_2$;
		\item Otherwise, $\gamma^\rmT(\bQ)=\rmH_1$.
	\end{itemize}
	For any $(P_1, P_2)\in \calP(\calX)^2$ and $(\ba,\bb)$, we can relate the error probabilities of the two tests as in \eqref{Eq:beta eta begin}--\eqref{Eq:eta*beta1} and \eqref{eqn:beta2}--\eqref{Eq:beta eta end} at the top of the next page,
	\begin{figure*}
 \begin{align}\label{Eq:beta eta begin}
\beta_1(\gamma,P_1,P_2)&=\mathbb{P}_1\{\gamma(Z^n,\tilY_1^N,\tilY_2^N)=\rmH_2\}\\*
	&=\sum_{\mathbf{Q}\in\calP_{n\ba+2N\bb}([L])} \mathbb{P}_1\{(Z^n,\tilY_1^N,\tilY_2^N)\in\mathcal{A}^{\mathsf{c}}(\gamma,\ba,\bb)\cap \mathcal{\tilT}_{\mathbf{Q}}^{n+2N}\}\\
	&\geq \sum_{\substack{\mathbf{Q}\in \calP_{n\ba+2N\bb}([L]):\\
			|\mathcal{A}^{\mathsf{c}}(\gamma,\ba,\bb)\bigcap \mathcal{\tilT}_{\mathbf{Q}}^{n+2N}|>\eta |\mathcal{\tilT}_{\mathbf{Q}}^{n+2N}|}}  \mathbb{P}_1\{(Z^n,\tilY_1^N,\tilY_2^N)\in  \mathcal{A}^{\mathsf{c}}(\gamma,\ba,\bb)\cap \mathcal{\tilT}_{\mathbf{Q}}^{n+2N}\}\\
	\label{Eq:elements equally likely}&\geq \sum_{\substack{\mathbf{Q}\in \calP_{n\ba+2N\bb}([L]):\\
			|\mathcal{A}^{\mathsf{c}}(\gamma,\ba,\bb)\bigcap \mathcal{\tilT}_{\mathbf{Q}}^{n+2N}|>\eta |\mathcal{\tilT}_{\mathbf{Q}}^{n+2N}|}} \eta\mathbb{P}_1\{(Z^n,\tilY_1^N,\tilY_2^N)\in \mathcal{\tilT}_{\mathbf{Q}}^{n+2N} \} \\
	&= \sum_{\substack{\mathbf{Q}\in \calP_{n\ba+2N\bb}([L]):\\
			|\mathcal{A}^{\mathsf{c}}(\gamma,\ba,\bb)\bigcap \mathcal{\tilT}_{\mathbf{Q}}^{n+2N}|>\eta |\mathcal{\tilT}_{\mathbf{Q}}^{n+2N}|}} \eta\bigg(\prod_{k=1}^{K}\mathbb{P}_1\{Z^{na_k}\in\mathcal{\tilT}_{Q_{1,k}}^{na_{k}} \}\mathbb{P}_1\{\tilY_1^{Nb_k}\in\mathcal{\tilT}_{Q_{2,k}}^{Nb_k} \}\mathbb{P}_1\{\tilY_2^{Nb_k}\in\mathcal{\tilT}_{Q_{3,k}}^{Nb_k}\}  \bigg) \\ 
	\label{Eq:eta*beta1}&\geq\eta \beta_1(\gamma^\rmT,P_1,P_2),
	\end{align}  \hrulefill	\end{figure*}
	where \eqref{Eq:elements equally likely} follows since the elements in $\mathcal{T}_{\mathbf{Q}}^{n+2N}$ are equally likely (under any product distribution) for any $\bQ$. 
		\begin{figure*}
 \begin{align}
\beta_2(\gamma,P_1,P_2)
	&=\mathbb{P}_2\{\gamma(Z^n,\tilY_1^N,\tilY_2^N)=\rmH_1\} \label{eqn:beta2}\\
	&=\sum_{\mathbf{Q}\in \calP_{n\ba+2N\bb}([L])} \mathbb{P}_2\{(Z^n,\tilY_1^N,\tilY_2^N)\in\mathcal{A}(\gamma,\ba,\bb)\cap \mathcal{\tilT}_{\mathbf{Q}}^{n+2N}\}\\
	&= \sum_{\mathbf{Q}\in  \calP_{n\ba+2N\bb}([L])} (\mathbb{P}_2\{(Z^n,\tilY_1^N,\tilY_2^N)\in\mathcal{\tilT}_{\mathbf{Q}}^{n+2N}\}-\mathbb{P}_2\{(Z^n,\tilY_1^N,\tilY_2^N)\in\mathcal{A}^{\mathsf{c}}(\gamma,\ba,\bb)\cap \mathcal{\tilT}_{\mathbf{Q}}^{n+2N}\}) \\
	&\geq \sum_{\substack{\mathbf{Q}\in \calP_{n\ba+2N\bb}([L]):\\
			|\mathcal{A}^{\mathsf{c}}(\gamma,\ba,\bb)\bigcap \mathcal{\tilT}_{\mathbf{Q}}^{n+2N}|\leq\eta |\mathcal{\tilT}_{\mathbf{Q}}^{n+2N}|}} (1-\eta)\mathbb{P}_2\{(Z^n,\tilY_1^N,\tilY_2^N)\in\mathcal{\tilT}_{\mathbf{Q}}^{n+2N}\}\\
	&= \sum_{\substack{\mathbf{Q}\in \calP_{n\ba+2N\bb}([L]):\\
			|\mathcal{A}^{\mathsf{c}}(\gamma,\ba,\bb)\bigcap \mathcal{\tilT}_{\mathbf{Q}}^{n+2N}|\leq\eta |\mathcal{\tilT}_{\mathbf{Q}}^{n+2N}|}} (1-\eta)\bigg(\prod_{k=1}^{K}\mathbb{P}_2\{Z^{na_k}\in\mathcal{\tilT}_{Q_{1,k}}^{na_{k}} \}\mathbb{P}_2\{\tilY_1^{Nb_k}\in\mathcal{\tilT}_{Q_{2,k}}^{Nb_k} \}\mathbb{P}_2\{\tilY_2^{Nb_k}\in\mathcal{\tilT}_{Q_{3,k}}^{Nb_k}\}\bigg)\label{Eq:Aset_wrong}\\*
	&\geq(1-\eta) \beta_2(\gamma^\rmT,P_1,P_2).\label{Eq:beta eta end}
	\end{align}  \hrulefill	\end{figure*}
\end{proof} 

Let $\delta_n:=\frac{1}{n}\big(L\sum_{k\in[K]}(\log(na_k+1)+2\log(Nb_k+1))\big)=o(1)$ and fix an arbitrary sequence  $\{\delta'_n\}\subset (0,1)$ to be such that $\lim_{n\to\infty }\delta'_n=0$.
\begin{lemma}\label{Lem:type-based test optimal}
Given any $(\lambda,\alpha)\in \bbR_+^2$ and any $(\ba,\bb)$, for any type-based test $\gamma^\rmT$ satisfying that for all pairs of distributions $(\tilP_1, \tilP_2)\in \calP(\calX)^2$,  
\begin{equation}\label{Eq:beta_1<lambda}
	\beta_1(\gamma^\rmT,\tilP_1,\tilP_2)\leq \exp(-n\lambda),
\end{equation}
we have that for any pair of distributions $(P_1, P_2)\in \calP(\calX)^2$,
 \begin{align}\label{Eq:lemma12 beta2>...}
&\beta_2(\gamma^\rmT,P_1,P_2)\nn\\*
&\!\geq\! \mathbb{P}_2\left\{\!(\bT_{\bZ^{n\ba}},\bT_{\tilde{\bY}_1^{N\bb}},\bT_{\tilde{\bY}_2^{N\bb}})\!\in\!\calT_{\lambda-\delta_n-\delta'_n}(\alpha,\ba,\bb,\calW)\right\}\!.\!\!
\end{align}

\end{lemma}
{\color{black}\begin{proof}[Proof of Lemma \ref{Lem:type-based test optimal}]
    Let $\mathrm{LD}(\bQ,\tilde{\bQ}_1,\tilde{\bQ}_2,\tilP,\tilP,P)=\mathrm{LD}(\bQ,\tilde{\bQ}_1,\tilde{\bQ}_2,\tilP,\tilP,P|\alpha,\ba,\bb,\calW)$.
    In other words, Lemma \ref{Lem:type-based test optimal} claims that for any type-based test $\gamma^\rmT$ satisfying \eqref{Eq:beta_1<lambda}, if any $(\bT_{\bz^{n\ba}},\bT_{\tilde{\by}_1^{N\bb}},\bT_{\tilde{\by}_2^{N\bb}})\in \calP_{n\ba+2N\bb}([L])$ satisfies 
    \begin{equation}
    \begin{aligned}[b]
&    \min_{(\tilP_1,\tilP_2)\in\calP(\calX)^2}\mathrm{LD} (\bT_{\bz^{n\ba}},\bT_{\tilde{\by}_1^{N\bb}},\bT_{\tilde{\by}_2^{N\bb}},\tilP_1,\tilP_1,\tilP_2)+\delta_n\nn\\*
&    \qquad\leq\lambda-\delta'_n,
    \end{aligned} 
    \end{equation}
    then we have $\gamma^\rmT(\bT_{\bz^{n\ba}},\bT_{\tilde{\by}_1^{N\bb}},\bT_{\tilde{\by}_2^{N\bb}})=\rmH_1$.
    
	This claim can be proved by contradiction. Suppose there exists $(\bQ_{\bz^{n\ba}},\bQ_{\tilde{\by}_1^{N\bb}},\bQ_{\tilde{\by}_2^{N\bb}})\in \calP_{n\ba+2N\bb}([L])$ such that 
\begin{align}		
\!\!\min_{\tilP_1,\tilP_2}\mathrm{LD} (\bQ_{\bz^{n\ba}},\bQ_{\tilde{\by}_1^{N\bb}},\bQ_{\tilde{\by}_1^{N\bb}},\tilP_1,\tilP_1,\tilP_2)\!+\!\delta_n\!\leq\!\lambda\!-\!\delta'_n, \label{eqn:LDlamb} 
\end{align}
and  
$\gamma^\rmT(\bQ_{\bz^{n\ba}},\bQ_{\tilde{\by}_1^{N\bb}},\bQ_{\tilde{\by}_2^{N\bb}}) =\mathrm{H}_2$. 
For all $(\tilP_1, \tilP_2)$, we have
	\begin{align}
		\label{Eq:beta1 gammaT lower begin}\nn&\beta_1(\gamma^\rmT,\tilP_1,\tilP_2)\\*
%		&=\sum_{\substack{(\bz^{n\ba},y_1^N,y_2^N): \\ \gamma^\rmT(\bT_{\bz^{n\ba}},\bT_{\tilde{\by}_1^{N\bb}},T_{\tily_2^N})=H_2}}\bigg( \prod_{k=1}^K (\tilP_1W_k)^{na_k}(z^{na_k})\bigg)\tilP_1^N(y_1^N)\tilP_2^N(y_2^N)\\
		&= \sum_{\begin{subarray}{c}
			(\bz^{n\ba},\tilde{\by}_1^{N\bb},\tilde{\by}_2^{N\bb}):\\
			\gamma^\rmT(\bT_{\bz^{n\ba}},\bT_{\tilde{\by}_1^{N\bb}},\bT_{\tilde{\by}_2^{N\bb}})=\rmH_2
			\end{subarray}}
		\bigg(\prod_{k=1}^{K}(\tilP_1W_k)^{na_k}(z^{na_k}) \nn\\*
		&\qquad\qquad\times (\tilP_1W_k)^{Nb_k}(\tily_1^{Nb_k})(\tilP_2W_k)^{Nb_k}(\tily_2^{Nb_k})\bigg)\\
		\nn&=\sum_{\begin{subarray}{c}
			(\bT_{\bz^{n\ba}},\bT_{\tilde{\by}_1^{N\bb}},\bT_{\tilde{\by}_2^{N\bb}}):\\
			\gamma^\rmT(\bT_{\bz^{n\ba}},\bT_{\tilde{\by}_1^{N\bb}},\bT_{\tilde{\by}_2^{N\bb}})=\rmH_2
			\end{subarray}}\bigg(\prod_{k=1}^{K}|\calT_{z^{na_k}}||\calT_{\tily_1^{Nb_k}}||\calT_{\tily_2^{Nb_k}}|\bigg) \nn\\
			&\qquad\times\exp\bigg(-n\sum_{k=1}^{K}a_k\left(D(T_{z^{na_k}}\|\tilP_1W_k)+H(T_{z^{na_k}})\right)\nn\\
		&\qquad  -N\sum_{k=1}^{K}b_k\left(D(T_{\tily_1^{Nb_k}}\|\tilP_1W_k)+H(T_{\tily_1^{Nb_k}})\right) \nn\\
		&\qquad-N\sum_{k=1}^{K}b_k\left(D(T_{\tily_2^{Nb_k}}\|\tilP_2W_k)+H(T_{\tily_2^{Nb_k}})\right)\bigg)\\
		\nn& \geq \sum_{\begin{subarray}{c}
			(\bT_{\bz^{n\ba}},\bT_{\tilde{\by}_1^{N\bb}},\bT_{\tilde{\by}_2^{N\bb}}):\\
			\gamma^\rmT(\bT_{\bz^{n\ba}},\bT_{\tilde{\by}_1^{N\bb}},\bT_{\tilde{\by}_2^{N\bb}})=\rmH_2
			\end{subarray}}\!\!\!\! \exp\bigg(-n\sum_{k=1}^{K}\Big(a_k D(T_{z^{na_k}}\|\tilP_1W_k) \nn\\*
			&\qquad+\alpha b_k D(T_{\tily_1^{Nb_k}}\|\tilP_1W_k)+\alpha b_kD(T_{\tily_2^{Nb_k}}\|\tilP_2W_k)\bigg)\nn\\*
		&\qquad -\sum_{k=1}^{K}\left(L\log(na_k+1)+2L\log(Nb_k+1)\right) \bigg) \label{Eq: lem9 beta_1 size of type class}\\
		&\geq \exp\bigg(-n\bigg(\sum_{k=1}^{K}\Big(a_kD(Q_{z^{na_k}}\|\tilP_1W_k ) \nn\\*
		&\qquad\!+\!\alpha b_k D(Q_{\tily_1^{Nb_k}}\|\tilP_1W_k)\!+\!\alpha b_kD(Q_{\tily_2^{Nb_k}}\|\tilP_2W_k)\Big)\!+\!\delta_n\bigg)\!\bigg), \label{Eq:beta1 gammaT lower end}
	\end{align}
	where~\eqref{Eq: lem9 beta_1 size of type class} follows since $|\calT_{z^{na_k}}|\geq (na_k+1)^{-L} \exp(na_kH(T_{z^{na_k}}))$ (cf.~\cite[Ch.~2]{Csi97}) and similar lower bounds hold for $|\calT_{\tily_1^{Nb_k}}|$ and $|\calT_{\tily_2^{Nb_k}}|$.
%	However, if we choose $\tilP_2$ such that 
%\begin{equation}
%	\tilP_2V =Q_{\tily_2^{N}} \label{eqn:P2}
%	\end{equation}	
%(the existence of such a  $\tilP_2$ follows from Assumption \ref{assump: W}), and
However, if we choose $(\tilP_1,\tilP_2)$ such that
	\begin{align}
		&(\tilP_1,\tilP_2)=(\tilP_1^*,\tilP_2^*) \nn\\*
		&:=\argmin_{(\barP_1,\barP_2)\in\calP(\calX)^2}\mathrm{LD} (\bQ_{\bz^{n\ba}},\bQ_{\tilde{\by}_1^{N\bb}},\bQ_{\tilde{\by}_2^{N\bb}},\barP_1,\barP_1,\barP_2),
	\end{align}
	then we have
	\begin{align}
& \beta_1(\gamma^\rmT,\tilP_1,\tilP_2)\nn\\*
 &\geq \exp\big(\! -n (\mathrm{LD} (\bQ_{\bz^{n\ba}},\bQ_{\tilde{\by}_1^{N\bb}},\bQ_{\tilde{\by}_2^{N\bb}},\tilP^*_1,\tilP^*_1,\tilP_2^*)\! +\! \delta_n \! )\big)\label{eqn:lower_LD}  \\
		&\geq \exp\big(-n(\lambda-\delta'_n)\big)\label{Eq:lem9 beta_1 min LD}\\*
		&> \exp(-n\lambda ) \label{Eq: strictly larger},
	\end{align}
	where \eqref{eqn:lower_LD} follows from the definition of $\mathrm{LD}(\cdot)$ in \eqref{def:LD};~\eqref{Eq:lem9 beta_1 min LD} follows from \eqref{eqn:LDlamb} and the strict inequality in~\eqref{Eq: strictly larger} follows because $\delta'_n>0$ for all $n$.
	The result in~\eqref{Eq: strictly larger} contradicts the assumption that~(\ref{Eq:beta_1<lambda}) is satisfied for all $(\tilP_1, \tilP_2)\in \calP(\calX)^2$. 
\end{proof}}

\begin{corollary}\label{Coro:beta_2>beta_2(GJS)}
	Given any $(\lambda,\alpha)\in \bbR_+^2$ and any $(\ba,\bb)$, for any test $\gamma$ satisfying that for all $(\tilP_1, \tilP_2)\in \calP(\calX)^2$,  
	\begin{equation}\label{Eq:beta_1 < lambda+delta}
		\beta_1(\gamma,\tilP_1, \tilP_2)\leq \exp\Big(-n\Big(\lambda+\frac{\log 2}{n}\Big)\Big),
	\end{equation}
	we have that for any pair of distributions $(P_1, P_2)\in \calP(\calX)^2$,
\begin{align}
		&\beta_2(\gamma,P_1,P_2) \nn\\*
		&\geq \frac{1}{2}\mathbb{P}_2 
		\Big\{\big(\bT_{\bZ^{n\ba}},\bT_{\tilde{\bY}_1^{N\bb}},\bT_{\tilde{\bY}_2^{N\bb}}\big)\in\calT_{\lambda-\delta_n-\delta'_n}(\alpha,\ba,\bb,\calW)\Big\}.
\end{align}
\end{corollary}

Corollary \ref{Coro:beta_2>beta_2(GJS)} can be directly obtained from Lemma \ref{Lem:extended from Lin} and Lemma \ref{Lem:type-based test optimal} by letting $\eta=\frac{1}{2}$. Using Corollary \ref{Coro:beta_2>beta_2(GJS)}, we have 
\begin{align}\label{Eq:beta2 lower begin}
&	 \beta_2(\gamma,P_1,P_2) \nn\\*
	 &\geq\frac{1}{2}\mathbb{P}_2 
	\Big\{\big(\bT_{\bZ^{n\ba}},\bT_{\tilde{\bY}_1^{N\bb}},\bT_{\tilde{\bY}_2^{N\bb}}\big)\in\calT_{\lambda-\delta_n-\delta'_n}(\alpha,\ba,\bb,\calW)\Big\}\\*
	&=\sum_{\begin{subarray}{c}
		(\bT_{\bz^{n\ba}},\bT_{\tilde{\by}_1^{N\bb}},\bT_{\tilde{\by}_2^{N\bb}})\\
		\in\calT_{\lambda-\delta_n-\delta'_n}(\alpha,\ba,\bb,\calW)
		\end{subarray}}\bigg(\prod_{k=1}^{K}|\calT_{z^{na_k}}||\calT_{\tily_1^{Nb_k}}||\calT_{\tily_2^{Nb_k}}|\bigg) \nn\\*
		&\qquad\exp\bigg\{-n\sum_{k=1}^{K}a_k\big(D(T_{z^{na_k}}\|P_2W_k)+H(T_{z^{na_k}})\big) \nonumber\\*
	&\qquad -N\sum_{k=1}^K b_k\big(D(T_{\tily_1^{Nb_k}}\|P_1W_k)+H(T_{\tily_1^{Nb_k}})	\big)\nn\\*
	&\qquad-N\sum_{k=1}^K b_k\big(D(T_{\tily_2^{Nb_k}}\|P_2W_k)+H(T_{\tily_2^{Nb_k}})	\big)-\log 2 \bigg\}\\
	&\geq \sum_{\begin{subarray}{c}
		(\bT_{\bz^{n\ba}},\bT_{\tilde{\by}_1^{N\bb}},\bT_{\tilde{\by}_2^{N\bb}})\\
		\in\calT_{\lambda-\delta_n-\delta'_n}(\alpha,\ba,\bb,\calW)
		\end{subarray}}\exp\bigg\{-n\sum_{k=1}^{K}\bigg(a_k D(T_{z^{na_k}}\|P_2W_k)\nn\\*
		&\qquad+\sum_{i\in[2]}\alpha b_k D(T_{\tily_i^{Nb_k}}\|P_iW_k)+\frac{c_n+\log 2}{n}\bigg)\bigg\}\\
	&\geq \exp\bigg\{-n \min_{\begin{subarray}{c}
		(\bT_{\bz^{n\ba}},\bT_{\tilde{\by}_1^{N\bb}},\bT_{\tilde{\by}_2^{N\bb}})\\
		\in\calT_{\lambda-\delta_n-\delta'_n}(\alpha,\ba,\bb,\calW)
		\end{subarray}}\sum_{k=1}^{K}\bigg(a_k D(T_{z^{na_k}}\|P_2W_k)\nn\\*
		&\qquad+\sum_{i\in[2]}\alpha b_kD(T_{\tily_i^{Nb_k}}\|P_iW_k)+\frac{c_n+\log 2}{n}\bigg)	\bigg\}. \label{Eq:beta2 lower end} 
\end{align}

 Note that the union of the set of types $\cup_{n\in\bbN}\calT_{\lambda}(\alpha,\ba,\bb,\calW)$ (where $\calT_\lambda$ is  defined in~\eqref{def:calTlambda}) is dense in the set of distributions $\calQ_{\lambda}(\alpha,\ba,\bb,\calW)$ (defined in \eqref{def:calQlambda}); this follows from the continuity of $(\bQ,\tilde{\bQ}_1,\tilde{\bQ}_2)\mapsto \min\limits_{(\tilP,P)\in\calP(\calX)^2}\mathrm{LD}(\bQ,\tilde{\bQ}_1,\tilde{\bQ}_2,\tilP,\tilP,P|\alpha,\ba,\bb,\calW)$.  Also $n\rightarrow \infty$, $\delta_n,\delta'_n, \frac{c_n}{n},\frac{\log 2}{n}\to 0$. As a result,  for any test $\gamma$ satisfying \eqref{Eq:beta_1 < lambda+delta}, the type-II error exponent can be upper bounded as follows
\begin{align}\label{Eq:converse}
	&\limsup\limits_{n\rightarrow \infty}\frac{1}{n}\log\frac{1}{\beta_2(\gamma,P_1,P_2)} \nn\\*
&\leq  \min_{\substack{
		(\bQ,\tilde{\bQ}_1,\tilde{\bQ}_2 )\\
		\in \calQ_{\lambda}(\alpha,\ba,\bb,\calW)
}}\mathrm{LD}(\bQ,\tilde{\bQ}_1,\tilde{\bQ}_2,P_2,P_1,P_2|\alpha,\ba,\bb,\calW).
\end{align}
Combining the lower and upper bounds in~\eqref{Eq:achievability} and~\eqref{Eq:converse} respectively, we conclude that the optimal type-II error exponent is given by \eqref{Eq:thm1 exponent} in Theorem \ref{Thm:exponent}, completing the proof.
%\begin{align}
%	\lim_{n\to\infty}E^*(n,\lambda,\alpha,P_1,P_2|\ba,\bb,V,\calW)=\min_{\substack{(\bQ,\tilde{\bQ})\\
%			\in \calQ_{\lambda}(\alpha,\ba,\bb,\calW)}}\mathrm{LD}(\bQ,\tilde{\bQ},P_2,P_1|\alpha,\ba,\bb,\calW).
%\end{align}

\subsection{Proof of Lemma \ref{extreme:f}}
\label{proof:extremef}

Given any vector $\bb\in\calP([K])$ and any distributions $(P_1,P_2)\in\calP(\calX)^2$, define the following set of distributions
 \begin{align}\label{def:tildecalQ}
&\tilde{\calQ}(P_1,P_2|\bb,\calW)
:=\Big\{(\tilde{\bQ}_1,\tilde{\bQ}_2)\in\calP([L])^{2K}\,:\, \forall~k\in[K]\nn\\*
&\quad  b_k\|\tilQ_{1,k}-P_1 W_k\|_{\infty}=b_k\|\tilQ_{2,k}-P_2 W_k\|_{\infty}=0\Big\}.
\end{align} 

Recall  $\mathrm{LD}(\bQ,\tilde{\bQ}_1,\tilde{\bQ}_2,\tilP,\tilP,P |\alpha,\ba,\bb,\calW)$ in \eqref{def:LD} and $f_{\alpha}(\ba,\bb,\lambda)$ in \eqref{def:f}. As $\alpha\to\infty$, the objective function of $f_{\alpha}(\ba,\bb,\lambda)$ tends to infinity unless 
$(\tilde{\bQ}_1,\tilde{\bQ}_2)\in\tilde{\calQ}(P_1,P_2|\bb,\calW)$.% Furthermore, we have the following lemma.

\begin{lemma}\label{Lem:g(infty)}
	For any $\bQ\in\calP([L])^K$ and any $(\tilde{\bQ}_1,\tilde{\bQ}_2)\in\tilde{\calQ}(P_1,P_2|\bb,\calW)$, we have
	\begin{align}
		&\lim\limits_{\alpha\to\infty}\min_{(\tilP,P)\in\calP(\calX)^2}\mathrm{LD}(\bQ,\tilde{\bQ}_1,\tilde{\bQ}_2,\tilP,\tilP,P|\alpha,\ba,\bb,\calW
		)\nn\\*
		&=\min_{\substack{\tilP\in\calP(\calX):\forall~k\in[K],\\b_k\|\tilQ_{1,k}-\tilP W_k\|_{\infty}=0}}\sum_{k\in[K]}a_kD(Q_k\|\tilP W_k).
	\end{align}
\end{lemma}
\begin{proof}
	Given any $\tilde{\bQ}=(\tilQ_1,\ldots,\tilQ_K)\in\calP([L])^K$, let
	\begin{align}
	\calS(\tilde{\bQ},\bb,\calW) :=\{\tilP :b_k\|\tilQ_k-\tilP W_k \|_{\infty}=0, ~\forall k\in[K] \}.
	\end{align}
	For any $(\tilde{\bQ}_1,\tilde{\bQ}_2)\in\tilde{\calQ}(P_1,P_2|\bb,\calW)$ and any $i\in[2]$,
 \begin{align}
&	b_k\|\tilQ_{i,k}-\tilP W_k \|_{\infty}\nn\\*
	&\leq b_k\|\tilQ_{i,k}-P_i W_k \|_{\infty}+b_k\|P_iW_k-\tilP W_k \|_{\infty} \nn\\*
	&=b_k\|P_iW_k-\tilP W_k \|_{\infty},
	\end{align}
	so we have $\calS(\{P_iW_k\}_{k\in[K]},\bb,\calW)\subset \calS(\tilde{\bQ}_i,\bb,\calW)$. Since $\calS(\{P_iW_k\}_{k\in[K]},\bb,\calW)\neq \emptyset$ for $i\in[2]$, we have $\calS(\tilde{\bQ}_i,\bb,\calW)\neq \emptyset$.
	
	Note first that $(\alpha,\tilP,P)\in[0,\infty)\times \calP(\calX)^2 \mapsto \mathrm{LD}(\bQ,\tilde{\bQ}_1,\tilde{\bQ}_2,\tilP,\tilP,P|\alpha,\ba,\bb,\calW)$ is jointly continuous, and $\calP(\calX)$ is compact. As such, the function
	\begin{equation}
	g(\alpha):=\min_{\tilP,P }\mathrm{LD}(\bQ,\tilde{\bQ}_1,\tilde{\bQ}_2,\tilP,\tilP,P|\alpha,\ba,\bb,\calW)
	\end{equation}
	with domain $[0,\infty)$ is continuous in $\alpha$ (cf.~\cite[Lemma 14]{tan2011large}). Let 
	\begin{align}
&	(\tilP^*(\alpha),P^*(\alpha)) \nn\\*
&	:=\argmin_{(\tilP,P)\in\calP(\calX)^2}\mathrm{LD}(\bQ,\tilde{\bQ}_1,\tilde{\bQ}_2,\tilP,\tilP,P|\alpha,\ba,\bb,\calW).
    \end{align}	
Thus, 
\begin{align}
&	\lim\limits_{\alpha\to\infty}\mathrm{LD}(\bQ,\tilde{\bQ}_1,\tilde{\bQ}_2,\tilP^*(\alpha),\tilP^*(\alpha),P^*(\alpha)|\alpha,\ba,\bb,\calW)\nn\\*
&	=\lim\limits_{\alpha\to\infty}g(\alpha)=g(\infty).
\end{align}
 Since $\calS(\tilde{\bQ}_i,\bb,\calW)\neq \emptyset$ where $i\in[2]$ for any $(\tilde{\bQ}_1,\tilde{\bQ}_2)\in\tilde{\calQ}(P_1,P_2|\bb,\calW)$, we have for any $\alpha\in[0,\infty)$
	\begin{align}\label{Eq:g(alpha)<infty}
	g(\alpha)&\leq \min_{\substack{\tilP\in\calS(\tilde{\bQ}_1,\bb,\calW),\\P\in\calS(\tilde{\bQ}_2,\bb,\calW)}}\mathrm{LD}(\bQ,\tilde{\bQ}_1,\tilde{\bQ}_2,\tilP,\tilP,P|\alpha,\ba,\bb,\calW)\nn\\*
	&=\min_{\tilP\in\calS(\tilde{\bQ}_1,\bb,\calW)}\sum_{k\in[K]}a_kD(Q_k\| \tilP W_k)<\infty.
	\end{align}
	If $\tilP^*(\infty),P^*(\infty)\notin \calS(\tilde{\bQ}_1,\bb,\calW)$, then $g(\infty)=\infty$, which violates \eqref{Eq:g(alpha)<infty}. So $\tilP^*(\infty),P^*(\infty)\in \calS(\tilde{\bQ},\bb,\calW)$ and
	\begin{equation}
	g(\infty)=\min_{\tilP\in\calS(\tilde{\bQ}_1,\bb,\calW)}\sum_{k\in[K]}a_kD(Q_k\| \tilP W_k).
	\end{equation}This concludes the proof.
\end{proof}

Thus, for any $\bQ\in\calP([L])^K$ and any $(\tilde{\bQ}_1,\tilde{\bQ}_2)\in\tilde{\calQ}(P_1,P_2|\bb,\calW)$, we have 
\begin{align}
\nn&\lim\limits_{\alpha\to\infty}\min_{(\tilP,P)\in\calP(\calX)^2}\mathrm{LD}(\bQ,\tilde{\bQ}_1,\tilde{\bQ}_2,\tilP,\tilP,P|\alpha,\ba,\bb,\calW)\\*
&=\min_{\tilP\in\calS(\tilde{\bQ}_1,\bb,\calW)}\sum_{k\in[K]}a_kD(Q_k\|\tilP W_k)\label{avoidinf}\\
&=\min_{\substack{\tilP \in\calP(\calX):\forall~k\in[K],\\b_k\|P_1W_k-\tilP W_k\|_{\infty}=0}}\sum_{k\in[K]}a_kD(Q_k\|\tilP W_k)\label{usetilcalQ}\\
&=\min_{\tilP \in\calP(P_1|\bb,\calW)}\sum_{k\in[K]}a_kD(Q_k\|\tilP W_k)\label{usecalPP1}\\
&=\kappa(\bQ,P_1|\ba,\bb,\calW)\label{usekappa},
\end{align}
where \eqref{avoidinf} follows from Lemma \ref{Lem:g(infty)}, \eqref{usetilcalQ} follows since $0\leq b_k\|P_1W_k-\tilP W_k\|_{\infty}\leq b_k\|\tilQ_{1,k}-\tilP W_k\|_{\infty}+b_k\|\tilQ_{1,k}-P_1 W_k\|_{\infty}=0$ for all  $k\in[K]$, \eqref{usecalPP1} follows from the definition of $\calP(P_1|\bb,\calW)$ in \eqref{def:calPP1} and \eqref{usekappa} is   due to the definition of $\kappa(\cdot)$ in~\eqref{def:kappa}.

Combining above results, we have the desired result in Lemma \ref{extreme:f}.

\subsection{Proof of Corollary \ref{Coro:alpha extreme}}
\label{proof:coro:alpha extreme}
From Lemma \ref{extreme:f}, we know that as $\alpha\rightarrow \infty$, given any $\lambda$ and any $(\ba,\bb)\in\calP([K])^2$, $f_{\alpha}(\ba,\bb,\lambda)$ converges to
\begin{align}
f_{\infty}(\ba,\bb,\lambda)
=\min_{\substack{\bQ\in\calP([L])^K:\\
\kappa(\bQ,P_1|\ba,\bb,\calW)\leq \lambda}}\sum_{k\in[K]}a_kD(Q_k\|P_2W_k).
\end{align}
For any $(\ba,\bb)\in\calP([K])^2$, we define 
\begin{align}\label{Eq:calJ(a,b)}
\calJ(\ba,\bb)
:=\{k\in[K]:~a_k>0\mbox{~and~}b_k=0\}.
\end{align}

Recalling the definition of $\kappa(\cdot)$ in \eqref{def:kappa} and noting that $P_1\in\calP(P_1|\bb,\calW)$, given any $(\ba,\bb)$ such that $\calJ(\ba,\bb)\neq \emptyset$, we have that
\begin{align}
f_{\infty}(\ba,\bb,\lambda)
&\leq \min_{\substack{\bQ\in \calP([L])^K:\\
	\sum_{k\in[K]}a_kD(Q_k\|P_1W_k)\leq\lambda	}}\sum_{k\in[K]}a_k D(Q_k\|P_2W_k)\label{differentsupport}.
\end{align}
On the other hand, for any $(\ba,\bb)$ such that $\calJ(\ba,\bb)=\emptyset$, we have that
\begin{align}
f_{\infty}(\ba,\bb,\lambda)
&=\min_{\substack{\bQ\in \calP([L])^K:\\
	\sum_{k\in[K]}a_kD(Q_k\|P_1W_k)\leq\lambda	}}\sum_{k\in[K]}a_k D(Q_k\|P_2W_k)\label{samesupport}.
\end{align}

Note that for any $\ba\in\calP([K])$, there exists $\bb\in\calP([K])$ such that $\calJ(\ba,\bb)=\emptyset$ (e.g., $\mathrm{supp}(\ba)=\mathrm{supp}(\bb)$); on the other hand, there also exists $\bb$ such that $\calJ(\ba,\bb)\neq \emptyset$, e.g., $\mathrm{supp}(\bb)\subset\mathrm{supp}(\ba)$ when $|\mathrm{\supp}(\ba)|\geq 1$ and $\mathrm{supp}(\bb)\cap\mathrm{supp}(\ba)=\emptyset$ when $|\mathrm{\supp}(\ba)|=1$. Thus, combining \eqref{differentsupport} and \eqref{samesupport}, we have that for any $\lambda\in\bbR_+$, 
\begin{align}
& \sup_{(\ba,\bb)\in\calP([K])^2}f_{\infty}(\ba,\bb,\lambda)\nn\\*
 &=\sup_{\ba\in\calP([K])}\sup_{\bb\in\calP([K]):\calJ(\ba,\bb)=\emptyset}f_{\infty}(\ba,\bb,\lambda)\\
&=\sup_{(\ba,\bb)\in\calP([K])^2:\calJ(\ba,\bb)=\emptyset}f_{\infty}(\ba,\bb,\lambda)\label{samesupportconclude}.
\end{align}

For each $k\in[K]$, given $\lambda\in\bbR_+$, let $Q_k^*$ achieve $f_{\infty}(\be_k,\be_k,\lambda)$, i.e.,
\begin{align}
f_{\infty}(\be_k,\be_k,\lambda)
&=\min_{\substack{Q_{k}\in \mathcal{P}([L]):\\
	D(Q_{k}\|P_1W_k)\leq\lambda	}}D(Q_{k}\|P_2W_k)\\*
&=D(Q_k^*\|P_2W_k)\label{Eq:f(ak)}.
\end{align}
Given any $(\ba,\bb)$ such that $\calJ(\ba,\bb)=\emptyset$, from \eqref{Eq:f(ak)}, we know that
\begin{align}
\sum_{k\in[K]}a_kD(Q_k^*\|P_1W_k)\leq \lambda,
\end{align}
and thus
\begin{align}
 f_{\infty}(\ba,\bb,\lambda)&=\min_{\substack{\mathbf{Q}\in \calP([L])^K:\\
	\sum_{k\in[K]}a_k D(Q_{k}\|P_1W_k) \leq\lambda}}\sum_{k\in[K]}a_k D(Q_k\|P_2W_k)\\
&\leq \sum_{k\in[K]}a_k D(Q_k^*\|P_2W_k)\\
&\leq \max_{k\in[K]}D(Q_k^*\|P_2W_k)\\*
&=\max_{k\in[K]}f_{\infty}(\be_k,\be_k,\lambda)\label{todirect}.
\end{align}
Combining \eqref{samesupportconclude} and \eqref{todirect}, we have that
\begin{align}
\sup_{(\ba,\bb)\in\calP([K])^2}f_{\infty}(\ba,\bb,\lambda)
&\leq\max_{k\in[K]}f_{\infty}(\be_k,\be_k,\lambda)\label{direct}.
\end{align}
On the other hand, it is easy to verify that
\begin{align}
\sup_{(\ba,\bb)\in\calP([K])^2}f_{\infty}(\ba,\bb,\lambda)
&\geq\max_{k\in[K]}f_{\infty}(\be_k,\be_k,\lambda)\label{converse}.
\end{align}
The proof of Corollary \ref{Coro:alpha extreme} is completed by combining~\eqref{direct} and~\eqref{converse}.

\subsection{Numerical Evaluations of $f_{\alpha}(\ba,\bb,\lambda)$ when $K=3$}\label{Sec:figures}

\begin{figure*}[t]
	\centering
	\subfigure[$\bb=(0.7,0.1,0.2)$]{
		\begin{minipage}[t]{0.48\linewidth}
			\centering
			\includegraphics[scale=0.45]{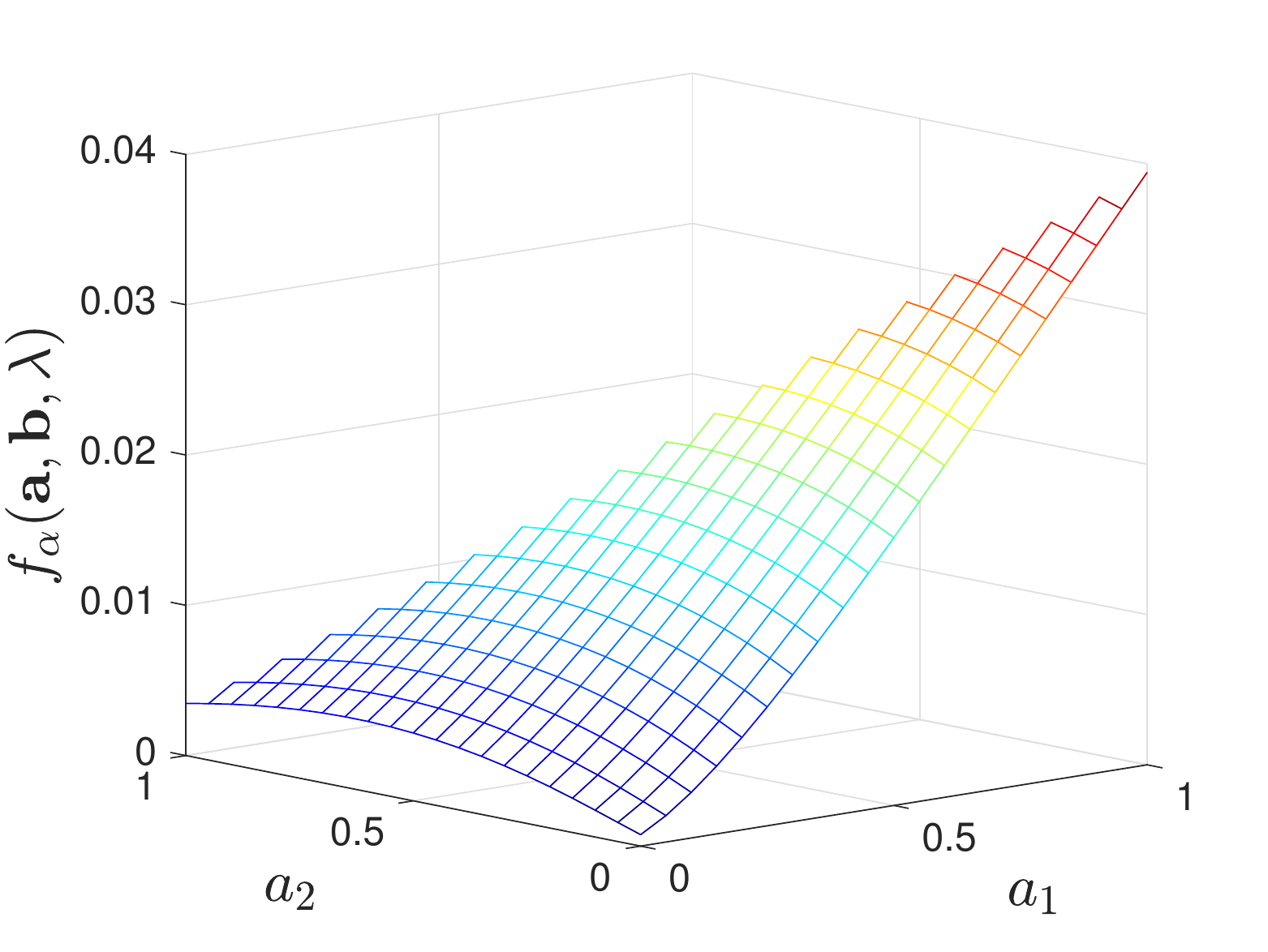}
			\label{fig:K3_(0.7,0.1,0.2)}
	\end{minipage}}
	\subfigure[$\bb=(0.1,0.5,0.4)$]{
		\begin{minipage}[t]{0.48\linewidth}
			\centering
			\includegraphics[scale=0.45]{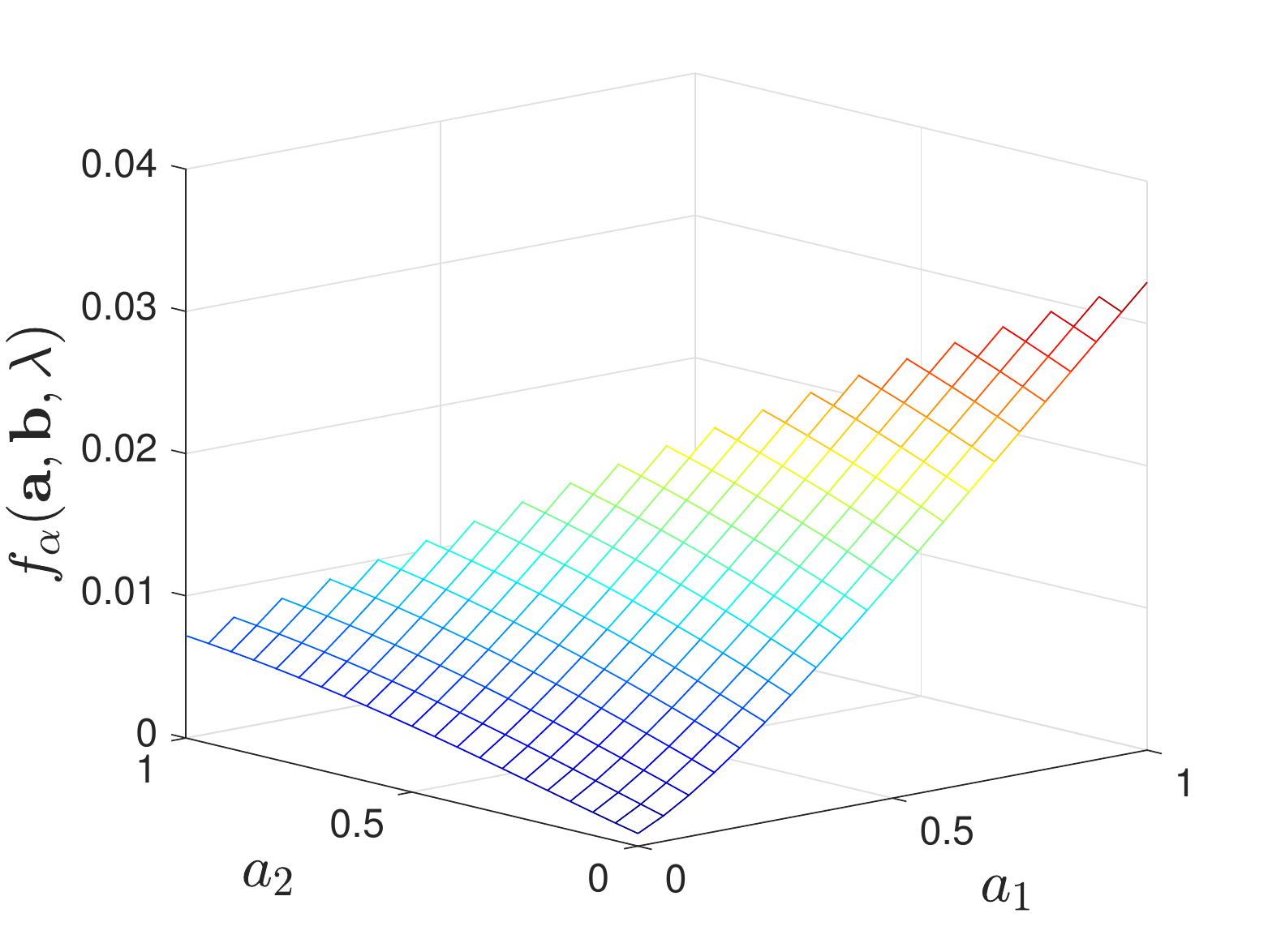}
			\label{fig:K3_(0.1,0.5,0.4)}
	\end{minipage}}
	\caption{Numerical evaluations  of $f_{\alpha}(\ba,\bb,\lambda)$ with the assumption that $V\in\calV_I$  when $K=3$, $\lambda=0.01$, $\alpha=10$ and $ W_1,W_2,W_3$ are \emph{deterministic}. Note that the maxima occur at the corner points of $[0,1]^2$, corroborating Conjecture~\ref{conj:det}. }
	\label{fig:add1}
	\hrulefill
\end{figure*}
In this subsection, we present further numerical evidence for Conjecture~\ref{conj:det} in Figs.~\ref{fig:add1} and \ref{fig:W random}. See the captions for descriptions of the figures.
\begin{figure*}[t]
	\centering
	\subfigure[$\bb=(1,0,0)$]{
		\begin{minipage}[t]{0.48\linewidth}
			\centering
			\includegraphics[scale=0.45]{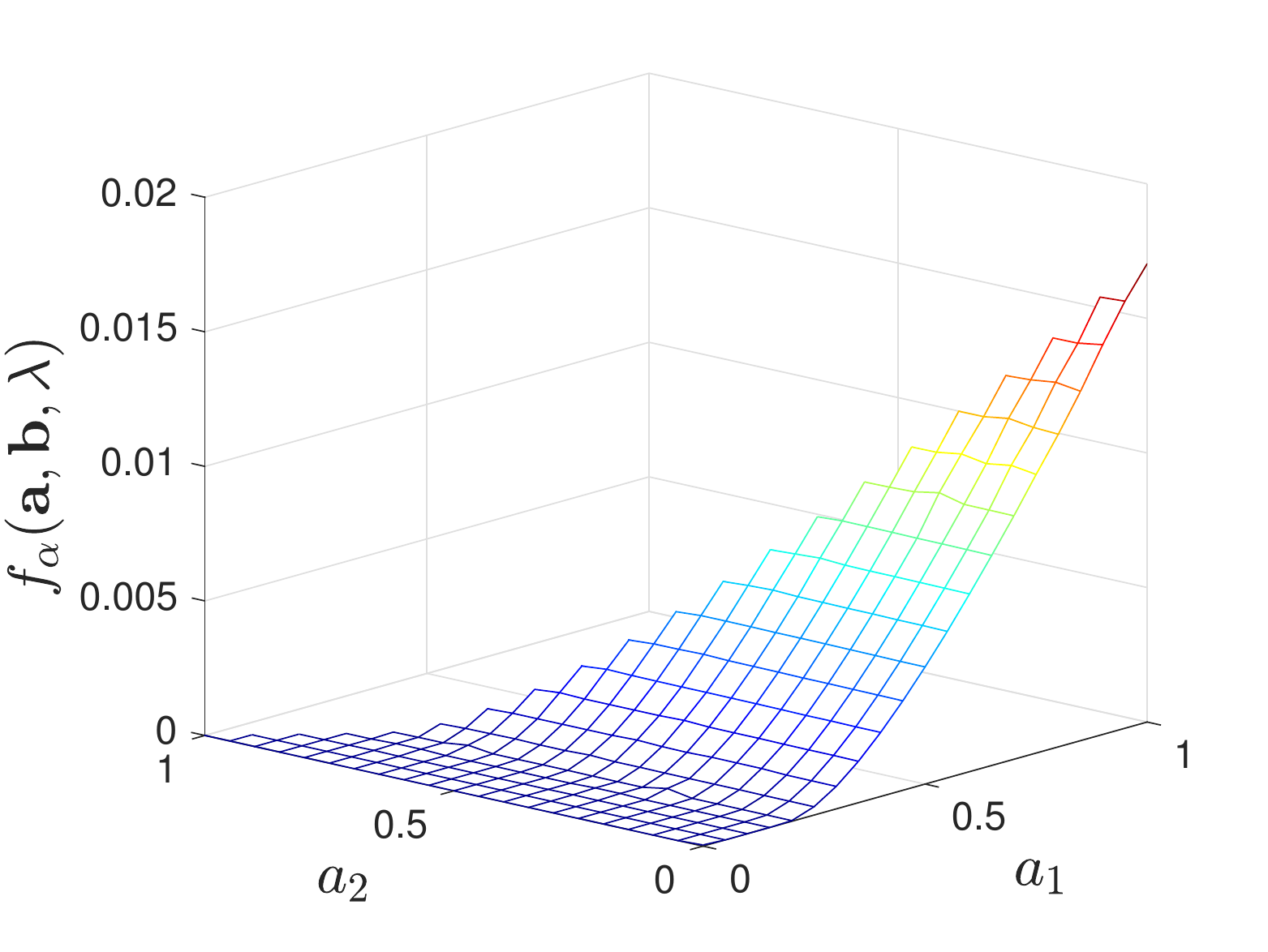}
			\label{fig:K3_Wrandom_100}
	\end{minipage}}
	\subfigure[$\bb=(0,0,1)$]{
		\begin{minipage}[t]{0.48\linewidth}
			\centering
			\includegraphics[scale=0.45]{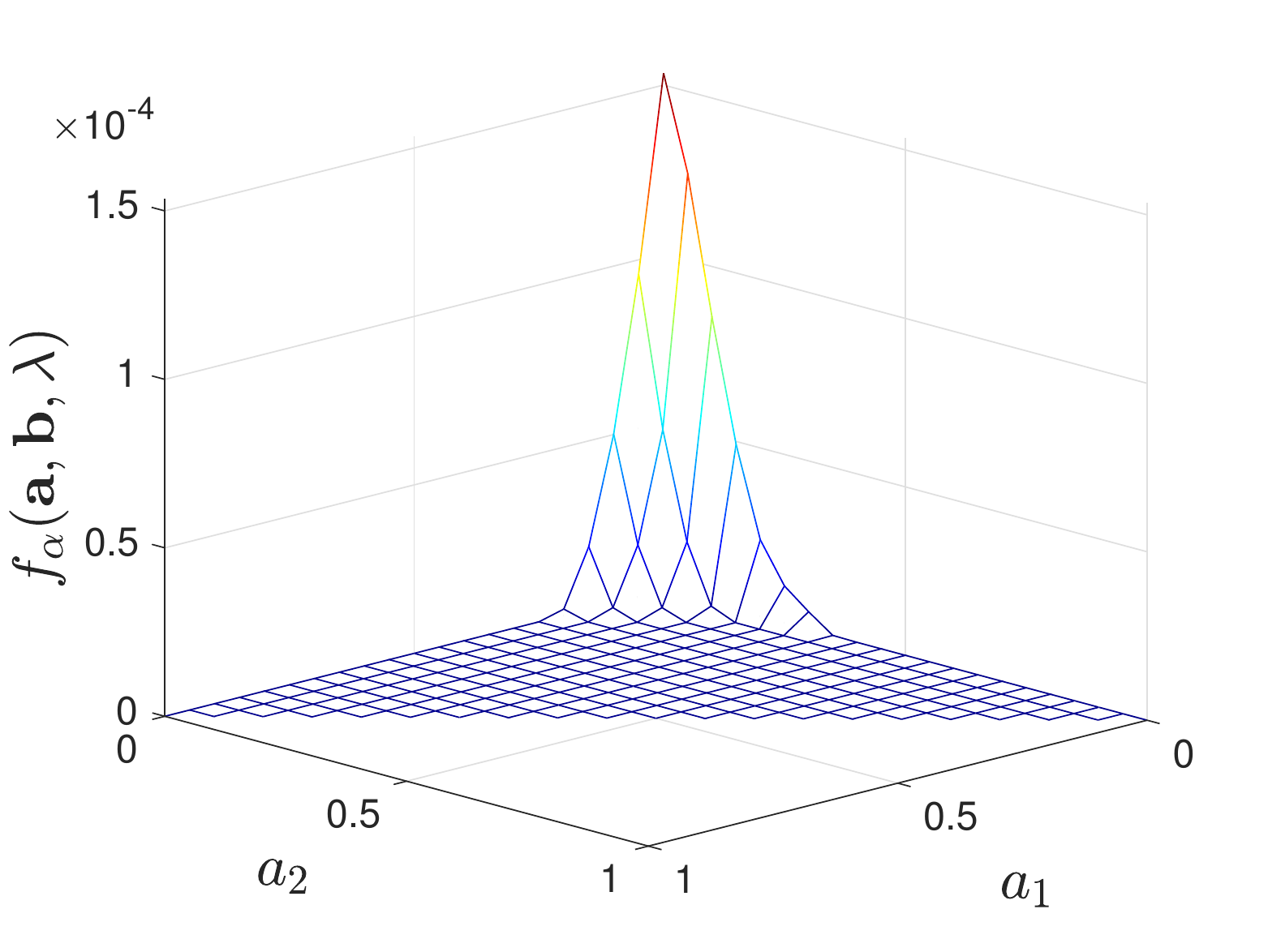}
			\label{fig:K3_Wrandom_001}
	\end{minipage}}
	\subfigure[$\bb=(\frac{1}{3},\frac{1}{3},\frac{1}{3})$]{
		\begin{minipage}[t]{0.48\linewidth}
			\centering
			\includegraphics[scale=0.45]{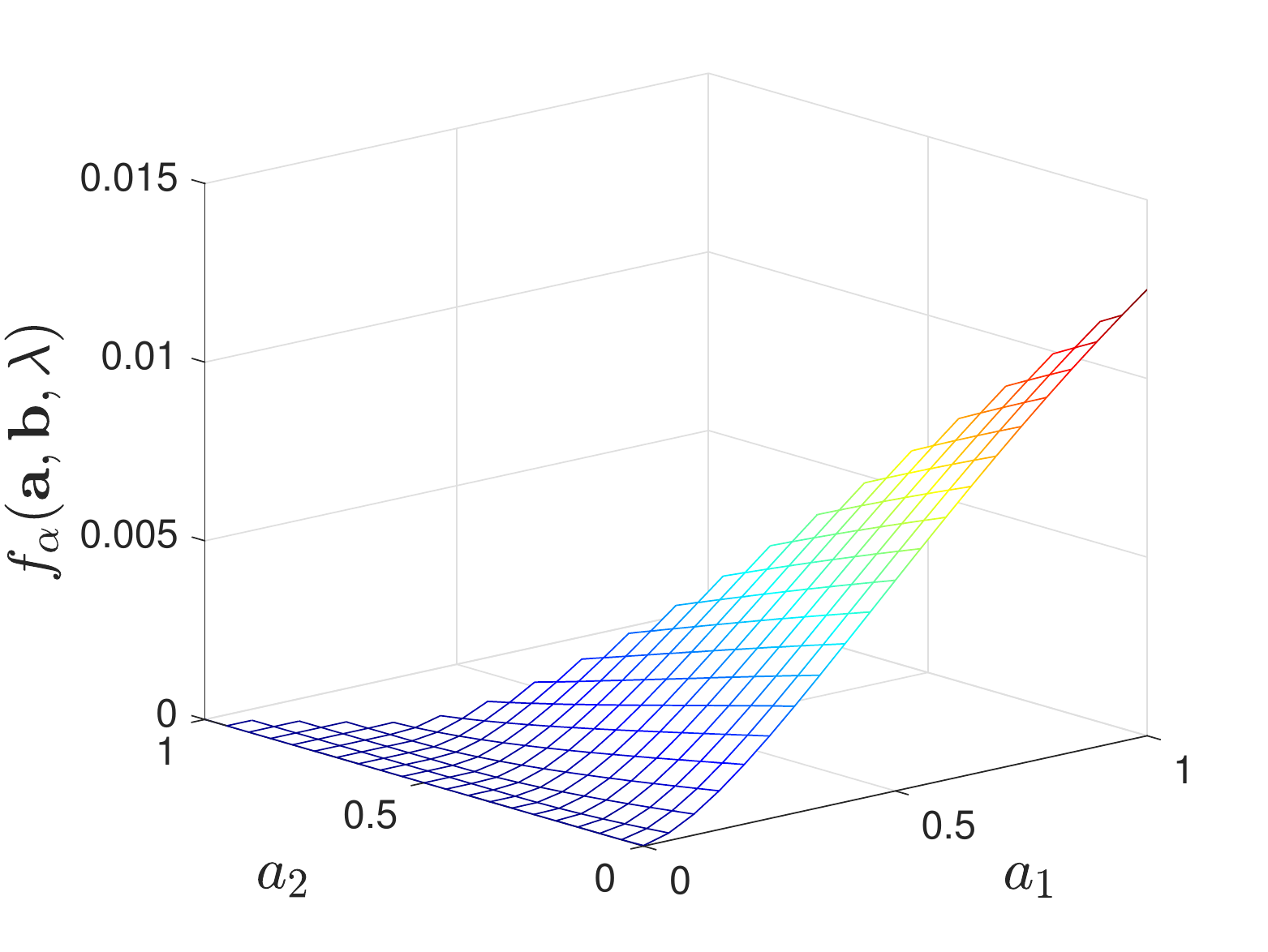}
			\label{fig:K3_Wrandom_random1}
	\end{minipage}}
	\subfigure[$\bb=(0.1,0.5,0.4)$]{
		\begin{minipage}[t]{0.48\linewidth}
			\centering
			\includegraphics[scale=0.45]{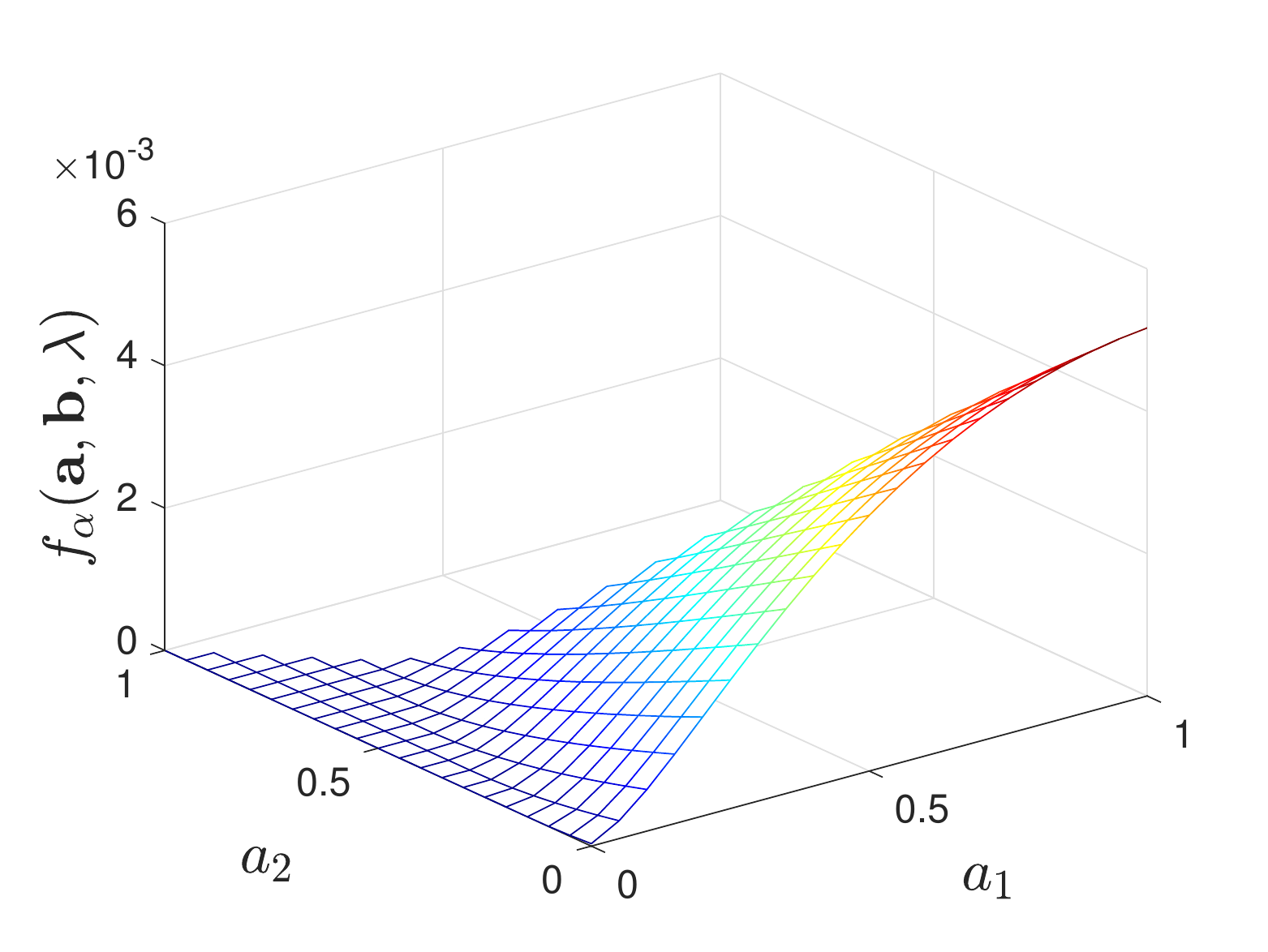}
			\label{fig:K3_Wrandom_random2}
	\end{minipage}}
	\caption{Numerical evaluations  of $f_{\alpha}(\ba,\bb,\lambda)$ with the assumption of $V\in\calV_I$  when $K=3$, $\lambda=0.01$, $\alpha=10$ and $ W_1,W_2,W_3$ are \emph{stochastic}. Note that the maxima occur at the corner points of $[0,1]^2$, corroborating Conjecture~\ref{conj:det}.}
	\label{fig:W random}
	\hrulefill
\end{figure*}

\subsection{Proof of Theorem \ref{Thm:rejection exponent}}\label{proof:rej exponent}

Recall the definitions of $i_1(\bZ^n,\tilde{\bY}^N)$ in \eqref{def:i1} and $i_2(\bZ^n,\tilde{\bY}^N)$ in \eqref{def:i2}. Given any tuple of distributions $(\bQ,\{\tilde{\bQ}_i\}_{i\in[m]})\in\calP([L])^{(m+1)K}$, we use $\mLD_j$ to denote $\mLD_j(\bQ,\{\tilde{\bQ}_i\}_{i\in[m]})$ if there is no risk of confusion.

\subsubsection{Achievability}
We use the test in \eqref{Eq:Unn test} with $\lambda$ replaced by $\tilde{\lambda}=\lambda+\frac{c_n}{n}$, where $c_n:=\sum_{k\in[K]}L(\log(na_k+1)+m\log(\alpha nb_k+1))=O(\log n)$. Given any $\bP\in\calP(\calX)^m$, for any tuple of types $(\bT_{\bZ^{n\ba}},\{\bT_{\tilde{\bY}_i^{N\bb}}\}_{i\in[m]})\in\calP_{n\ba+mN\bb}([L])$, for each $j\in[m]$, the type-$j$ error and rejection probabilities for the test in \eqref{Eq:Unn test} are respectively
\begin{align}
\beta_j(\gamma,\bP)&=\bbP_j\{i_1\neq j~\text{and}~ \mLD_{i_2}> \tilde{\lambda} \} \quad\mbox{and}\\
\rej(\gamma,\bP)&=\bbP_j\{\mLD_{i_2}\leq \tilde{\lambda} \}.
\end{align}

For each $j\in[m]$ and for all $\tilde{\bP}\in\calP(\calX)^m$, we   upper bound the type-$j$ error probability as follows:
\begin{align}
&\beta_j(\gamma,\bP)=\bbP_j\{i_1\neq j,~ \mLD_i> \tilde{\lambda},~\forall i\neq i_1  \}\\*
\label{Eq:upper bound by LDj>lambda}&\leq \bbP_j\{\mLD_j> \tilde{\lambda} \}\\
&= \sum_{\substack{(\bz^{n\ba},\{\tilde{\by}_i^{N\bb}\}_{i\in[m]}): \\
		 \mLD_j>\tilde{\lambda}}}\exp\Big\{\sum_{k\in[K]}\Big(\sum_{i:h(i)=k}\log(\tilP_jW_k)(z_i) \nn\\*
		  &\qquad\qquad+\sum_{l\in[m]}\sum_{i:g(i)=k}\log(\tilP_lW_k)(\tily_{l,i}) \Big) \Big\}\\
&\leq\sum_{\substack{(\bT_{\bz^{n\ba}},\{\bT_{\tilde{\by}_i^{N\bb}}\}_{i\in[m]}):\\
		 \mLD_j>\tilde{\lambda}}}\exp\Big\{-n\Big(\sum_{k\in[K]}a_kD(T_{z^{na_k}} \|\tilP_jW_k) \nn\\*
		 &\qquad\qquad+\sum_{l\in[m]}\sum_{k\in[K]}\alpha b_k D(T_{\tily_l^{Nb_k}}\|\tilP_lW_k) \Big) \Big\}     \\
\nn&\leq \sum_{\substack{(\bT_{\bz^{n\ba}},\{\bT_{\tilde{\by}_i^{N\bb}}\}_{i\in[m]}):\\
		\mLD_j>\tilde{\lambda}}}\!\!\!\exp\Big\{-n\Big(\min_{\tilP_j}\!\sum_{k\in[K]}\!\big(a_kD(T_{z^{na_k}} \|\tilP_jW_k) \nn\\*
		&\qquad\qquad+\alpha b_k D(T_{\tily_l^{Nb_k}}\|\tilP_jW_k)\big) \nn\\*
		&\qquad\qquad+\sum_{l\in[m]:l\neq j}\min_{\tilP_l}\sum_{k\in[K]}\alpha b_k D(T_{\tily_l^{Nb_k}}\|\tilP_lW_k) \Big) \Big\} \\
\label{Eq:follow def mLD}&\leq \exp(-n\tilde{\lambda})\prod_{k\in[K]}|\calP_{na_k}([L])||\calP_{Nb_k}([L])|^m\\
&\leq \exp\left(-n\left(\tilde{\lambda}-\frac{c_n}{n}\right)\right)=\exp(-n\lambda),
\end{align}
where \eqref{Eq:upper bound by LDj>lambda} follows since $\bigcap_{i\neq i_1, i_1\neq j }\big\{\mLD_i> \tilde{\lambda} \big\}\subset \big\{\mLD_j> \tilde{\lambda}\big\}$ and \eqref{Eq:follow def mLD} follows from the definition of $\mLD_j$ in \eqref{Eq:mary LD}.

Similarly, for each $j\in[m]$, we upper bound the type-$j$ rejection probability as  follows: %in \eqref{eqn:Pj_start}--\eqref{Eq:mary sanov} at the top of the next page.
\begin{align}
&\rej(\gamma,\bP)\nn\\*
&=\bbP_j\{\exists~(l,i)\in[m]^2:~l\neq i,~\mLD_i\leq \tilde{\lambda},~\mLD_l\leq \tilde{\lambda}\} \label{eqn:Pj_start}\\
&\leq \frac{m(m-1)}{2}\max_{(i,l)\in[m]^2:l\neq i}\bbP_j\{\mLD_i\leq \tilde{\lambda},~\mLD_l\leq \tilde{\lambda}\}\\
&\dotleq  \max_{(i,l)\in[m]^2:l\neq i} \sum_{\substack{(\bT_{\bz^{n\ba}},\{\bT_{\tilde{\by}_i^{N\bb}}\}_{i\in[m]}):\\
		\mLD_i\leq \tilde{\lambda},~\mLD_l\leq \tilde{\lambda}
} } \nn\\*
&\hspace{.6in}\exp\Big\{-n\Big(\sum_{k\in[K]}a_kD(T_{z^{na_k}} \|\tilP_jW_k) \nn\\*
&\hspace{.6in}+\sum_{i\in[m]}\sum_{k\in[K]}\alpha b_k D(T_{\tily_i^{Nb_k}}\|\tilP_iW_k) \Big) \Big\}\\
&\dotleq \exp\Big\{-n\min_{(i,l)\in[m]^2:l\neq i} \min_{\substack{(\bT_{\bz^{n\ba}},\{\bT_{\tilde{\by}_i^{N\bb}}\}_{i\in[m]}):\\
		\mLD_i\leq \tilde{\lambda},~\mLD_l\leq \tilde{\lambda}
}} \nn\\*
&\hspace{.6in}\Big(\sum_{k\in[K]}a_kD(T_{z^{na_k}} \|\tilP_jW_k) \nn\\*
&\hspace{.6in}+\sum_{i\in[m]}\sum_{k\in[K]}\alpha b_k D(T_{\tily_i^{Nb_k}}\|\tilP_iW_k) \Big) \Big\}\\
\label{Eq:mary sanov}&\leq \exp\Big\{-n\min_{(i,l)\in[m]^2:l\neq i}\min_{\substack{(\bQ,\{\tilde{\bQ}_i\}_{i\in[m]})\\
		\in\calQ_{\lambda,i,l}(\alpha,\ba,\bb,\calW)}} \nn\\*
&\hspace{.6in}\sum_{k\in[K]}\big(a_kD(Q_k\| P_jW_k) \nn\\*
		&\hspace{.6in}+\sum_{i\in[m]}\alpha b_k D(\tilQ_{i,k}\|P_i W_k) \big)\Big\}. 
\end{align}
%\hrulefill
%\end{figure*}
Using  \eqref{Eq:mary sanov} and the definition of $\mathrm{LD}_j^{[m]}(\bQ,\{\tilde{\bQ}_i\}_{i\in[m]},\bP)$ in \eqref{def:LDjm}, we arrive at  the following lower bound on the type-$j$ rejection exponent:
\begin{align}\label{Eq:m-ary direct}
&\liminf\limits_{n\to\infty}\frac{1}{n}\log\frac{1}{\rej(\gamma,\bP)} \nn\\*
&\geq \min_{(i,l)\in[m]^2:l\neq i}\min_{\substack{(\bQ,\{\tilde{\bQ}_i\}_{i\in[m]})\\
		\in\calQ_{\lambda,i,l}(\alpha,\ba,\bb,\calW)}}\mathrm{LD}_j^{[m]}(\bQ,\{\tilde{\bQ}_i\}_{i\in[m]},\bP).
\end{align}

\subsubsection{Converse}
Similar to the binary case, the converse proof proceeds by showing (i) type-based tests (i.e., tests $\gamma$ that depend only on the types or partial types of the sequences $(Z^n,\tilde{\bY}^N)$, i.e., $\bT_{\bZ^{n\ba}}$ and $\{\bT_{\tilde{\bY}_i^{N\bb}}\}_{i\in[m]}$), are almost optimal and (ii) the test in \eqref{Eq:Unn test} is an asymptotically optimal type-based test.
\begin{lemma}\label{Lem:m-ary type based optimal}
	For any test $\gamma$, $(\eta_1,\ldots, \eta_m)\in[0,1]^m$, any $(\ba,\bb)\in \calP_{n}([K])\times\calP_{\alpha n}([K])$ and any tuple of distributions $\bP\in\calP(\calX)^m$, we can construct a type-based test $\gamma^\rmT$ such that for each $j\in[m]$,
	\begin{align}
	&\beta_j(\gamma,\bP)\geq \eta_{\min}\beta_j(\gamma^\rmT,\bP),\\
	&\rej(\gamma,\bP)\geq (1-\eta_{\mathrm{sum}})\rej(\gamma^\rmT,\bP),
	\end{align}
	where $\eta_{\min}:=\min_{i\in[m]}\eta_i$ and $\eta_{\mathrm{sum}}:=\sum_{i\in[m]}\eta_i$.
\end{lemma}
The proof of Lemma \ref{Lem:m-ary type based optimal} is similar to the proof of Lemma \ref{Lem:extended from Lin} and thus omitted.

Before starting the next result, let $\delta_n=\frac{c_n}{n}$, and fix an arbitrary sequence $\{\delta_n'\}\subset (0,1)$ be such that $\lim_{n\to\infty}\delta_n'=0$.

\begin{lemma}\label{Lem:m-ary contradiction}
	Given any $(\lambda,\alpha)\in\bbR_+^2$, any $(\ba,\bb)$, for any type-based test $\gamma^{\rmT}$ satisfying that for all tuples of distributions $\tilde{\bP}\in\calP(\calX)^m$,
	\begin{equation}\label{Eq:type-based beta_j < exp(-n lambda) }
	\beta_j(\gamma^{\rmT},\tilde{\bP})\leq \exp(-n\lambda), ~\forall j\in[m],
	\end{equation}
	we have that for any $\bP\in\calP(\calX)^m$, 
	\begin{align}
	&\rej(\gamma^{\rmT},\bP)\nn\\*
	&\geq \bbP_j\Big\{\mLD_{i_2}(\bT_{\bZ^{n\ba}},\{\bT_{\tilde{\bY}_i^{N\bb}}\}_{i\in[m]})\leq \lambda-\delta_n-\delta_n' \Big\}.
	\end{align}
\end{lemma}
\begin{proof}
	The lemma is proved by showing that for any type-based test $\gamma^{\rmT}$ satisfying \eqref{Eq:type-based beta_j < exp(-n lambda) }, if any  $(\bT_{\bZ^{n\ba}},\{\bT_{\tilde{\bY}_i^{N\bb}}\}_{i\in[m]})\in \calP_{n\ba+mN\bb}([L])$ satisfies
	\begin{equation}
	\mLD_{i_2}(\bT_{\bz^{n\ba}},\{\bT_{\tilde{\by}_i^{N\bb}}\}_{i\in[m]})+\delta_n\leq \lambda-\delta_n',
	\end{equation}
	then we have $\gamma^{\rmT}(\bT_{\bz^{n\ba}},\{\bT_{\tilde{\by}_i^{N\bb}}\}_{i\in[m]})=\rmH_\rmr$.
	
	To prove this claim, it suffices to show by contradiction that there exists $(\bQ_{\bZ^{n\ba}},\{\bQ_{\tilde{\bY}_i^{N\bb}}\}_{i\in[m]})\in \calP_{n\ba+mN\bb}([L])$ such that i) $\gamma^{\rmT}(\bQ_{\bz^{n\ba}},\{\bQ_{\tilde{\by}_i^{N\bb}}\}_{i\in[m]})=\rmH_k$ for some $k\in[m]$, and ii) there exists $(l,j)\in[m]^2$ such that $l\neq j$ and
	\begin{align}
\mLD_l(\bQ_{\bz^{n\ba}},\{\bQ_{\tilde{\by}_i^{N\bb}}\}_{i\in[m]})+\delta_n &\leq \lambda-\delta_n', \\
	\mLD_j(\bQ_{\bz^{n\ba}},\{\bQ_{\tilde{\by}_i^{N\bb}}\}_{i\in[m]})+\delta_n&\leq \lambda-\delta_n'.
	\end{align}
In the following analysis, fix $j\in[m]$ such that $j\neq k$. We can then lower bound the type-$j$ error probability as follows:
\begin{align}
&\beta_{j}(\gamma^{\rmT},\bP) \nn\\*
&=\bbP_j\{\gamma^{\rmT}(\bT_{\bZ^{n\ba}},\{\bT_{\tilde{\bY}_i^{N\bb}}\}_{i\in[m]})\notin\{\rmH_j,\rmH_\rmr\} \}\\
&\geq \bbP_j\{\gamma^{\rmT}(\bT_{\bZ^{n\ba}},\{\bT_{\tilde{\bY}_i^{N\bb}}\}_{i\in[m]})=\rmH_k \}\\
&= \sum_{\substack{(\bz^{n\ba},\{\tilde{\by}_i^{N\bb}\}_{i\in[m]}):\\
		\gamma^{\rmT}(\bT_{\bz^{n\ba}},\{\bT_{\tilde{\by}_i^{N\bb}}\}_{i\in[m]})=\rmH_k}}\!\!\!\!\!\! \exp\Big\{\sum_{k\in[K]}\Big(\!\!\sum_{i:h(i)=k}\!\!\log(\tilP_jW_k)(z_i) \nn\\*
		&\qquad+\sum_{l\in[m]}\sum_{i:g(i)=k}\!\!\log(\tilP_lW_k)(\tily_{l,i}) \Big) \Big\}\\
&\geq \exp\Big\{-n\Big(\sum_{k\in[K]}a_kD(\bQ_{z^{na_k}} \|\tilP_jW_k) \nn\\*
&\qquad+\sum_{l\in[m]}\sum_{k\in[K]}\alpha b_k D(\bQ_{\tily_l^{Nb_k}}\|\tilP_lW_k)+\delta_n \Big) \Big\}.
\end{align}
If we set
\begin{align}
\tilP_j &: =\argmin_{\barP_j\in\calP(\calX)}\mathrm{LD}(\bQ_{\bz^{n\ba}},\bQ_{\tilde{\by}_j^{N\bb}}, \barP_j,\barP_j|\alpha,\ba,\bb,\calW ),\\
\tilP_l &:=\argmin_{\barP_l\in\calP(\calX)} \sum_{k\in[K]}\alpha b_k D(\bQ_{\tily_l^{Nb_k}}\|\barP_lW_k),\quad \forall\,  l\neq j,
\end{align}
  then we have
\begin{align}
\label{Eq:mary choose P}\beta_{j}(\gamma^{\rmT},\bP)&\geq \exp\big\{\!-n\big(\mLD_{j}(\bQ_{\bz^{n\ba}},\{\bQ_{\tilde{\by}_i^{N\bb}}\}_{i\in[m]})\!+\!\delta_n \big) \big\}\\*
&\geq \exp\{-n(\lambda-\delta_n') \}\\*
\label{Eq:mary contradict exp(-n lambda)}&>\exp\{-n\lambda \},
\end{align}
where \eqref{Eq:mary choose P} follows from the definition of $\mLD_j$ in \eqref{Eq:mary LD}. Thus, the inequality in \eqref{Eq:mary contradict exp(-n lambda)} contradicts the conditions in \eqref{Eq:type-based beta_j < exp(-n lambda) } and the proof of Lemma \ref{Lem:m-ary contradiction} is completed.
\end{proof}

% In the proof of Lemma \ref{Lem:m-ary contradiction}, we are able to choose the specific $\tilP_j$ and $\tilP_l$ for all $l\neq j$ to get to \eqref{Eq:mary choose P} because of the design of the linear combination of KL-divergences $\mLD_j(\bT_{\bZ^{n\ba}},\{\bT_{\tilde{\bY}_i^{N\bb}}\}_{i\in[m]})$ for all $ j\in[m]$. If we use other designs, e.g. the linear combination of KL-divergences in binary case \big(i.e.,  $\min_{\tilP\in\calP(\calX)}\mathrm{LD}\big(\bT_{\bZ^{n\ba}},\bT_{\tilde{\bY}_i^{N\bb}},\tilP,\tilP \, \big|\sim \big)$\big), we are not able to find $\tilP_j$ and $\tilP_l$ where $l\neq j$ that can result in the contradiction.

Using Lemmas \ref{Lem:m-ary type based optimal} and \ref{Lem:m-ary contradiction}, we obtain the following corollary, which provides a lower bound on the rejection probability for any test whose error probabilities decay exponentially fast under all hypotheses for all tuples of distributions.

\begin{corollary}\label{Coro:m-ary beta_R lower bound}
	Given any $(\lambda,\alpha)\in\bbR_+^2$, any $(\ba,\bb)$, for any test $\gamma$ satisfying that for all tuples of distributions $\tilde{\bP}\in\calP(\calX)^m$,
	\begin{equation}\label{Eq:any test beta_j < exp(-n lambda) }
	\beta_j(\gamma,\tilde{\bP})\leq \exp\bigg(-n\Big(\lambda+\frac{\log (2m)}{n}\Big)\bigg), \quad\forall j\in[m],
	\end{equation}
	we have that for any $\bP\in\calP(\calX)^m$, 
	\begin{align}
	&\rej(\gamma,\bP) \nn\\*
	&\geq \frac{1}{2} \bbP_j\Big\{\mLD_{i_2}(\bT_{\bZ^{n\ba}},\{\bT_{\tilde{\bY}_i^{N\bb}}\}_{i\in[m]})\!\leq\! \lambda\!-\!\delta_n\!-\!\delta_n' \Big\}.
	\end{align}
\end{corollary} 

Since 
\begin{align}
&\{\mLD_{i_2}\leq \lambda-\delta_n-\delta_n'\} \nn\\*
&=\bigcup_{l\neq i}\{\mLD_i\leq \lambda-\delta_n-\delta_n',\mLD_l\leq \lambda-\delta_n-\delta_n'\},
\end{align}
using Corollary \ref{Coro:m-ary beta_R lower bound}, we have
\begin{align}
&\rej(\gamma,\bP) \nn\\*
&\geq \frac{1}{2}\bbP_j\Big\{\bigcup_{l\neq i}\{\mLD_i\!\leq\! \lambda\!-\!\delta_n\!-\!\delta_n',\mLD_l\!\leq\! \lambda\!-\!\delta_n\!-\!\delta_n'\} \Big\}\\
&\geq \frac{1}{2}\max_{\substack{(i,l)\in[m]^2:\\ l\neq i}}\bbP_j\Big\{\mLD_i\leq \lambda\!-\!\delta_n\!-\!\delta_n',\mLD_l\leq \lambda\!-\!\delta_n\!-\!\delta_n'\Big\}\\
\nn&\geq \exp\bigg\{-n\min_{ \substack{(i,l)\in[m]^2:\\ l\neq i}}\min_{\substack{(\bT_{\bz^{n\ba}},\{\bT_{\tilde{\by}_i^{N\bb}}\}_{i\in[m]}):\\
		\mLD_i\leq \lambda-\delta_n-\delta_n',\\
		\mLD_l\leq \lambda-\delta_n-\delta_n'}}  \nn\\*
		&\quad\Big(\sum_{k\in[K]}a_kD(T_{z^{na_k}} \|\tilP_jW_k) \nn\\*
		&\quad+\sum_{i\in[m]}\sum_{k\in[K]}\alpha b_k D(T_{\tily_i^{Nb_k}}\|\tilP_iW_k)\!+\!\delta_n\!+\!\frac{\log 2}{n} \Big) \bigg\}\label{converse2use}.
\end{align}

Using \eqref{converse2use}, for each $j\in[m]$, given any tuple of distributions $\bP$, the type-$j$ rejection exponent can be upper bounded as follows
\begin{align}\label{Eq:m-ary converse}
&\limsup\limits_{n\to\infty}\frac{1}{n}\log\frac{1}{\rej(\gamma,\bP)} \nn\\*
&\leq\min_{(i,l)\in[m]^2:l\neq i}\min_{\substack{(\bQ,\{\tilde{\bQ}_i\}_{i\in[m]})\\
		\in\calQ_{\lambda,i,l}(\alpha,\ba,\bb,\calW)}}\mathrm{LD}_j^{[m]}(\bQ,\{\tilde{\bQ}_i\}_{i\in[m]},\bP).
\end{align}

\paragraph*{Acknowledgments}
The authors would like to thank Nicolas Gillis  and I-Hsiang Wang for helpful discussions. 

\bibliography{ref}
\bibliographystyle{IEEEtran}

\begin{IEEEbiographynophoto}{Haiyun He} (S'18) is currently a Ph.D.\ student in the Department of Electrical and Computer Engineering (ECE) at the National University of Singapore (NUS). From Sep  2017 to Jul  2018, she was first a Research Assistant in ECE at NUS. She received the B.E.\ degree in Beihang University (BUAA) in 2016 and the M.Sc.\ (Electrical Engineering)\ degree in ECE from NUS in 2017. Her research interests include information theory, statistical learning and their applications.
\end{IEEEbiographynophoto}

\begin{IEEEbiographynophoto}{Lin Zhou} (S'15-M'18) is currently a Research Fellow in the Department of Electrical Engineering and Computer Science at the University of Michigan, Ann Arbor. From Mar  2018 to Dec  2018, he was first a Research Engineer and then a Research Fellow in the Department of Electrical and Computer Engineering (ECE) at the National University of Singapore (NUS). He received the Ph.D.\ degree in ECE from NUS in 2018 and the B.E.\ degree in Information Engineering from Shanghai Jiao Tong University (SJTU) in 2014. His research interests include information theory, statistical inference, machine learning, physical layer security and their applications. 
\end{IEEEbiographynophoto}

\begin{IEEEbiographynophoto}{Vincent Y.\ F.\ Tan} (S'07-M'11-SM'15)  was born in Singapore in 1981. He is currently a Dean's Chair Associate Professor in the Department of Electrical and Computer Engineering  and the Department of Mathematics at the National University of Singapore (NUS). He received the B.A.\ and M.Eng.\ degrees in Electrical and Information Sciences from Cambridge University in 2005 and the Ph.D.\ degree in Electrical Engineering and Computer Science (EECS) from the Massachusetts Institute of Technology (MIT)  in 2011.  His research interests include information theory, machine learning, and statistical signal processing.

Dr.\ Tan received the MIT EECS Jin-Au Kong outstanding doctoral thesis prize in 2011, the NUS Young Investigator Award in 2014,  the Singapore National Research Foundation (NRF) Fellowship (Class of 2018) and the NUS Young Researcher Award in 2019. He was also an IEEE Information Theory Society Distinguished Lecturer during 2018/9. He is currently serving as an Associate Editor of the {\em IEEE Transactions on Signal Processing} and an Associate Editor of Machine Learning for the {\em IEEE Transactions on Information Theory}.
\end{IEEEbiographynophoto}
\end{document}